\documentclass[letterpaper,11pt]{article}

\usepackage[T1]{fontenc}

\usepackage{amsmath, amssymb, amsthm}

\usepackage{authblk}

\usepackage{graphicx}
\graphicspath{{figures/}}
\usepackage{pgfplots}

\usepackage{caption}
\usepackage{enumitem}
\usepackage{mdframed}

\usepackage{subcaption}
\usepackage{soul}
\usepackage{lineno}

\usepackage{todonotes}
\usepackage{nicefrac}

\usepackage{thm-restate}

\usepackage{mathrsfs}

\usepackage[procnumbered, linesnumbered,algoruled,noend]{algorithm2e}

\usepackage[hypertexnames=false,hidelinks]{hyperref}
\hypersetup{
	colorlinks   = true, %
	urlcolor     = blue, %
	linkcolor    = blue, %
	citecolor   = red %
}

\usepackage[capitalize]{cleveref}

\usepackage{dsfont}

\usepackage{xspace}

\usepackage{mathtools}
\usepackage{bbm}
\usepackage{amssymb}

\newenvironment{claimproof}{\begin{proof}}{\end{proof}}

\crefname{claim}{Claim}{Claims}
\Crefname{claim}{Claim}{Claims}

\crefname{algorithm}{Algorithm}{Algorithms}
\Crefname{algorithm}{Algorithm}{Algorithms}
\crefname{ALC@unique}{Line}{Lines}
\Crefname{ALC@unique}{Line}{Lines}

\newcommand{\LP}{\textnormal{LP}}
\newlist{properties}{enumerate}{1}
\setlist[properties,1]{label={\textbf{P\arabic*}}, align=left, left=0pt, itemindent=*}

\crefname{propertiesi}{property}{properties}

\newcommand{\cI}{\mathcal{I}}

\newcommand{\floor}[1]{ {\left\lfloor #1 \right\rfloor}}
\newcommand{\ceil}[1]{ {\left\lceil #1 \right\rceil}}
\newcommand{\abs}[1]{{\left| #1\right|}}
\newcommand{\weight}{\bar{w}}
\newcommand{\capacity}{\bar{W}}

\DeclareMathOperator*{\argmax}{\textnormal{argmax}}
\DeclareMathOperator*{\argmin}{\textnormal{argmin}}

\DeclareMathOperator{\sign}{sign}

\newcommand{\semi}{\mathpunct{:}\!}

\newcommand{\eps}{\varepsilon}

\newcommand{\simpleddim}{\textnormal{\texttt{dVK-Simple}}}
\newcommand{\mitmddim}{\textnormal{\texttt{dKnapsack-MitM}}}

\newcommand{\linear}{\textnormal{\texttt{Linear2D}}}

\newcommand{\embed}{\textnormal{\texttt{embed}}}
\newcommand{\range}{\textnormal{\texttt{range}}}
\newcommand{\best}{\textnormal{\texttt{best}}}
\newcommand{\set}{\textnormal{\texttt{set}}}
\newcommand{\vitem}{\textnormal{\texttt{item}}}
\newcommand{\vright}{\textnormal{\texttt{right}}}
\newcommand{\vleft}{\textnormal{\texttt{left}}}

\newcommand{\lowprofit}{\textnormal{\texttt{RepresentativeSolutions}}}

\newcommand{\tS}{\tilde{S}}

\newcommand{\cC}{\mathcal{C}}

\newcommand{\OPT}{\textnormal{\texttt{OPT}}}

\newcommand{\Oh}{O}

\newcommand{\comment}[1]{}

\newtheorem{theorem}{Theorem}[section]
\newtheorem{lemma}[theorem]{Lemma}
\newtheorem{definition}[theorem]{Definition}
\newtheorem{claim}[theorem]{Claim}

\newtheorem{corollary}[theorem]{Corollary}
\newtheorem{observation}[theorem]{Observation}

\newcommand{\cL}{\mathcal{L}}

\newcommand{\tp}{{\tilde{p}}}
\newcommand{\hp}{\widehat{p}}

\newcommand{\veca}{\ensuremath{\boldsymbol{a}}}

\newcommand{\vecs}{\ensuremath{\boldsymbol{s}}}

\newcommand{\vecx}{\ensuremath{\boldsymbol{x}}}
\newcommand{\vecy}{\ensuremath{\boldsymbol{y}}}
\newcommand{\vecz}{\ensuremath{\boldsymbol{z}}}

\usepackage[margin=1in]{geometry}
\usepackage{amsmath,amssymb}

\usepackage{pgfplots}
\usetikzlibrary{arrows.meta}
\newcommand{\tOh}{\widetilde O}

\title{
	 Fine-Grained Complexity of Approximating Vector Knapsack: \\ A Faster Algorithm and Equivalence with $k$-SUM in 2D}
 \author{
     Karl Bringmann\footnote{ETH Zurich, Switzerland.},\; 
     Ariel Kulik\footnote{Ben-Gurion University of the Negev, Beer-Sheva, Israel.},\; and 
     Karol W\k{e}grzycki\footnote{Max Planck Institute for Informatics, Saarbrücken, Germany.
	 Supported by 
    the Deutsche Forschungsgemeinschaft (DFG, German Research Foundation) grant number
    559177164.} 
 }

\date{}

\begin{document}

\maketitle

\thispagestyle{empty}
\begin{abstract}
	
We revisit the $d$-dimensional Vector Knapsack problem (abbreviated as $d$-Knapsack): Given a $d$-dimensional capacity vector and a set of items, each with a $d$-dimensional weight vector and a profit, the goal is to select a set of items that maximizes the total profit without exceeding the capacity in any dimension.
For any $d \ge 2$, the best known approximation scheme for $d$-Knapsack runs in time $O(n^{\lceil d/\eps \rceil - d})$ [Caprara, Kellerer, Pferschy, Pisinger '00]. 

We improve this running time to $n^{\frac{d-1}{2\eps} + \Oh(d)}$. 
This is the first improvement for $d$-Knapsack in over 25 years, and the first result that improves the exponent by a constant factor, as all prior algorithms had an exponent of $d/\eps \pm O(d)$. 
We achieve this speedup by designing the first meet-in-the-middle algorithm for $d$-Knapsack. 
This requires replacing the LP solver used in prior algorithms by a dynamic programming algorithm to generate representative solutions, building on an LP-based structural argument.

We further establish a fine-grained equivalence of 2-Knapsack and $k$-SUM: For $k = \lceil 1/\eps \rceil - 1$, computing a $(1-\eps)$-approximation for 2-Knapsack has the same time complexity as the $k$-SUM problem up to a factor $n^{\Oh(1)}$, where the $\Oh(1)$ is independent of $\eps$.
Consequently, unless $k$-SUM admits substantially faster algorithms than are currently known, 2-Knapsack requires time $n^{\frac 1{2\eps} - \Oh(1)}$, which together with our algorithm establishes the optimal running time exponent of 2-Knapsack as $\frac 1{2\eps} \pm \Oh(1)$.
Conversely, if significantly faster algorithms for $k$-SUM are found, they automatically transfer to 2-Knapsack. 

%We complement this by a fine-grained lower bound based on $k$-SUM showing that 2-Knapsack requires time $n^{\lceil \frac{1}{2\eps} - \frac 12 \rceil - o(1)}$. This establishes the optimal exponent of 2-Knapsack as $\frac 1{2\eps} \pm O(1)$, which is precise up to an \emph{additive} $O(1)$, while previously it was only known up to a \emph{factor} $O(1)$. To the best of our knowledge, this is the first result that determines the optimal exponent more precisely than up to a factor $O(1)$, for \emph{any} problem that admits a PTAS but no EPTAS.
%
%We also design a $(2/3 - o(1))$-approximation for 2-Knapsack in near-linear time. Combining our two algorithms yields a $(1-\eps-o(1))$-approximation for 2-Knapsack with running time $\tOh_\eps(n^{\lceil \frac{1}{2\eps} - \frac 12 \rceil})$. This nearly matches our lower bound, because for a slightly better approximation ratio a slightly better running time is impossible -- so our algorithm is \emph{bicriteria-optimal}. In particular, we establish the optimal exponent of 2-Knapsack as $\lceil \frac{1}{2\eps} - \frac 12 \pm o(1) \rceil \pm o(1)$. 

\end{abstract}

\clearpage
\setcounter{page}{1}

\section{Introduction}
\label{sec:intro}

We revisit the $d$-dimensional Vector Knapsack problem, which for short we call \emph{$d$-Knapsack}. The input consists of a $d$-dimensional vector of capacities $\capacity \in \mathbb{R}_{\ge 0}^d$ and a set $I$ of $n$ items. Each item~$i \in I$ has an associated profit $p(i) \in \mathbb{R}_{\ge 0}$ and a $d$-dimensional vector of weights $\weight(i) \in \mathbb{R}_{\ge 0}^d$. The goal is to maximize the total profit $p(S) \coloneqq \sum_{i \in S} p(i)$ over all subsets $S \subseteq I$ such that the total weight respects the capacities, i.e., $\weight_t(S) \coloneqq \sum_{i \in S} \weight_t(i) \le \capacity_t$ for all $t \in \{1,\ldots,d\}$.
The $d$-Knapsack problem is the same as \emph{packing integer programs with $\{0,1\}$-variables}; the name $d$-Knapsack stresses that the number of non-trivial constraints $d$ is considered to be small.

Note that the standard Knapsack problem is the same as 1-Knapsack. The $d$-Knapsack problem is a natural generalization that incorporates multiple constraints, which frequently arises in applications that optimize several types of cost measures (see, e.g.,~\cite{ZhouHWJTZW25,DuLZGLSZXZ24,PsychasG18} for some recent examples). For this reason, $d$-Knapsack has been studied intensely, with initial work dating back to the 50s~\cite{LorieS1955,MarkowitzM1957,WeingartnerN67}. We refer to the survey~\cite{Freville04} for an overview of the $d$-Knapsack literature on exact algorithms (using techniques such as dynamic programming, branch and bound, LP bounds, Lagrangian relaxation, and preprocessing) and on heuristics (employing greedy, dual greedy, local search, LP-based search, simulated annealing, tabu search, etc.).

\paragraph{Approximation Schemes}
$d$-Knapsack admits a polynomial-time approximation scheme (PTAS), i.e., a $(1-\eps)$-approximation algorithm for any $\eps > 0$. The first PTAS, designed by Chandra, Hirschberg, and Wong~\cite{ChandraHW76}, achieved running time $O(n^{\lceil d/\eps \rceil})$. Strictly speaking, their algorithm solves an unbounded problem variant where each item can be selected multiple times, but it easily extends to the standard variant of $d$-Knapsack~\cite{OguzM80}. The running time was improved to $\tilde O(n^{\lceil d/\eps \rceil - d + 5})$ by Frieze and Clarke~\cite{FriezeC84}, and further to $O(n^{\lceil d/\eps \rceil - d})$ by Caprara, Kellerer, Pferschy, and Pisinger~\cite{CapraraKPP00}. These algorithms all follow a common approach, as we will discuss in \Cref{sec:techoverview}.

\paragraph{Lower Bounds}
Korte and Schrader~\cite{KorteS81} and Magazine and Chern~\cite{MagazineC84} showed that 2-Knapsack has no fully polynomial-time approximation scheme (FPTAS), i.e., no approximation scheme that runs in time $\textup{poly}(n,1/\eps)$, assuming $\textsf{P} \ne \textsf{NP}$. This follows from a reduction from Subset Sum.

Kulik and Shachnai~\cite{KulikS10} found a reduction from $k$-SUM, and used it to show that 2-Knapsack has no efficient polynomial-time approximation scheme (EPTAS), i.e., no approximation scheme that runs in time $f(1/\eps) n^{O(1)}$, assuming $\textsf{FPT} \ne \textsf{W}[1]$. They also used the same reduction to show that 2-Knapsack has no approximation scheme that runs in time $f(1/\eps) n^{o(1/\sqrt{\eps})}$, assuming a different complexity-theoretic conjecture.
Later, Jansen, Land, and Land~\cite{JansenLL16} used Kulik and Shachnai's reduction to show that 2-Knapsack requires time $n^{\Omega(1/\eps)}$, assuming the Exponential Time Hypothesis (ETH)\footnote{Two AI companies recently claimed a refutation of ETH: \url{https://hexagonmath.org/2610.00577} and \url{https://github.com/openai/math/raw/main/preprints/The-Exponential-Time-Hypothesis-is-False-October-25-2026/paper.pdf}. This does not affect any of the results in this paper, but some related work, like this one. \label{footnote:ETH}}.

For $d$-Knapsack, Doron-Arad, Kulik, and Manurangsi~\cite{DoronAradKM26} recently proved that any approximation scheme requires time $n^{\tilde \Omega(d/\eps)}$, where the $\tilde \Omega$ hides logarithmic factors in $d/\eps$, assuming ETH\textsuperscript{\ref{footnote:ETH}}.

\subsection{Our Results}

We present an improved algorithm for $d$-Knapsack, and an equivalence of 2-Knapsack and $k$-SUM. In what follows we discuss our results in detail.
See \Cref{fig:exponentplots} for illustrations.

\paragraph{Improved Algorithm for \boldmath$d$-Knapsack}
As our main result, we design an improved approximation scheme for $d$-Knapsack.

\begin{restatable}{theorem}{thmMainAlgo} \label{thm:dalgo}
	Given any $d \ge 2$ and $\eps,\rho \in (0,1)$, $d$-Knapsack has a $(1-\eps)$-approximation algorithm with running time $\big(n^{\lceil \frac{d-1}{2\eps} - \frac 12 + \rho \rceil}+n^{d}\big) \cdot \big( \frac{d \log n}{\rho\cdot \eps }\big)^{O(d)}$.  
\end{restatable}

This running time can be bounded by $n^{\frac{d-1}{2\eps} + \Oh(d)}$.\footnote{Indeed, if $d \gg n$ or $1/\eps \gg n$ then brute forcing the problem in time $2^{\Oh(n)}$ is bounded by $n^{\frac{d-1}{2\eps} + \Oh(d)}$. Otherwise, by setting $\rho \coloneqq 1/2$ we have $\big( \frac{d \log n}{\rho\cdot \eps }\big)^{O(d)} \le n^{\Oh(d)}$, so the bound follows.}
For constant $d, \eps$, and $\rho \coloneqq 1/2$, the running time is $\tOh(n^{\lceil (d-1)/(2\eps) \rceil} + n^d)$.
Our result improves the previous time $O(n^{\lceil d/\eps \rceil - d})$~\cite{CapraraKPP00} for every $\eps \in (0,1/2)$.\footnote{Indeed, for every $\eps, \rho \in (0,1/2)$ we have $\frac d \eps - d > d$ and thus $\lceil \frac d \eps \rceil - d > d$, and we have $\frac d \eps - d > \frac{d-1}{2\eps} - \frac 12 + \rho + 1$ and thus $\lceil \frac d \eps \rceil - d > \lceil \frac{d-1}{2\eps} - \frac 12 + \rho \rceil$.}
This is the first algorithmic improvement for $d$-Knapsack in over 25 years.
Since all prior approximation schemes had exponent $d/\eps \pm O(d)$, for small~$\eps$ our result is \emph{the first improvement of the exponent by a constant factor}. 

We provide an overview of our algorithm in \Cref{sec:techoverview} and the details in \Cref{sec:ddim_improved}.

\paragraph{Near-Linear-Time Approximation for 2-Knapsack}
The running time of \Cref{thm:dalgo} is never below $\Theta(n^d)$. For 2-Knapsack we bridge this gap by providing a near-$\frac 23$-approximation algorithm with near-linear running time, see \Cref{sec:linear2d} for details.

\begin{restatable}{theorem}{thm2DLinear} \label{thm:linear}
	Given any $\delta \in \big(0,\frac 23\big)$, 2-Knapsack has a $\big(\frac 23 - \delta\big)$-approximation algorithm with running time $n \cdot \big( \frac{\log n}{\delta }\big)^{O(1)}$.  
\end{restatable}

Combining the above theorems yields the following time bound for 2-Knapsack, see \cref{sec:proof_cor}.

\begin{corollary} \label{cor:2Dtime}
	Fix $\eps \in (0,1)$. If $1/\eps$ is not an odd integer, then a $(1-\eps)$-approximation for 2-Knapsack can be computed in time $\tOh(n^{\lceil \frac 1{2\eps} - \frac 12 \rceil}) = \tOh(n^{\lceil (\lceil \frac 1 \eps \rceil - 1 ) / 2 \rceil})$. Otherwise, such an approximation can be computed in time $\tOh(n^{\lceil \frac 1{2\eps} - \frac 12 \rceil + 1})$.
\end{corollary}

This improves the previous time $O(n^{\lceil 2/\eps \rceil - 2})$~\cite{CapraraKPP00} for 2-Knapsack for every $\eps \in (0,2/3)$. For $\eps \ge 2/3$, Caprara et al.~\cite{CapraraKPP00} already achieved the optimal running time $O(n)$. 

\paragraph{Reduction from \boldmath$k$-SUM to 2-Knapsack}
We observe that the known reduction from $k$-SUM to 2-Knapsack~\cite{KulikS10} shows the following relation, see \Cref{sec:techoverview}.

\begin{theorem} \label{thm:2lower}
  Fix $\eps \in (0,1)$ and $k = \lceil 1/\eps \rceil - 1$. For any $\alpha \ge 1$, if $(1-\eps)$-approximate 2-Knapsack can be solved in time $\Oh(n^\alpha)$ then $k$-SUM can be solved in time $\Oh(n^\alpha)$.
\end{theorem}

Note that this also applies to $d$-Knapsack for any $d>2$, by a trivial embedding of 2-Knapsack.

In the $k$-SUM problem, given a set of $n$ numbers $A \subseteq \mathbb{R}$, the task is to decide whether there exist distinct $a_1,\ldots,a_k \in A$ with $a_1+\ldots+a_k = 0$. We remark that in the literature this formulation of $k$-SUM is sometimes called Real-$k$-SUM; in this paper it is the most natural formulation, as we work on the Real RAM. For convenience, we define this problem for any $k \ge 1$. 
While the naive running time for $k$-SUM is $O(n^k)$, with meet in the middle the problem can be solved in time $\tOh(n^{\lceil k/2 \rceil})$. For any $k \ge 3$, the $k$-SUM Hypothesis~\cite{AbboudL13} postulates that this running time is essentially optimal, i.e., $k$-SUM cannot be solved in time $O(n^{\lceil k/2 \rceil - \delta})$ for any $\delta > 0$, see, e.g.,~\cite{Erickson99,AbboudL13,AbboudBBK17,AbboudBBK20,
Kunnemann22,GokajK25} for uses in fine-grained complexity. For every odd $k \ge 3$ this hypothesis was recently refuted~\cite{AlmanVW26,ChanIMY26}.\footnote{These algorithms work on polynomially bounded integers. For real numbers, a subquadratic algorithm for 3-SUM follows from combining \cite{AlmanVW26} with the reduction in~\cite{FischerPS26}. For $k$-SUM for $k \ge 4$ such a reduction is not known.} The obtained improvements are small (the factor is less than $n^{0.1}$, in fact much smaller on real numbers).

Observe that plugging \cref{cor:2Dtime} into \cref{thm:2lower} matches the meet-in-the-middle running time for $k$-SUM, whenever $1/\eps$ is not an odd integer. It follows that improving upon \cref{cor:2Dtime} by a factor, say, $n^{0.1}$ for some $\eps \in (0,1/3)$ where $1/\eps$ is not an odd integer would improve the state of the art for $k$-SUM for some $k$. 

%\karl{remove the following paragraph?}
%In particular, the following hypothesis for $k$-SUM is still plausible:
%\begin{center}
%	New $k$-SUM Hypothesis: $k$-SUM cannot be solved in time $\Oh(n^{k/2 - \delta})$ for any $\delta > 0$.
%\end{center}
%Assuming this hypothesis, \cref{thm:2lower} shows that 2-Knapsack is not in time $\Oh(n^{(\lceil 1/\eps \rceil - 1)/2 - \delta})$ for any $\delta > 0$, so by \cref{cor:2Dtime} the optimal running time exponent of 2-Knapsack is $\frac 1{2\eps} \pm \Oh(1)$.
%This conditionally determines the optimal exponent up to an \emph{additive} constant $\pm \Oh(1)$, which is surprising as we are not aware of any prior results that conditionally determine the optimal running time exponent of a problem that admits a PTAS but no EPTAS more precisely than up to a constant factor. \karl{remove until here?}

Note that our results (nearly and conditionally) determine the best approximation for 2-Knapsack that can be computed in near-linear time: a $(\frac 23 - \delta)$-approximation can be computed in time $\tOh(n)$ by \Cref{thm:linear}, and no $(\frac 23 + \delta)$-approximation can be computed in time $\tOh(n)$ by \Cref{thm:2lower} unless 3-SUM can be solved in time $\tOh(n)$.

\paragraph{Reduction from 2-Knapsack to \boldmath$k$-SUM}
Finally, we establish a converse reduction.

\begin{theorem} \label{thm:2reductiontokSUM}
  There exists an absolute constant $c \ge 0$ such that the following holds. Let $\eps \in (0,1)$ and $k = \lceil 1/\eps \rceil - 1$. For any $\alpha \ge 1$, if $k$-SUM can be solved in time $\Oh(n^\alpha)$ then $(1-\eps)$-approximate 2-Knapsack can be solved in time $\Oh(n^{\alpha + c})$.
  
  More generally, fix $d \ge 2$, $\eps \in (0,1)$, and $k = \lceil \frac{d-1}\eps \rceil - 1$. For any $\alpha \ge 1$, if $k$-SUM can be solved in time $\Oh(n^\alpha)$ then $(1-\eps)$-approximate $d$-Knapsack can be solved in time $\Oh(n^{\alpha + c \cdot d})$.
\end{theorem}

The reductions in \cref{thm:2lower,thm:2reductiontokSUM} show that the running times of 2-Knapsack and $k$-SUM for $k = \lceil 1/\eps \rceil - 1$ are the same up to a polynomial factor $\Oh(n^c)$, where $c$ is independent of $\eps$. 
Hence, if in the future significantly faster algorithms for $k$-SUM are found, then this automatically yields improvements for 2-Knapsack. More generally, up to polynomial factors the question for the optimal time complexity of 2-Knapsack is now replaced by the question for the optimal time complexity of $k$-SUM. In particular, conditional on knowing the optimal running time exponent of $k$-SUM, we determined the optimal running time exponent of 2-Knapsack to be the same up to $\pm \Oh(1)$.

\begin{figure}[t]
\centering

\begin{subfigure}[t]{0.48\textwidth}
\centering

\begin{tikzpicture}
\begin{axis}[
  width=\linewidth,
  height=6.3cm,
  xmin=0, xmax=1,
  ymin=1, ymax=7,
  % axis lines without built-in arrow heads
  axis x line*=bottom,
  axis y line*=left,
  axis line style={line width=1pt},
  tick align=outside,
  tick style={line width=1pt},
  major tick length=2.5pt,
  xlabel={},
  ylabel={},
  ytick={1,...,7},
  xtick={0,1/7,1/5,1/3,2/5,1/2,2/3,1},
  xticklabels={
    $0$,
    $\frac{1}{7}$,
    $\frac{1}{5}$,
    $\frac{1}{3}$,
    $\frac{2}{5}$,
    $\frac{1}{2}$,
    $\frac{2}{3}$,
    $1$
  },
  xticklabel style={font=\normalsize, rotate=0, anchor=north},
  yticklabel style={font=\normalsize},
  ymajorgrids=true,
  grid style={gray!20},
  clip=false,
  after end axis/.code={
    % x-axis arrow whose tip ends exactly at x=1
    \draw[-{Stealth[length=3pt,width=4pt]}, line width=1pt]
      (rel axis cs:0.96,0) -- (rel axis cs:1,0);
    % y-axis arrow extending beyond the top
    \draw[-{Stealth[length=3pt,width=4pt]}, line width=1pt]
      (rel axis cs:0,1) -- (rel axis cs:0,1.08);
    % manual axis labels
    \node[font=\normalsize, anchor=west]
      at (rel axis cs:1.02,0.) {$\varepsilon$};
    \node[font=\normalsize, anchor=south]
      at (rel axis cs:0.11,1.085) {time exponent};
    \node[font=\normalsize, anchor=south]
      at (rel axis cs:0.93,1.085) {$d=2$};
  },
]

% ------------------------------------------------------------
% Red region:
% below ceil(1/(2 eps) - 1/2) - 0.001, clipped to y<=7
% ------------------------------------------------------------
\addplot[
  draw=none,
  fill=red!30,
  fill opacity=0.8
] coordinates {
  (0,7)
  ({1/15},7)     ({1/15},6.999)
  ({1/13},6.999) ({1/13},5.999)
  ({1/11},5.999) ({1/11},4.999)
  ({1/9},4.999)  ({1/9},3.999)
  ({1/7},3.999)  ({1/7},2.999)
  ({1/5},2.999)  ({1/5},1.999)
  ({1/3},1.999)  ({1/3},1)
  (1,1)
  (0,1)
} \closedcycle;

% ------------------------------------------------------------
% Green staircase:
% ceil(1/(2 eps) - 1/2 + 0.001)
% ------------------------------------------------------------
\def\delta{0.001}

% solid part above y=2
\foreach \k in {3,...,7}{
  \pgfmathsetmacro{\xl}{1/(2*\k+1-2*\delta)}
  \pgfmathsetmacro{\xr}{1/(2*\k-1-2*\delta)}
  \addplot[green!60!black, very thick] coordinates {
    (\xl,\k) (\xr,\k)
  };
}

\foreach \k in {3,...,7}{
  \pgfmathsetmacro{\xx}{1/(2*\k-1-2*\delta)}
  \pgfmathtruncatemacro{\kmone}{\k-1}
  \addplot[green!60!black, very thick] coordinates {
    (\xx,\kmone) (\xx,\k)
  };
}

% solid piece at y=2 up to x=1/3
\pgfmathsetmacro{\xltwo}{1/(5-2*\delta)}
\addplot[green!60!black, very thick] coordinates {
  (\xltwo,2) ({1/3},2)
};

% dotted part from (1/3,2) via (1/3,1) to (1,1)
\addplot[green!60!black, very thick, dashed] coordinates {
  ({1/3},2) ({1/3},1) (1,1)
};

% ------------------------------------------------------------
% Blue staircase:
% ceil(2/eps) - 2
% ------------------------------------------------------------
\foreach \b in {1,...,7}{
  \pgfmathsetmacro{\xl}{2/(\b+2)}
  \pgfmathsetmacro{\xr}{min(1,2/(\b+1))}
  \addplot[blue, very thick, dotted] coordinates {
    (\xl,\b) (\xr,\b)
  };
}

\foreach \b in {2,...,7}{
  \pgfmathsetmacro{\xx}{2/(\b+1)}
  \pgfmathtruncatemacro{\bmone}{\b-1}
  \addplot[blue, very thick, dotted] coordinates {
    (\xx,\bmone) (\xx,\b)
  };
}

% ------------------------------------------------------------
% Annotations
% ------------------------------------------------------------
\node[font=\normalsize] at (axis cs:{1/6},{3/2})
  {Thm.\ \ref{thm:2lower}};

\node[font=\normalsize, align=center] at (axis cs:{1/4},3.05)
  {Thm.\\ \ref{thm:dalgo}};

\node[font=\normalsize] at (axis cs:{0.48},{1.5})
  {Thm.\ \ref{thm:linear}};

\node[font=\normalsize] at (axis cs:0.58,4.43)
  {Caprara et al.\ \cite{CapraraKPP00}};

\end{axis}
\end{tikzpicture}

%\caption{First figure.}
%\label{fig:first}
\end{subfigure}
\hfill
\begin{subfigure}[t]{0.48\textwidth}
\centering

\begin{tikzpicture}
\begin{axis}[
  width=\linewidth,
  height=6.3cm,
  xmin=0, xmax=1,
  ymin=1, ymax=7,
  axis x line*=bottom,
  axis y line*=left,
  axis line style={line width=1pt},
  tick style={line width=1pt},
  major tick length=2.5pt,
  xlabel={},
  ylabel={},
ytick={1,...,7},
tick align=outside,
tick style={line width=1pt},
major tick length=2.5pt,
xtick={
  0,
  2/9,
  2/7,
  1/3,
  1/2,
  3/5,
  3/4,
  1
},
xticklabels={
  $0$,
  $\frac{2}{9}$,
  $\frac{2}{7}$,
  $\frac{1}{3}$,
  $\frac{1}{2}$,
  $\frac{3}{5}$,
  $\frac{3}{4}$,
  $1$
},
xticklabel style={
  font=\normalsize,
  rotate=0,
  anchor=north
},
yticklabel style={
  font=\normalsize
},
  ymajorgrids=true,
  grid style={gray!20},
  clip=false,
  after end axis/.code={
    % x-axis arrow, ending at x=1
    \draw[-{Stealth[length=3pt,width=4pt]}, line width=1pt]
      (rel axis cs:0.96,0) -- (rel axis cs:1,0);
    % y-axis arrow
    \draw[-{Stealth[length=3pt,width=4pt]}, line width=1pt]
      (rel axis cs:0,1) -- (rel axis cs:0,1.08);
    % manual axis labels
    \node[font=\normalsize, anchor=west]
      at (rel axis cs:1.02,0.) {$\varepsilon$};
    \node[font=\normalsize, anchor=south]
      at (rel axis cs:0.11,1.085) {time exponent};
    \node[font=\normalsize, anchor=south]
      at (rel axis cs:0.93,1.085) {$d=3$};
  },
]

% ------------------------------------------------------------
% Red region:
% below ceil(1/(2 eps) - 1/2) - 0.001, clipped to y<=7
% ------------------------------------------------------------
\addplot[
  draw=none,
  fill=red!30,
  fill opacity=0.8
] coordinates {
  (0,7)
  ({1/15},7)     ({1/15},6.999)
  ({1/13},6.999) ({1/13},5.999)
  ({1/11},5.999) ({1/11},4.999)
  ({1/9},4.999)  ({1/9},3.999)
  ({1/7},3.999)  ({1/7},2.999)
  ({1/5},2.999)  ({1/5},1.999)
  ({1/3},1.999)  ({1/3},1)
  (1,1)
  (0,1)
} \closedcycle;

% ------------------------------------------------------------
% Green staircase:
% max( ceil(1/eps - 1/2 + 0.001), 3 )
% ------------------------------------------------------------
\def\delta{0.001}

% Horizontal pieces for y = 4,...,7
\foreach \k in {4,...,7}{
  \pgfmathsetmacro{\xl}{1/(\k + 0.5 - \delta)}
  \pgfmathsetmacro{\xr}{1/(\k - 0.5 - \delta)}
  \addplot[green!60!black, very thick] coordinates {
    (\xl,\k) (\xr,\k)
  };
}

% Vertical pieces between y = 4,...,7
\foreach \k in {5,...,7}{
  \pgfmathsetmacro{\xx}{1/(\k - 0.5 - \delta)}
  \pgfmathtruncatemacro{\kmone}{\k-1}
  \addplot[green!60!black, very thick] coordinates {
    (\xx,\kmone) (\xx,\k)
  };
}

% Bottom horizontal piece at y = 3, extending to eps = 1
\pgfmathsetmacro{\xthree}{1/(3.5 - \delta)}
\addplot[green!60!black, very thick] coordinates {
  (\xthree,3) (1,3)
};

% Jump from 4 to 3
\addplot[green!60!black, very thick] coordinates {
  (\xthree,3) (\xthree,4)
};

% ------------------------------------------------------------
% Blue staircase:
% ceil(3/eps) - 3
% ------------------------------------------------------------
\foreach \b in {1,...,7}{
  \pgfmathsetmacro{\xl}{3/(\b+3)}
  \pgfmathsetmacro{\xr}{min(1,3/(\b+2))}
  \addplot[blue, very thick, dotted] coordinates {
    (\xl,\b) (\xr,\b)
  };
}

\foreach \b in {2,...,7}{
  \pgfmathsetmacro{\xx}{3/(\b+2)}
  \pgfmathtruncatemacro{\bmone}{\b-1}
  \addplot[blue, very thick, dotted] coordinates {
    (\xx,\bmone) (\xx,\b)
  };
}

% ------------------------------------------------------------
% Annotations
% ------------------------------------------------------------
\node[font=\normalsize] at (axis cs:{1/6},{3/2})
  {Thm.\ \ref{thm:2lower}};

\node[font=\normalsize, align=center] at (axis cs:{1/3},4.05)
  {Thm.\\ \ref{thm:dalgo}};

\node[font=\normalsize] at (axis cs:0.68,4.43)
  {Caprara et al.\ \cite{CapraraKPP00}};

\end{axis}
\end{tikzpicture}

%\caption{Second figure.}
%\label{fig:second}
\end{subfigure}

\caption{Illustration of our results, showing the running time exponent over $\eps$. The region where by \Cref{thm:2lower} one would improve the state of the art for $k$-SUM is filled red. The previously fastest algorithm by Caprara et al.~\cite{CapraraKPP00} is drawn dotted blue; other prior algorithms~\cite{ChandraHW76,OguzM80,FriezeC84} are not shown as they are dominated by~\cite{CapraraKPP00}. Our results are drawn green. The left side shows $d=2$, where we draw the better of  \Cref{thm:dalgo} (thick green) and \Cref{thm:linear} (dashed green), and we improve the exponent for every $\eps \in (0,2/3)$. The right side shows $d=3$, where we improve the exponent for every $\eps \in (0,1/2)$.}
\label{fig:exponentplots}
\end{figure}

\subsection{Technical Overview}
\label{sec:techoverview}

\paragraph{Reduction from \boldmath$k$-SUM to 2-Knapsack}
Here we describe the complete, simple proof of \Cref{thm:2lower}, demonstrating the close connection of 2-Knapsack and $k$-SUM.
Kulik and Shachnai~\cite{KulikS10} found a reduction from $k$-SUM to 2-Knapsack, which we describe here slightly reformulated. 
We work with the following variant of $k$-SUM: given a set of $n$ numbers $A \subseteq \mathbb{R}$, decide whether there exist distinct $a_1,\ldots,a_k \in A$ with $a_1+\ldots+a_k = 0$. 
Given a $k$-SUM instance $A$, let $M$ be a sufficiently large number, say $M \coloneqq 1 + \max_{a \in A} |a|$. For each $a \in A$ we construct an item with profit $p(a) \coloneqq 1$ and weights $\weight_1(a) \coloneqq 1 + a/M$ and $\weight_2(a) \coloneqq 1 - a/M$. We set both weight capacities to $\capacity_1 \coloneqq \capacity_2 \coloneqq k$. This finishes the construction. Observe that if there exists a $k$-SUM solution $a_1,\ldots,a_k \in A$, then the corresponding set of items $J = \{a_1,\ldots,a_k\}$ has weights $\weight_1(J) = k + (a_1+\ldots+a_k)/M = k$ and $\weight_2(J) = k - (a_1+\ldots+a_k)/M = k$, so it is a feasible solution. Since it has total profit $p(J) = k$, the optimal profit is at least $k$. Conversely, if no $k$-SUM solution exists, then for every set of items $J = \{a_1,\ldots,a_k\}$ we have $a_1+\ldots+a_k \ne 0$ and thus either $\weight_1(J) > k$ or $\weight_2(J) > k$. Hence, no set consisting of $k$ items is feasible, and thus also no set consisting of at least $k$ items is feasible, so the optimal profit is at most $k-1$.
Note that for any $\eps > 0$ satisfying $(1-\eps) k > k-1$, a $(1-\eps)$-approximation algorithm for 2-Knapsack distinguishes these two cases, and thus requires at least as much time as $k$-SUM. Rearranging $(1-\eps) k > (k-1)$ to $k < 1/\eps$ yields that for $k \coloneqq \lceil 1/\eps \rceil - 1$ a $(1-\eps)$-approximation for 2-Knapsack requires at least as much time as $k$-SUM. This proves \Cref{thm:2lower}.

%From here Kulik and Shachnai~\cite{KulikS10} concluded that 2-Knapsack does not have an EPTAS, since $k$-SUM cannot be solved in FPT time $f(k) n^{O(1)}$~\cite{DowneyF95}, assuming FPT$\ne$W[1]. 
%From the same reduction, Jansen, Land, and Land~\cite{JansenLL16} concluded that 2-Knapsack has no PTAS in time $n^{o(1/\eps)}$, since $k$-SUM is not in time $n^{o(k)}$~\cite{PatrascuW10}, assuming ETH. 
%
%We instead invoke the $k$-SUM Hypothesis, which states that $k$-SUM is not in time $O(n^{\lceil k/2 \rceil - \delta})$ for any $\delta > 0$. Substituting $k = \lceil 1/\eps \rceil - 1$, we obtain that 2-Knapsack has no PTAS in time $O(n^{\lceil (\lceil 1/\eps \rceil - 1)/2 \rceil - \delta})$ for any $\delta > 0$. Using the identities $\lceil x \rceil - m = \lceil x - m \rceil$ and $\lceil \frac{\lceil x \rceil}m \rceil = \lceil \frac xm \rceil$ for any $x \in \mathbb{R}$ and $m \in \mathbb{N}$, we can simplify the exponent to $\lceil (\lceil 1/\eps \rceil - 1)/2 \rceil = \lceil 1/(2\eps) - 1/2 \rceil$. Consequently, 2-Knapsack has no PTAS running in time $O(n^{\lceil 1/(2\eps) - 1/2 \rceil - \delta})$ for any $\delta > 0$, proving \Cref{thm:2lower}.

\paragraph{Prior Algorithms}
Prior approximation schemes for $d$-Knapsack~\cite{ChandraHW76,OguzM80,FriezeC84,CapraraKPP00} all use the following common approach. Order the items by non-increasing profit $p(1) \ge p(2) \ge \ldots \ge p(n)$. Let $\OPT = \{i_1,\ldots,i_k\}$ be an optimal solution, with $i_1 < \ldots < i_k$. We guess the high-profit items $H = \{i_1,\ldots,i_q\}$ for an appropriately chosen~$q$, meaning we iterate over all $O(n^q)$ choices. (We can assume $k > q$; otherwise we can guess $\OPT$ directly.) For each guess, we solve the LP relaxation of $d$-Knapsack on the remaining items $i_q+1,\ldots,n$ with the residual capacities $\capacity - \weight(H)$. 
Since $L \coloneqq \OPT \setminus H$ is a feasible solution for this LP, the computed optimal basic solution~$x$ satisfies $p(x) \ge p(L)$.
Since the LP has $d$ general linear constraints corresponding to the $d$ weight dimensions, it is well known that the basic solution $x$ has at most $d$ fractional values. Discarding the fractional items of $x$ and adding the guessed items $H$ yields a solution~$S$ of profit $p(S) \ge p(x) - d \cdot p(i_q) + p(H)$, as we lose up to $d$ fractional items, each of which has profit up to~$p(i_q)$. By  $p(x) \ge p(L) = p(\OPT) - p(H)$ we obtain $p(S) \ge p(\OPT) - d \cdot p(i_q)$. By sortedness we have $p(i_1) \ge \ldots \ge p(i_q)$ and thus $p(i_q) \le p(H)/q \le p(S)/q$. This yields $p(S) \ge p(\OPT) - d \cdot p(i_q) \ge p(\OPT) - d \cdot p(S) / q$, which rearranges to $p(S) \ge p(\OPT) / (1 + d/q)$. Setting $q \coloneqq \lceil d \cdot (1-\eps)/\eps \rceil$ guarantees the desired approximation ratio $p(S) \ge (1-\eps) \cdot p(\OPT)$. Outputting the best solution found over all guesses thus yields an approximation scheme.

The algorithms in~\cite{ChandraHW76,OguzM80,FriezeC84,CapraraKPP00} differ in how they solve the LP in the above algorithm approach. For example, when~\cite{ChandraHW76} was published, it was not yet known that an LP with a constant number of general linear constraints can be solved in polynomial time, and thus that paper was only able to solve an unbounded problem variant of $d$-Knapsack. Once polynomial-time solvers for such LPs were found, the algorithm of~\cite{ChandraHW76} could easily be adapted to solve the standard variant of $d$-Knapsack~\cite{OguzM80}. Some of the speedups of later work simply came from employing faster LP solvers; this line of improvements ended with an $O(n)$-time algorithm for solving the LP relaxation of $d$-Knapsack~\cite{MegiddoT93}.

Using this linear-time LP solver, the outlined algorithm approach requires time $O(n^q \cdot n) = O(n^{\lceil d /\eps \rceil - d + 1})$, as it is dominated by solving an LP for each of $O(n^q)$ guesses, where $q = \lceil d \cdot (1-\eps)/\eps \rceil$.

Caprara et al.~\cite{CapraraKPP00} further improved this running time by a factor $n$ to $O(n^{\lceil d /\eps \rceil - d})$, by introducing one more idea: For the LP relaxation they consider the solutions that pick a single fractional item, in addition to the solution that discards all fractional items. The best of these up to $d+1$ solutions is a $1/(d+1)$-approximation to the LP. This turns out\footnote{Indeed, this $1/(d+1)$-approximation yields a solution $S$ with $p(S) \ge p(H) + \frac 1{d+1} p(L)$, and we previously argued that $p(S) \ge p(\OPT) - d \cdot p(i_q) \ge p(\OPT) - \frac dq p(H) = (1 - \frac dq) p(H) + p(L)$. Adding these inequalities with weights $1 - \alpha$ and $\alpha$ for $\alpha \coloneqq \frac q{q+1+d}$ yields $p(S) \ge (1 - \frac d{q+1+d}) (p(H) + p(L)) = (1 - \frac d{q+1+d}) p(\OPT)$. This is a $(1-\eps)$-approximation for $q \coloneqq \lceil d/\eps \rceil - d - 1$, which improves the running time to $O(n^q \cdot n) = O(n^{\lceil d /\eps \rceil - d})$.}
to improve the total time by a factor~$n$.

\paragraph{Gap Between Upper and Lower Bound}
Prior algorithms solve 2-Knapsack in time $n^{2/\eps + O(1)}$, and by the above reduction improving this time to $n^{1/(2\eps) - \Omega(1)}$ (for every sufficiently small $\eps$) would yield faster algorithms for $k$-SUM (for every sufficiently large $k$). 
Thus, without improving the state of the art for $k$-SUM, the remaining gap is roughly a factor-4 in the exponent.
We develop a new algorithm for $d$-Knapsack that runs in time $\tOh(n^{(d-1)/(2\eps) + O(1)} + n^{d})$ for constant $d$ and $\eps$. For $d=2$ and small $\eps$ this closes the factor-4 gap. 
Our improvement comes in two stages. The first stage gives a factor-$d/(d-1)$ improvement by replacing the LP solver by an algorithm for \emph{instances with slack}. The second gives a factor-2 improvement by \emph{meet in the middle}. 

\paragraph{Stage 1: Instances with Slack}
Suppose we want to find a solution satisfying capacities~$\capacity$, but we compare only with the best solution satisfying the stricter capacities $(1-\delta) \capacity$. This is known as resource augmentation. As it is a problem with slack in the weights, it is easy to solve. For example, by rounding all weights to multiples of $\frac \delta n \capacity$ and using dynamic programming this problem can be solved in time $(n/\delta)^{O(d)}$, see \Cref{lem:ddim_dp} in \Cref{sec:dalgo_algoslack}. In fact, by adapting techniques from modern Knapsack approximation algorithms~\cite{Chan18a}, we can improve this running time to $n \cdot (\log(n)/\delta)^{O(d)}$, see \Cref{lem:ddim_dpv} in \Cref{sec:dalgo_representativesols}.
Crucially, we observe that this problem remains solvable in the same time if we only have slack in $d-1$ weight dimensions, i.e., we compare with the best solution with capacity $(1-\delta) \capacity_t$ in dimension $t=1,\ldots,d-1$ and capacity $\capacity_d$ in dimension~$d$. Intuitively, this works because we inherently have slack in the profit constraint (as we are approximating profit) and in $d-1$ weight constraints, so only one constraint needs to be handled exactly, which is possible in the rounding and dynamic programming approach. 

In order to apply an algorithm for instances with slack, we first need to somehow \emph{generate slack}.
We achieve this by an LP-based argument showing that: 
Any set of items $L$ has a subset~$Q$ such that $\weight_t(Q) \le (1 - \delta) \weight_t(L)$ for each $t \in \{1,\ldots,d-1\}$ and $p(Q) \ge (1-\delta) p(L) - (d-1) \cdot \max_{i \in L} p(i)$, see \Cref{lem:structural} in \Cref{sec:dalgo_slackgen}. Applying this to the low-profit items of the optimal solution $L \coloneqq \{i_{q+1},\ldots,i_k\}$ shows that an algorithm for the problem with slack, when run on the items $i_q+1,\ldots,n$, returns a solution of profit at least $(1-\delta) p(L) - (d-1) \cdot p(i_q)$. (For comparison, the prior algorithms achieved profit at least $p(L) - d \cdot p(i_q)$ after removing the fractional items of an LP solution).
The $(1-\delta)$-factor is negligible by setting, say $\delta \coloneqq \eps^{10}$. More importantly, the factor $d$ is reduced to $d-1$, which allows to pick $q \approx (d-1)/\eps$ (compared to $q \approx d/\eps$ in prior algorithms). This decrease by essentially a factor $d/(d-1)$ improves the final running time exponent by the same factor.
Ultimately, this improvement stems from our approach only needing slack in $d-1$ weight dimensions, so it suffices to apply our LP-based structural argument to an LP with only $d-1$ general linear constraints. 

For the reader's convenience, in \Cref{sec:dalgo_simplealgo} we describe a simple algorithm for $d$-Knapsack based on the ideas described so far. This algorithm achieves time $n^{(d-1)/\eps} \cdot (n/\eps)^{O(d)}$, which as promised improves the prior algorithms by a factor $d/(d-1)$ in the exponent, up to the additional~$O(d)$.

\paragraph{Stage 2: Meet in the Middle}
So far our ideas solve $d$-Knapsack in time $n^{(d-1)/\eps} \cdot (n/\eps)^{O(d)}$, so in particular 2-Knapsack in time $n^{1/\eps} \cdot (n/\eps)^{O(1)}$. 
Recall that any $(1-\eps)$-approximation algorithm for 2-Knapsack also solves $k$-SUM for $k = \lceil 1/\eps \rceil - 1$. The latter can naively be solved in time $O(n^k) = O(n^{\lceil 1/\eps \rceil - 1}) = n^{1/\eps \pm O(1)}$, which matches what we achieved for 2-Knapsack.
Meet in the middle\footnote{Roughly speaking, to solve $k$-SUM via meet in the middle one precomputes and sorts all sums of $\lceil k/2 \rceil$ input numbers. Subsequently, for each sum $s$ of $\lfloor k/2 \rfloor$ input numbers, one performs binary search to check whether $-s$ appears in the precomputed list. Some additional care is needed to ensure that no number is selected twice.} improves $k$-SUM from the naive time $O(n^k)$ to $\tOh(n^{\lceil k/2 \rceil})$, and all ways known to solve $k$-SUM faster than the naive time use some variant of meet in the middle.
%The only known way to solve $k$-SUM for $k \ge 4$ faster than the naive time $O(n^k)$ is via meet in the middle, which solves $k$-SUM in time $\tOh(n^{\lceil k/2 \rceil})$ (and the $k$-SUM Hypothesis postulates that this time is essentially optimal). 
%
Hence, to achieve running time better than $n^{1/\eps \pm O(1)}$ for 2-Knapsack we need to apply meet in the middle (or some generalization thereof). As meet in the middle improves the naive running time exponent of $k$-SUM by essentially a factor~2, this will yield the promised factor-2 improvement for 2-Knapsack, improving the running time to $n^{1/(2\eps) + O(1)}$. 

If the optimal solution consisted only of the high-profit items $H = \{i_1,\ldots,i_q\}$ then applying meet in the middle would be relatively straightforward: Split $H = H_1 \cup H_2$, where $H_1 = \{i_1,\ldots,i_{\lceil q/2 \rceil}\}$ and $H_2 = H \setminus H_1$. In time $O(n^{\lceil q/2 \rceil})$ we can enumerate all choices of $H_1$, and for each one write down its weights and profit. Then in the same time we can enumerate all choices of $H_2$. We want to pair each enumerated $H_2$ with a valid $H_1$ such that $H_1 \cup H_2$ maximizes total profit without exceeding the capacities. This can be written as an \emph{orthogonal range query}, so by storing the choices for $H_1$ in an orthogonal range data structure we can find the optimal match for any given $H_2$ in time $(d \log n)^{O(d)}$, for details see \Cref{sec:dalgo_ORS}
(some care is needed to ensure that $H_1$ and $H_2$ are disjoint, which we omit here). 

The main challenge is to incorporate the low-profit items into this meet in the middle approach. Note that we can only afford to solve a very easy problem for each pair $(H_1,H_2)$ (such as verifying total weight sums, and maximizing total profit). We cannot solve an LP for each pair $(H_1,H_2)$, and thus combining meet in the middle with the approach of prior algorithms is hopeless.

Fortunately, our novel approach of solving instances with slack from Stage 1 can be combined with meet in the middle!
To this end, we show that the infinite set of all possible capacity vectors~$\capacity$ for a set of items $I$ can be \emph{discretized}, obtaining a set ${\cal W}$ of only $(n/\eps)^{O(d)}$ capacity vectors with the following property: Solving the problem with slack for each $\capacity \in {\cal W}$ produces a pool of \emph{representative solutions}, such that for any arbitrary $\capacity$ a solution to the problem with slack can be found among these representative solutions, see \Cref{lem:low_profit} in \Cref{sec:dalgo_representativesols}. 
That is, in all subsequent computations we can replace low-profit items by only $(n/\eps)^{O(d)}$ representative solutions. 
This allows to incorporate the low-profit items into meet in the middle as follows. We build an orthogonal range data structure on all choices of $H_1$. Then we iterate over all pairs of $H_2$ and a representative solution on the low-profit items, issuing an orthogonal range query to obtain the best match $H_1$. The total running time becomes $n^{\lceil q/2 \rceil} \cdot (n/\eps)^{O(d)}$, where the first factor is the number of choices for~$H_1$ and the second factor is the number of representative solutions (recall that $q \approx (d-1)/\eps$). 

\medskip
We develop two variants of this algorithm. In the variant presented in \cref{sec:ddim_improved}, we add several low-level optimizations and are significantly more careful in bounding the running time. This gives the final running time $\big(n^{\lceil \frac{d-1}{2\eps} - \frac 12 + \rho \rceil}+n^{d}\big) \cdot \big( \frac{d \log n}{\rho\cdot \eps }\big)^{O(d)}$ for any $\rho \in (0,1)$, proving \Cref{thm:dalgo}.

In the variant presented in \cref{sec:reduction_to_ksum}, we replace the step that performs meet in the middle via orthogonal range searching by an intermediate problem. Reducing this intermediate problem to $k$-SUM then yields \cref{thm:2reductiontokSUM}. The reduction makes use of fine-grained techniques to replace inequality constraints by equality constraints~\cite{VassilevskaWW13,GokajK25} as well as the Frank-Tardos lemma~\cite{FrankT87} to replace real input numbers by integers of polynomial bit complexity (as is also used in recent work on 3-SUM over the reals~\cite{FischerPS26}).

\subsection{Further Related Work}
\label{sec:further}

When $d$ is unbounded, $d$-Knapsack admits a polynomial-time $\Omega(1/d)$-approximation~\cite{CapraraKPP00,srinivasan1995improved}, but no $\Omega(d^{-1/2+\eps})$-approximation for any $\eps > 0$ assuming $\textnormal{P}\neq \textnormal{NP}$~\cite{Zuckerman07,Hastad96,chekuri2004multidimensional}.\footnote{This follows from Zuckerman's derandomization~\cite{Zuckerman07} of Håstad's hardness result for Maximum Independent Set~\cite{Hastad96} by a simple reduction to $d$-Knapsack; see also~\cite{chekuri2004multidimensional} for a generalization to a larger class of problems.}

In this work, as subroutines we develop algorithms for $d$-Knapsack with resource augmentation, where the capacity constraints are relaxed, see~\Cref{lem:ddim_dpv,lem:ddim_dp,lem:internal-APX-Lemma}. 
Different, incomparable settings of resource augmentation for generalizations of multidimensional Knapsack were recently studied by Brinkop et al.~\cite{brinkop_et_al:LIPIcs.STACS.2026.20} and Armbruster et al.~\cite{ArmbrusterGTW26}.

%Brinkop et al.~\cite{brinkop_et_al:LIPIcs.STACS.2026.20} recently studied a different, incomparable setting of resource augmentation for multidimensional Knapsack and, more generally, integer linear programs.

Knapsack is also extensively studied in geometric settings, where geometric
objects need to be packed without overlap into a square. We mention two examples that are faintly related to our work. Recently Kar, Khan, and Wiese~\cite{KarKW26} studied the two-dimensional
geometric Knapsack problem with rotations and designed a PTAS in the cardinality
case and a $(1.497+\eps)$-approximation in the weighted case. They also proved
an $n^{\Omega(1/\eps)}$ conditional lower bound for approximation schemes under
the $k$-SUM Hypothesis. 
As another example, Khan, Sharma, and Sreenivas~\cite{KhanSS21,khan_et_al:LIPIcs.FSTTCS.2022.23} studied a generalization of geometric Knapsack that
additionally imposes $d$-dimensional vector constraints, and presented a
$(2+\eps)$-approximation both with and without rotations.

Knapsack and its variants have also been extensively studied from the
perspective of pseudo-polynomial-time algorithms, where the input consists of
integers bounded by $W$, see, e.g.,~\cite{AbboudBHS22,Bringmann17SubsetSum,Bringmann24SmallItems,
BringmannDP24,BringmannDW26,Chan18a,Chan26SubsetSum,CyganMWW19,
DoronAradKM26,EisenbrandW20,JansenR23,Jin24Knapsack,KunnemannPS17,
RohwedderW25}. 
%In this setting,
%$d$-Knapsack admits a folklore $\poly(n) W^d$-time dynamic programming algorithm, see~\cite{DoronAradKM26}. Recently, Bringmann, Dürr, and W\k{e}grzycki~\cite{BringmannDW26}
%provided a fine-grained lower bound showing that $d$-Knapsack cannot be solved
%in time $2^{o(n)}W^{d-\delta}$ for any fixed $d\geq 1$ and $\delta>0$; see also~\cite{AbboudBHS22} for the case of $d=1$.

\section{Preliminaries}
\label{sec:prelim}

Let $d \ge 2$ be an integer. An instance of $d$-Knapsack consists of a $d$-dimensional vector of capacities $\capacity \in \mathbb{R}_{\ge 0}^d$ and a set $I$ of $n$ items. 
Each item~$i \in I$ has an associated profit $p(i) \in \mathbb{R}_{\ge 0}$ and a $d$-dimensional vector of weights $\weight(i) \in \mathbb{R}_{\ge 0}^d$. 
For a set of items $S\subseteq I$, let $\weight(S) \coloneqq \sum_{i\in S} \weight(i)$ and $p(S) \coloneqq \sum_{i\in S} p(i)$ denote the total weight and total profit of $S$, respectively. A feasible solution is a subset $S\subseteq I$ satisfying $\weight(S)\leq \capacity$. Here, for two vectors $\bar{x},\bar{y}\in \mathbb{R}^d$ we say that $\bar{x}\leq \bar{y}$ if and only if $\bar{x}_t \leq \bar{y}_t$ for all $t\in \{1,\ldots, d\}$.  The objective is to find a feasible solution $S$ of maximum profit~$p(S)$. Without loss of generality we assume that $\weight(i)\leq \capacity$ for every $i\in I$ throughout the paper.

We adopt an object-oriented view, where the profit $p(i)$ and weight vector $\weight(i)$ are attributes of the object $i$. Accordingly, we denote an instance by $(I,\capacity)$; in particular we do not need to specify a profit function $p$ or weight function $\weight$, as they are given by attributes of items.

We assume that items in $I$ are indexed by positive integers, e.g., by $\{1,\ldots,n\}$. We often abbreviate this assumption by writing $I = \{1,\ldots,n\}$. We denote by $\max(I)$ the maximum index of any item in $I$, and by $\min(I)$ the minimum index. For $i,i' \in I$ we write $i < i'$ to mean that the index of $i$ is less than the index of $i'$. We assume that \emph{items are sorted by non-increasing profits}, that is, for every $i,i'\in I$ such that $i<i'$ it holds that $p(i)\geq p(i')$. 
For every instance $\cI=(I,\capacity)$ we fix an arbitrary optimal solution, which we denote by $\OPT_{\cI}$. 
When the instance $\cI$ is clear from the context, we omit the subscript and simply write $\OPT$.

For a subset of items $S\subseteq I$, we use $S[j]$ to denote the $j$-th item in $S$, ordered by item indices. 
That is, if $S=\{i_1,\ldots,i_k\}$ with $i_1<i_2<\ldots<i_k$, then $S[j]=i_j$. 
Similarly, we write $S[a \semi b]$ for the set of items $\{i_j \mid \max\{a,1\} \le j \le \min\{b,|S|\} \}$ and use the abbreviations $S[\semi b] \coloneqq S[1 \semi b]$ and $S[a \semi] \coloneqq S[a \semi |S|]$.
In particular, under the above profit ordering of the items, $\OPT[j]$ denotes the item with the $j$-th highest profit in the optimal solution $\OPT$. 
Moreover, $\OPT[\semi  j]$ denotes the subset consisting of the $j$ highest-profit items in $\OPT$, and $\OPT[j\semi]$ denotes the subset of $\OPT$ obtained by removing the $j-1$ highest-profit items. 

$\tOh$-notation hides logarithmic factors in the argument, i.e., $\tOh(T) = \bigcup_{c \ge 0} O(T \log^c T)$.
We write $[d] = \{1,\ldots,d\}$. All logarithms in this paper are base 2. We use the conventions $\max(\emptyset) = -\infty$ and $\min(\emptyset) = \infty$.
We remark that the algorithms in this paper work both on the RealRAM model (with arbitrary real input numbers) and on the WordRAM model (with integer input numbers, each fitting into a machine cell). We focus our presentation on the RealRAM.

\section{A Simple Improved Algorithm for \boldmath$d$-Knapsack}
\label{sec:ddim}

As a warmup, in this section we present a simple algorithm that improves upon the state of the art. Specifically, we will achieve time $n^{(d-1)/\eps} \cdot (n/\eps)^{O(d)}$, which for $n \gg 1/\eps \gg d$ improves upon the previously best running time of $O(n^{\lceil d/\eps \rceil - d})$~\cite{CapraraKPP00}. 
As ingredients, in \Cref{sec:dalgo_slackgen} we present the Slack-Generating Lemma, and in \Cref{sec:dalgo_algoslack} we present an algorithm for instances with slack. We combine these ingredients to obtain our algorithm in \Cref{sec:dalgo_simplealgo}. 

For the further improved algorithm that we will see in \Cref{sec:ddim_improved}, we will only reuse \Cref{sec:dalgo_slackgen}.

\subsection{Slack-Generating Lemma}
\label{sec:dalgo_slackgen}
Our first building block is a lemma to generate slack. Given a subset $S\subseteq I$, it guarantees the existence of a subset $Q\subseteq S$ whose weight in each of the first $d-1$ dimensions is reduced by a factor of $1-\delta$, while its profit decreases by only a $(1-\delta)$ factor and an additive loss of at most $(d-1)$ times the maximum profit of an item in $S$. Although the lemma is purely existential, the generated slack will allow us to efficiently compute suitable surrogates for $S$ via dynamic programming.
\begin{lemma}[Slack-Generating Lemma]
	\label{lem:structural}
	Let $(I,\capacity)$ be a $d$-Knapsack instance, let $S\subseteq I$ be a nonempty subset of items,  and let $0<\delta<1$. 
	There exists  $Q\subseteq S$ satisfying:
	\begin{enumerate}
		\item $\weight_{t}(Q) \leq (1-\delta) \cdot \weight_t(S)$ for every $t\in \{1,\ldots,d-1\}$,
		\item $p(Q) \geq (1-\delta )\cdot p(S)-(d-1)\cdot \max_{i\in S}p(i)$, and
		\item $\max\big\{ p(Q), \max_{i \in S} p(i) \big\} \ge \frac {1-\delta}d \cdot p(S)$.
	\end{enumerate}
\end{lemma}

\begin{proof}	
The proof is based on a linear program defined over the items of $S$. The program scales the capacities in the first $d-1$ dimensions by a factor of $1-\delta$, and a simple scaling argument yields a feasible solution of profit $(1-\delta)\cdot p(S)$. We then consider a basic optimal solution and exploit the fact that it has at most $d-1$ fractional variables to construct the desired subset $Q$.
	
In more detail, let $\LP_S$ denote the following linear program:
	\begin{equation}
		\label{eq:structural_LP}
		\tag{$\LP_S$}
		\begin{aligned}
			\max \quad & \sum_{i \in S} p(i)\cdot x_i \\[6pt]
			\text{s.t.} \quad 
			& \sum_{i \in S} x_i\cdot  \weight_t(i) \;\le\; (1-\delta)\cdot \weight_t(S)
			&& \forall\, t \in \{1,\ldots,d-1\} \\[6pt]
			& 0 \;\le\; x_i \;\le\; 1
			&& \forall\, i \in S
		\end{aligned}
	\end{equation}
	Observe that the linear program \ref{eq:structural_LP}  only has variables for items in $S$. Also, note that in the first constraint the value of $t$ only goes up to $d-1$.

	Consider the vector $x\in[0,1]^S$ defined by
	$x_i=1-\delta$ for every $i\in S$.
	 Clearly, $0\leq x_i\leq 1$ for every $i\in S$. 
	Additionally, for every $t\in \{1,\ldots, d-1\}$ it holds that  $$\sum_{i\in S} x_i\cdot \weight_t(i) = \sum_{i\in S} (1-\delta)\cdot \weight_t(i) = (1-\delta)\cdot \weight_t(S).$$
	Hence $(x_i)_{i\in S}$ is a solution for \ref{eq:structural_LP}. Therefore, \begin{equation}
		\label{eq:opt_LPS}\OPT(\LP_S) \geq \sum_{i\in S} p(i)\cdot x_i=\sum_{i\in S} p(i)\cdot (1-\delta) = (1-\delta)\cdot p(S).
	\end{equation}
	
	Let $x^*$ be a {\em basic} optimal solution for $\LP_S$. Since $\LP_S$ has $d-1$ constraints beside the constraints $0\leq x_i\leq 1$, it holds that $x^*$ has at most $d-1$ fractional entries, that is, $|F|\leq d-1$ where $F=\{ i\in S~|~x^*_i\in (0,1)\}$. Define $Q= \{i\in S~|~x^*_i =1\}$. Clearly, $Q\subseteq S$.
	We will show that the set~$Q$ satisfies the additional properties stated in the lemma. 
	
	For property 1, note that for every $t\in \{1,\ldots, d-1\}$ it holds that 
	$$
	\weight_t(Q) = \sum_{i\in Q} \weight_t(i) \leq \sum_{i\in S} x^*_{i}\cdot \weight_t(i) \leq (1-\delta)\cdot \weight_t(S),
	$$
	where the last inequality holds as $x^*$ is a solution for $\LP_S$. 
	
	For property 2, we observe that 
	$$
	\begin{aligned}
		p(Q)&= \sum_{i\in Q} p(i)\\
		& = \sum_{i\in Q} x^*_i\cdot p(i) \\
		&= \sum_{i\in S} x^*_i\cdot p(i)  - \sum_{i\in F} x^*_i \cdot p(i)  \\
		&\geq \OPT(\LP_S) -|F|\cdot \max_{i\in S} p(i) \\
		&\geq (1-\delta)\cdot p(S) - (d-1)\cdot \max_{i\in S} p(i),
	\end{aligned}
	$$
	where the last inequality follows from \eqref{eq:opt_LPS}. 
	
	For property 3, observe that 
	\[
	p(Q) + \sum_{i \in F} p(i) \ge p(x^*) = \OPT(\LP_S) \ge (1-\delta)\cdot p(S),
	\]
	where the last inequality uses \eqref{eq:opt_LPS}. Since $|F| \le d-1$, the left hand side consists of at most $d$ summands. Since the maximum summand is at least the average summand, it holds that $\max\big\{ p(Q), \max_{i \in F} p(i) \big\} \ge \frac 1d \big(p(Q) + \sum_{i \in F} p(i)\big)$. Combining these inequalities yields
	\[
	 \max\big\{ p(Q), \max_{i \in S} p(i) \big\} \ge \max\big\{ p(Q), \max_{i \in F} p(i) \big\} \ge \frac 1d \bigg(p(Q) + \sum_{i \in F} p(i)\bigg) \ge \frac{1-\delta}d \cdot p(S),
	\]
	which proves the lemma.
\end{proof}

\subsection{A Simple Algorithm for Instances with Slack}
\label{sec:dalgo_algoslack}

We now show that instances with slack can be solved by dynamic programming. Specifically, we want to compute a set $S$ that is feasible with respect to the given capacities $\capacity = (\capacity_1,\ldots,\capacity_d)$ and achieves at least a $1-\eps$ fraction of the best profit of any solution that is feasible for the stricter capacities $((1-\eps) \capacity_1, \ldots, (1-\eps) \capacity_{d-1}, \capacity_d)$. This means we have slack in the first $d-1$ dimensions, because the solution that we compare to must adhere to stricter capacity bounds.

\begin{lemma}[Solving Instances with Slack by Dynamic  Programming]
	\label{lem:ddim_dp}
	Given a $d$-Knapsack instance $(I,\capacity)$ and an error parameter $\eps \in (0,1]$, in time $(\frac{n}{\eps})^{O(d)}$ we can compute a subset of items $S\subseteq I$ such that:
	\begin{itemize}
		\item $\weight(S)\leq \capacity$ and 
		\item $p(S) \geq(1-\eps)\cdot \max\left\{ p(Q)\,\middle|\,\begin{aligned} 
			&Q\subseteq I\\
			&\weight_t(Q)\leq (1-\eps)\cdot  \capacity_t~~~&\forall t\in [d-1]\\
			&\weight_d(Q)\leq \capacity_d 
		\end{aligned} \right\}.$ 
	\end{itemize}
\end{lemma}
\begin{proof}
	The lemma follows from a standard application of rounding and dynamic programming. We explain the details for completeness.
	We preprocess the instance by removing all items $i \in I$ such that $\weight_d(i) > \capacity_d$ or there exists $t \in [d-1]$ with $\weight_t(i) > (1-\eps)\capacity_t$. This does not affect the set 
	\[ \mathcal{Q} \coloneqq \{ Q \subseteq I \mid \weight_d(Q) \le \capacity_d, \weight_t(Q) \le (1-\eps) \capacity_t \; \forall t \in [d-1] \}. \]
	We may assume that $I = \{1,\ldots,n\}$, that is, the items are indexed by the numbers 1 to~$n$.	
	Let $\tp \coloneqq \max\{p(i) \mid i \in I\}$. In case after the preprocessing we have $I = \emptyset$ or $\tp = 0$ we can simply return $S = \emptyset$.
	Otherwise, for each item $i \in I$ we round down the profit $p(i)$
	to the nearest multiple of $\eps\tp/n$, and we round up the weight $\weight_t(i)$ to the nearest
	multiple of $\eps\capacity_t/n$ for all $t \in [d-1]$.
	Denote the resulting rounded profits and weights by
	$\hat p(i)$ and $\hat w_t(i)$.
	
	We construct a dynamic programming table that has an entry for each choice of:
	\begin{itemize}[nosep]
		\item $j\in\{0,\ldots,n\}$, denoting the number of items considered so far,
		\item $\pi \in \eps\tp/n \cdot \{0,\ldots, \lfloor n^2/\eps \rfloor\}$, denoting a possible rounded profit value,\footnote{Here we used the notation $a \cdot X = \{a \cdot x \mid x \in X\}$.}
		\item $\omega_t \in \eps\capacity_t/n \cdot \{0,\ldots, \lfloor n/\eps \rfloor \}$, denoting a possible rounded weight in dimension $t \in [d-1]$.
	\end{itemize}
	Formally, we define
	\begin{equation}\label{eq:dp-definition}
	T[j,\pi,\omega_1,\ldots,\omega_{d-1}] \coloneqq \min\left\{
		\weight_d(S) \,\middle|\, \begin{aligned}&S\subseteq\{1,\ldots,j\}\\ &\hat p(S)=\pi\\ &\hat w_t(S)=\omega_t\,~~~\forall t\in[d-1]
		\end{aligned} \right\},
	\end{equation}
	where we recall the convention $\min(\emptyset) = \infty$.

	We fill this table as follows.
	In the base case $j=0$ we set $T[0,0,0,\ldots,0]=0$, and every other
	entry with $j=0$ is set to $\infty$.
	The dynamic programming transitions are given by the following recurrence relation:
	\begin{equation}\label{eq:dp-recursion}
	T[j,\pi,\omega_1,\ldots,\omega_{d-1}]
	= \min \left\{\begin{aligned}
	&T[j-1,\pi,\omega_1,\ldots,\omega_{d-1}],\\
	&T[
		j-1,\,
		\pi-\hat p(j),\,
		\omega_1-\hat w_{1}(j),\,
		\ldots,\,
		\omega_{d-1}-\hat w_{d-1}(j)]
	+\weight_d(j)
	\end{aligned}
	\right\}
	\end{equation}
	Whenever a table entry on the right hand side is out of bounds, we treat it as $\infty$.
	After filling the table, we determine the maximum $\pi$ such that there exist $\omega_1,\ldots,\omega_{d-1}$ satisfying $\omega_t \leq \capacity_t$ for every
	$t \in [d-1]$ and $T[n,\pi,\omega_1,\ldots,\omega_{d-1}] \leq \capacity_d$. The
	subset of items $S$ corresponding to this entry can be recovered by a standard backtracking
	procedure. We return this set $S$.
	
	\medskip
	To establish correctness, we first need to prove that the recurrence \eqref{eq:dp-recursion} correctly fills the table according to definition \eqref{eq:dp-definition}. We proceed by induction on $j \in \{0,\ldots,n\}$. The base case $j=0$ is straightforward to verify, so we focus
	on the inductive step. Assume that for all $j' < j$ and all $\pi', \omega_1',\ldots, \omega_{d-1}'$, the value $T[j',
	\pi', \omega'_1, \ldots, \omega'_{d-1}]$ is correctly filled according to definition \eqref{eq:dp-definition}. 		We now analyze $T[j, \pi, \omega_1, \ldots, \omega_{d-1}]$. Let $S
	\subseteq \{1, \ldots, j\}$ be such that $\hat{p}(S) = \pi$, $\hat{w}_t(S) =
	\omega_t$ for all $t \in [d-1]$ and $\weight_d(S)$ is minimized.
	Clearly, the table entry computed according to \eqref{eq:dp-recursion} satisfies $T[j, \pi, \omega_1, \ldots, \omega_{d-1}] \ge \weight_d(S)$ as it stores
	a $d$-th weight of a feasible solution. We need to show the opposite inequality.

	\textbf{Case 1}: If $j \notin S$, then $S \subseteq \{1,
	\ldots, j-1\}$, so $S$ is feasible for the entry $T[j-1, \pi, \omega_1, \ldots, \omega_{d-1}]$, and we obtain $\weight_d(S) \ge T[j-1, \pi, \omega_1, \ldots, \omega_{d-1}]$. This contributes the
	first term in~\eqref{eq:dp-recursion}.

	\textbf{Case 2}: If $j \in S$, then let $S' \coloneqq S \setminus \{j\} \subseteq \{1, \ldots, j-1\}$. Observe that
		$\hat{p}(S') = \pi - \hat{p}(j)$, and
		$\hat{w}_t(S') = \omega_t - \hat{w}_t(j)$ for all $t \in [d-1]$. Thus, $S'$ is feasible for the entry $T[j-1, \pi - \hat{p}(j), \omega_1 - \hat{w}_1(j), \ldots, \omega_{d-1} - \hat{w}_{d-1}(j)]$, and we obtain
	\begin{align*}
		\weight_d(S) = \weight_d(S') + \weight_d(j) \ge T[j-1, \pi - \hat{p}(j), \omega_1 - \hat{w}_1(j), \ldots, \omega_{d-1} - \hat{w}_{d-1}(j)] + \weight_d(j).
	\end{align*}
	This contributes the second term in the recurrence. Hence, the recurrence~\eqref{eq:dp-recursion} correctly fills the table according to definition \eqref{eq:dp-definition}.
	
	\medskip
	Inspecting the final step of the algorithm and using (\ref{eq:dp-definition}), we can observe that the algorithm computes a set $S \subseteq \{1,\ldots,n\} = I$ that maximizes $\hat p(S)$ subject to $\hat w_t(S) \le \capacity_t$ for all $t \in [d-1]$ and $\weight_d(S) \le \capacity_d$. Since we only round up weights, we have $\weight_t(S) \le \hat w_t(S) \le \capacity_t$ for all $t \in [d-1]$, and thus $\weight(S) \le \capacity$, showing the first bullet of the lemma.
	We also remark that $\hat p(S) \le p(S) \le n \cdot \tp$, and thus the considered range of $\pi$ is sufficient.
	For the second bullet, consider $Q \in \mathcal{Q}$ that maximizes $p(Q)$. Then by the rounding of weights, for any $t \in [d-1]$ it holds that $\hat w_t(Q) \le \weight_t(Q) + |Q| \cdot \eps\capacity_t/n \le (1-\eps) \capacity_t + \eps \capacity_t \le \capacity_t$. Hence, $Q$ 
	is a candidate for $S$, and thus $\hat p(S) \ge \hat p(Q)$. By the rounding of profits it holds that $\hat p(Q) \ge p(Q) - |Q| \cdot \eps\tp/n \ge p(Q) - \eps \tp$. Since after our preprocessing any single item $i \in I$ satisfies $\{i\} \in \mathcal{Q}$, it holds that $p(Q) \ge p(i)$ for all $i \in I$, and thus $p(Q) \ge \tp$. This yields $p(Q) - \eps \tp \ge (1-\eps) \cdot p(Q)$. Chaining these inequalities together, we obtain $p(S) \ge \hat p(S) \ge (1-\eps) \cdot p(Q)$, which proves the second bullet of the lemma.

	\medskip
	To analyze the running time, we first bound the number of states. By definition of $T$, there are $n+1$ choices for~$j$, $O(n^2/\eps)$ choices for $\pi$, and $O(n/\eps)$ choices for $\omega_t$ for each $t \in [d-1]$. 
	Hence, the total number of states is $(n/\eps)^{\Oh(d)}$. The same bound holds for the running
	time, as a single transition takes constant time.
\end{proof}

\subsection{A Simple Algorithm for \boldmath$d$-Knapsack}
\label{sec:dalgo_simplealgo}

With the ideas presented above, we can design a simple approximation scheme for the $d$-Knapsack problem, achieving running time
$
n^{\frac{d-1}{\eps}} \cdot \left(\frac{n}{\eps}\right)^{O(d)}$. For $n \gg \frac 1\eps \gg d$, this already improves upon the state-of-the-art running time of $O(n^{\lceil d/\eps\rceil - d})$~\cite{CapraraKPP00}.

Pseudocode of our algorithm is shown in~\Cref{alg:ddim_simple}. As in prior approximation schemes, we rely on enumeration to guess the most profitable items in an optimal solution. Specifically, we enumerate all $S \subseteq I$ of size at most $q \coloneqq \lceil (d-1)/\eps \rceil$ with $\weight(S) \le \capacity$, with the goal of guessing $S = \OPT[\semi q]$.
For each $S$, we consider the \emph{residual instance} given by the items that follow after $S$ in the order by non-increasing profits, that is, $I' \coloneqq \{ i \in I \mid i > \max(S) \}$, and by the remaining weight capacities $\capacity' \coloneqq \capacity - \weight(S)$.
We run \Cref{lem:ddim_dp} on the instance $(I',\capacity')$ with error parameter $\frac \eps 2$, which returns some set $R \subseteq I'$. In the end, we return the best of all computed solutions $S \cup R$.

\begin{algorithm}[t]
	\caption{$\simpleddim(I,\capacity, \eps)$}
	
	\label{alg:ddim_simple}
	
	$q \leftarrow\ceil{ \frac{d-1}{\eps}}$ 
	
	$\best \leftarrow \emptyset$
	
	\For{each $S\subseteq  I$ such that  $\abs{S}\leq q$\label{simple:loop} and $\weight(S)\leq \capacity$}{
		
		Let $I' \leftarrow \{i \in I \mid i > \max(S)\}$ and $\capacity' \leftarrow \capacity - \weight(S)$.
		
		Apply \Cref{lem:ddim_dp} on $(I',\capacity')$ with error parameter $\frac{\eps}{2}$; let $R$ be the returned solution. \label{simple:DP}
		
		If $p(S\cup R) > p(\best)$ then $\best\leftarrow S\cup R$. 
		
	}
	\Return{\best}
	
\end{algorithm}

Note that in contrast to prior approximation schemes we do not use any LP solver, instead we rely on the dynamic programming procedure of \Cref{lem:ddim_dp}. 
In both our algorithm and prior approximation schemes,  enumeration is used to compensate for an additive loss incurred by the procedure that extends the enumerated solution -- linear programming in prior approximation schemes, and dynamic programming in our case. Crucially, the combination of 
 \Cref{lem:structural}  and the dynamic programming of \Cref{lem:ddim_dp} leads 
to an additive loss in profit which only depends on $d-1$ items, whereas prior algorithms incurred an additive loss corresponding to $d$ items. This allows us to reduce the size of the enumerated sets from approximately $\frac{d}{\eps}$ to $\frac{d-1}{\eps}$.

\begin{lemma}
	\Cref{alg:ddim_simple} is a $\left(1-\eps \right)$-approximation algorithm for $d$-Knapsack and runs in time $n^{\frac{d-1}{\eps}} \cdot \left(\frac{n}{\eps}\right)^{O(d)}$.
\end{lemma}
\begin{proof}
	The running time bound follows from the number of iterations being at most $n^q = n^{\frac{d-1}{\eps} + O(1)}$, and each iteration running in time $\big(\frac{n}{\eps/2}\big)^{O(d)} = \left(\frac{n}{\eps}\right)^{O(d)}$ by \Cref{lem:ddim_dp}.
	
	Next we show that the returned set $\best$ is a feasible solution. Note that either $\best=\emptyset$, which is feasible, or $\best = S \cup R$ where $S$ is one of the choices iterated over in \Cref{simple:loop} and $R$ is the solution found in \Cref{simple:DP} in the same iteration. 
	Due to the condition in \Cref{simple:loop} it holds that $\weight(S)\leq \capacity$.
	Also, the set $R$ returned in \Cref{simple:DP} satisfies $\weight(R)\leq \capacity'= \capacity - \weight(S)$ by \Cref{lem:ddim_dp}. Therefore, $\weight(\best)=\weight(S\cup R) \leq \capacity$, and the algorithm returns a feasible solution. 

	As  defined in \Cref{sec:prelim},  we write $\OPT$ for an optimal solution for the instance $(I,\capacity)$, $\OPT[\semi q]$ for the $q$ most profitable items in $\OPT$ and $\OPT[q+1 \semi]$ for the remaining items in $\OPT$.

	If $\abs{\OPT}\leq q$ then in one of the iterations of \Cref{simple:loop} we have $S=\OPT$. Following this iteration it holds that $p(\best)\geq p(\OPT\cup R) \geq p(\OPT)$, and thus the algorithm returns an optimal solution.

	We are left with the case $\abs{\OPT}>q$. 
	We focus on the iteration of \Cref{simple:loop} in which $S=\OPT[\semi q]$.
	Observe that $\OPT\setminus \OPT[\semi q] = \OPT[q+1 \semi]$ is a feasible solution for the residual instance~$(I',\capacity')$.
	Applying \Cref{lem:structural} on $S = \OPT[q+1 \semi]$ and $\delta \coloneqq \frac \eps 2$ shows that there exists a set $Q\subseteq \OPT[q+1 \semi]$ such that
	$$p(Q)\geq \left(1-\frac{\eps}{2}\right)\cdot p(\OPT[q+1 \semi]) - (d-1)\cdot \max_{i\in \OPT[q+1 \semi]}p(i) \geq \left(1-\frac{\eps}{2}\right)\cdot p(\OPT[q+1 \semi]) - (d-1)\cdot p(\OPT[q+1]),$$
	and $$\weight_t(Q ) \leq  \left(1-\frac{\eps}{2}\right) \cdot \weight_t(\OPT[q+1 \semi] )\leq \left(1-\frac{\eps}{2}\right)\cdot \capacity'_t,$$ 
	for all $t\in \{1,\ldots, d-1\}$.
	Since $Q\subseteq \OPT[q+1 \semi]$  it also holds that $\weight_d(Q)\leq \weight_d(\OPT[q+1 \semi]) \leq \capacity'_d$. 
	The set $R$ returned by \Cref{lem:ddim_dp} used with error parameter $\eps/2$ in \Cref{simple:DP} thus satisfies
	$$
	\begin{aligned}
		p(R) &\geq \left(1-\frac{\eps}{2}\right)\cdot p(Q) \\
		&\geq \left(1-\frac{\eps}{2}\right) \left(\left(1-\frac{\eps}{2}\right)\cdot p(\OPT[q+1 \semi]) - (d-1)\cdot p(\OPT[q+1])\right)  \\
		&\geq (1-\eps)\cdot p(\OPT[q+1 \semi]) -(d-1)\cdot p(\OPT[q+1]).
	\end{aligned}
	$$
	Therefore,
	\begin{equation}
		\label{eq:scupr}
		p(S\cup R) \,\geq \,p(\OPT[\semi q]) +(1-\eps)\cdot p(\OPT[q+1 \semi]) -(d-1)\cdot p(\OPT[q+1]).
	\end{equation}
	As we assume the items are sorted by profit, and $\abs{\OPT}>q$, it holds that 
	$$
	p(\OPT[\semi q]) \geq q \cdot p(\OPT[q] )\geq q \cdot p(\OPT[q+1]). 
	$$
	 Therefore, 
	 $$(d-1) \cdot p(\OPT[q+1]) \leq \frac{d-1}{q} \cdot p(\OPT[\semi q]) \leq \eps \cdot p(\OPT[\semi q]).$$
	 Using this observation in \eqref{eq:scupr}, we obtain
	$p(S\cup R) \geq (1-\eps)\cdot p(\OPT)$.
	It follows that \Cref{alg:ddim_simple} returns a feasible solution with profit at least $(1-\eps)\cdot p(\OPT)$. 
\end{proof}

\section{Further Improved Algorithm for \boldmath$d$-Knapsack}
\label{sec:ddim_improved}

We further improve the running time for $d$-Knapsack to $\big(n^{\lceil \frac{d-1}{2\eps}-\frac12+\rho \rceil}+n^d\big) \cdot \big(\frac{d\log n}{\rho\eps}\big)^{O(d)}$ for any $\rho \in (0,1)$, thereby proving \Cref{thm:dalgo}.
As ingredients, we introduce an implicit representation of sets that we call virtual items in \Cref{sec:dalgo_virtual_items}, construct representative solutions via a faster algorithm for instances with slack in \Cref{sec:dalgo_representativesols} (here we will reuse the Slack-Generating Lemma from \Cref{sec:dalgo_slackgen}), and implement meet in the middle via orthogonal range searching in \Cref{sec:dalgo_ORS}. We combine these ingredients in \Cref{sec:dalgo_combine}. The details of the faster algorithm for instances with slack from \Cref{sec:dalgo_representativesols} are deferred to \Cref{sec:dalgo_proofapxlem}.

\subsection{Implicit Representation of Sets by Virtual Items}
\label{sec:dalgo_virtual_items}

To achieve the claimed running time, we sometimes need to avoid manipulating sets of items explicitly, since doing so can introduce an additional $O(n)$ factor. For this reason, we introduce an implicit representation of sets of items, which we call \emph{virtual items}.

\paragraph{Interface}
A virtual item $v$ representing a set of items $S$ is an object with the following operations:
\begin{enumerate}[label=(\arabic*)]
\item $p(v)$ returns $p(S)$ in time $O(1)$,
\item $\weight_t(v)$ returns $\weight_t(S)$ in time $O(1)$, for any given $t \in [d]$,
\item $\set(v)$ returns $S$ in time $O(|S|+1)$.
\end{enumerate}
We can create virtual items via the following primitives:
\begin{enumerate}[label=(\roman*)]
\item A virtual item $v_\emptyset$ representing $\emptyset$ can be constructed in time $O(d)$,
\item For any item $i$, a virtual item $v_{\{i\}}$ representing the set $\{i\}$ can be constructed in time $O(d)$,
\item Given virtual items $v_1,v_2$ representing disjoint sets $S_1,S_2$, a virtual item $v$ representing their union $S_1 \cup S_2$ can be constructed in time $O(d)$. We denote this operation by $v \coloneqq v_1 \cup v_2$.
\end{enumerate}
This finishes the description of the \emph{interface} of virtual items. We use the term virtual item because, although it represents a set, its profit and any weight coordinate can be accessed in time $O(1)$, just as for an ordinary item.

\medskip 
We introduce some shorthand notation:
For a virtual item $v$ and a set of items $Q$ disjoint from $\set(v)$, we denote by $v \cup Q$ a virtual item representing the set $\set(v) \cup Q$. The notation $v \cup Q$ is an abbreviation of $v \cup \bigcup_{i \in Q} v_{\{i\}}$. In particular, the virtual item $v \cup Q$ can be computed in time $O(d(|Q|+1))$. 

We also extend the notation $p, \weight, \set$ to sets containing both items and virtual items. That is, for $S=T\cup\{v_1,\ldots,v_k\}$, where $T\subseteq I$ is a set of items, $v_1,\ldots,v_k$ are virtual items, and the sets $T,\set(v_1),\ldots,\set(v_k)$ are disjoint, we define
\begin{align*}
	p(S)\coloneqq p(T) + \sum_{i=1}^k p(v_i), \qquad
	\weight(S)\coloneqq \weight(T) + \sum_{i=1}^k \weight(v_i), \qquad
	\set(S)\coloneqq T\cup\bigcup_{i=1}^k \set(v_i).
\end{align*}

\paragraph{Implementation}
The implementation details of virtual items are straightforward. We explain them in the remainder of this subsection.
In our implementation, a virtual item $v$ representing a set $S$ is an object with the following attributes:
\begin{itemize}
\item $v.p = p(S)$, the profit of the set represented by $v$,
\item $v.\weight = \weight(S)$, the weight vector of the set represented by $v$,
\item $v.\vitem$, a set of items of size at most one, that is, either $v.\vitem = \emptyset$ or $v.\vitem = \{i\}$ \smash{for an item $i$},
\item $v.\vleft$ and $v.\vright$, pointers to virtual items. Each pointer may be the null pointer $\textup{NULL}$. 
\end{itemize}
These objects adhere to the following invariants.
We ensure that there are no pointer cycles. Denote by $R(v)$ the set of virtual items that are reachable from $v$ by following $\vleft/\vright$ pointers, where $v$ is reachable from itself. We ensure that $S$ is the disjoint union of all sets $u.\vitem$ over all $u \in R(v)$. We also ensure that $|R(v)| = \max\{2|S|-1,1\}$.

\smallskip
In what follows we implement the interface using this implementation. For (1) and (2), note that we can easily read off $p(v)$ and $\weight_t(v)$ from the attributes $v.p$ and $v.\weight$ stored in the object~$v$ in time $O(1)$. 
For (3), we can compute the set $S = \set(v)$ represented by $v$ by following the $\vleft/\vright$ pointers starting at $v$ to determine the set of reachable virtual items $R(v)$ and then printing all encountered items $u.\vitem$ for all $u \in R(v)$. This terminates because there are no pointer cycles, it computes $S$ because $S$ is the disjoint union of all sets $u.\vitem$ over all $u \in R(v)$, and it runs in time $O(|R(v)|) = O(|S|+1)$.

For (i), to construct the virtual item $v = v_\emptyset$ we simply set $v.p = 0$, $v.\weight = (0,\ldots,0)$, $v.\vitem = \emptyset$, and $v.\vleft = v.\vright = \textup{NULL}$.
For (ii), to construct the virtual item $v = v_{\{i\}}$ for a given item $i$, we set $v.p = p(i)$, $v.\weight = \weight(i)$, $v.\vitem = \{i\}$, and $v.\vleft = v.\vright = \textup{NULL}$.

Note that we can check whether a virtual item represents the empty set in time $O(1)$, since any such item can reach $|R(v)| \le 1$ virtual items, that is, $v$ represents the empty set if and only if $v.\vitem = \emptyset$ and $v.\vleft = v.\vright = \textup{NULL}$.

It remains to implement (iii).
For two virtual items $v_1,v_2$ representing disjoint sets, we construct their union $v = v_1 \cup v_2$ as follows. If $v_1$ represents the empty set, then we copy $v_2$ to $v$. Else if $v_2$ represents the empty set, then we copy $v_1$ to $v$. Otherwise, we construct a new virtual item $v$ with $v.p \coloneqq v_1.p + v_2.p$, $v.\weight \coloneqq v_1.\weight + v_2.\weight$, $v.\vitem \coloneqq \emptyset$, $v.\vleft \coloneqq v_1$, and $v.\vright \coloneqq v_2$.

This construction clearly ensures that there are no pointer cycles, and that the set $S = \set(v_1) \cup \set(v_2)$ represented by $v$ is the disjoint union of all sets $u.\vitem$ over all $u \in R(v)$. Finally, by our treatment of the cases where $v_1$ or $v_2$ is the empty set, it follows that if $v$ represents the empty set then $|R(v)| = 1$, and if $v$ represents a non-empty set $S$ then $R(v)$ forms a binary tree whose leaves are in one-to-one correspondence with $S$ and thus $|R(v)| = 2|S|-1$. Thus, we ensured the size bound on $R(v)$. This finishes the implementation details.

\subsection{Representative Solutions}
\label{sec:dalgo_representativesols}
\label{sec:low_profit}

We next present a highly optimized algorithm for solving instances with slack. Compared to the simple dynamic programming algorithm from \Cref{lem:ddim_dp}, the running time is improved from $(n/\eps)^{O(d)}$ to $n \cdot (\log(n)/\eps)^{O(d)}$, by adopting techniques of a modern Knapsack approximation algorithm~\cite{Chan18a}. The setup is also somewhat changed, in order to fit perfectly to the rest of our $d$-Knapsack algorithm. The proof of \cref{thm:simple-apx-data-structure} is deferred to \cref{sec:dalgo_proofapxlem}.

\begin{restatable}[Approximation Lemma; Solving Instances with Slack by Convolution]{lemma}{APXLemma} \label{thm:simple-apx-data-structure}\label{lem:ddim_dpv}
	Given a set~$I$ of $n$ items, weight estimates $\capacity_2,\ldots,\capacity_{d-1} \in \mathbb{R}_{\ge 0}$, a profit estimate $\tp \in \mathbb{R}_{\ge 0}$, and error parameters $\eps, \eta \in (0,1]$, in time $n\cdot \left(\log(n/\eta)/\eps\right)^{O(d)}$ we can compute a set $\mathcal{L}$ of virtual items with $\set(v) \subseteq I$ for every $v \in \mathcal{L}$ and the following property.
	Write
	\[
	\mathcal{Q} \coloneqq \left\{ Q \subseteq I \, \middle| \,p(Q) \ge \eta \, \tp,\, \weight_t(Q) \le \eta^{-1} \capacity_t \textnormal{ for each } t \in \{2,\ldots,d-1\} \right\}.
	\]
	Then for every $Q \in \mathcal{Q}$, there exists a virtual item $v \in \mathcal{L}$ that satisfies:
	\begin{align*}
		p(v)& \geq \min\left\{(1-\eps)\, p(Q), \eta^{-1} \tp\right\},\nonumber\\
		\weight_1(v)& \leq (1+\eps)\, \weight_1(Q),\nonumber\\
		\weight_{d}(v) & \leq   \weight_d(Q),\text{ and }\nonumber\\
	\weight_t(v)& \leq \max\left\{(1+\eps)\, \weight_t(Q), \eta \, \capacity_t\right\}\text{ for each } 1 < t < d.\\
	\end{align*}
\end{restatable}

Next, we use the approximate solutions provided by \Cref{thm:simple-apx-data-structure} to generate a list of representative solutions, which we will use as surrogates for the low-profit items in our algorithm. 
For simplicity, we assume that the items in $I$ are indexed by $1,\ldots,n$, that is, $I = \{1,\ldots,n\}$.
As usual, we also assume that items are sorted by non-increasing profits, that is, $p(1) \ge \ldots \ge p(n)$.

\begin{lemma}[Representative Solutions]
	\label{lem:low_profit}
	Given a set of items $I = \{1,\ldots,n\}$ sorted by non-increasing profits and an error parameter
	$0<\eps<1$, in time
	$n^{d} \cdot \big(\frac{ \log n }{\eps}\big)^{O(d)}$ we can compute sets of virtual items $\cL_1,\ldots,\cL_{n+1}$ such that for every $i \in [n+1]$:
	\begin{itemize}
		\item $\abs{\cL_i}\le n^{d-1} \cdot  \big(\frac{  \log n}{\eps}\big)^{O(d)}$,
		\item For every $v\in \cL_i$ it holds that $\set(v) \subseteq \{i,\ldots,n\}$,
		\item $v_\emptyset \in \cL_i$,
		\item For every $S\subseteq \{i,\ldots, n\}$ there exists
		$v\in\cL_i$ such that $\weight(v)\le \weight(S)$, 
		\[
		\begin{aligned}
		p(v)\,&\ge\, (1-\eps)\cdot p(S) - \left(d-1+\frac{\eps}{n}\right)\cdot p(i), \quad \textnormal{ and } \\
		p(v)\,&\ge\, \frac{1-\eps}{d} \cdot p(S) - \frac{\eps}{n} \cdot p(i),
		\end{aligned} 
		\]
		where we interpret $p(n+1) = 0$.
	\end{itemize}
\end{lemma}

The lemma shows that for any set of items $S$, among the representative solutions one can find a surrogate virtual item $v$ that can replace $S$ while preserving feasibility and incurring only a small loss in profit.

\begin{proof}[Proof of \Cref{lem:low_profit}]
The pseudocode of our algorithm is given in \Cref{alg:low_profit}. The loop in \Cref{lowprofit:loop} has $n^{d-1}$ iterations. The running time of each iteration is dominated by the call to \Cref{lem:ddim_dpv}, which takes time $n\cdot (\log(n/\eta)/\eps')^{O(d)}=n\cdot (\log(n)/\eps)^{O(d)}$.  Thus, the total running time of \Cref{alg:low_profit} is $n^d\cdot (\log(n)/\eps)^{O(d)}$.

Note that the algorithm computes $\cL_{n+1} = \{ v_\emptyset \}$, and this set satisfies all claims. Thus, in what follows we can restrict our attention to $i \in [n]$.

We next bound the size of the sets returned by the algorithm. Since the output of each call to the procedure of \Cref{lem:ddim_dpv} is explicitly constructed within its stated running time, each such call returns at most $n\cdot (\log(n/\eta)/\eps')^{O(d)}=n\cdot (\log(n)/\eps)^{O(d)}$ virtual items. For each $i\in I$, the set $\cL_i$ initially contains the $n-i+2$ items added in \Cref{lowprofit:init}. 
It is then extended by the outputs of $n^{d-2}$ calls to  \Cref{lem:ddim_dpv}, one for each choice of $i_2,\ldots,i_{d-1}$.
 Consequently, $\lvert\cL_i\rvert\leq n^{d-2}\cdot n\cdot (\log(n)/\eps)^{O(d)}=n^{d-1}\cdot (\log(n)/\eps)^{O(d)}$. Therefore, \Cref{alg:low_profit} runs in time $n^d\cdot (\log(n)/\eps)^{O(d)}$ and returns sets $\cL_1,\ldots,\cL_{n+1}$, each containing at most $n^{d-1}\cdot (\log(n)/\eps)^{O(d)}$ virtual items.
Also, by construction, for every $i\in I$ and $v\in \cL_i$ we have $\set(v) \subseteq\{i,i+1,\ldots, n\}$, and $v_\emptyset \in \cL_i$ for all $i$.

\begin{algorithm}[t]
	\caption{$\lowprofit(I = \{1,\ldots,n\},\eps)$}
	\label{alg:low_profit}
	
	$\eps' \leftarrow \frac{\eps}{10}$
	
	$\cL_i \gets \{ v_{\emptyset }, v_{\{i\}},v_{\{i+1\}},\ldots, v_{\{n\}}\}$ for all $i\in [n+1]$ \label{lowprofit:init}

	\For{each $i\in I$ and each  $i_2,\ldots,i_{d-1} \in I$\label{lowprofit:loop}}{
		
		$I' \gets \{i'\in I \,|\, i' \ge i \}$ \label{lowprofit:IR}
		
		$\tp \gets  p(i)$
		
		$\capacity_t \gets  \weight_t(i_t)$ for all $t \in \{2,\ldots, d-1\}$
		
		Apply the procedure from \Cref{lem:ddim_dpv} to $I', \capacity_2,\ldots,\capacity_{d-1}, \tp$, and error parameters $\eps'$ and $\eta \coloneqq \frac \eps n$; let $\tilde{\cL}$ be the result \label{lowprofit:DP}
		
		$\cL_i \leftarrow \cL_i \cup \tilde{\cL}$

	}
	
	Return $\cL_1,\ldots, \cL_{n+1}$ 
		
\end{algorithm}

We proceed to prove  the last property of the lemma. 
Let $i\in I$ and $S\subseteq \{i,i+1,\ldots, n\}$. 
	If $S = \emptyset$ then it is covered by $v_\emptyset \in \cL_i$. 
	Otherwise, by \Cref{lem:structural} (with error parameter $\eps'$) there exists a set $Q\subseteq S$ such that
\begin{equation}
	\label{eq:Q_weight}
	\forall t\in \{1,\ldots,d-1\}:\qquad \weight_t(Q) \leq  (1-\eps') \cdot \weight_t(S), \end{equation}
as well as   \begin{equation}
	\label{eq:Q_profit}
	p(Q)\geq (1-\eps') \cdot p(S) - (d-1)\cdot \max_{i' \in S} p(i') \geq (1-\eps') \cdot p(S) - (d-1)\cdot p(i)\end{equation}
and 
\begin{equation}
	\label{eq:Q_profit_caprara} 
	\max\big\{ p(Q), \max_{i' \in S} p(i') \big\} \ge \frac {1-\eps'}d \cdot p(S).
\end{equation}
Note that since $Q\subseteq S$ we also have $\weight_d(Q)\leq \weight_d(S)$.  Consider the following cases.  

\paragraph*{Case 1: $p(Q) =\max\left\{p(Q), \max_{i'\in S} p(i')\right\}$ and $p(Q) > \frac{\eps}{n}\cdot p(i)$.} 	
	For every $t\in \{2,\ldots, d-1\}$ let $i_t = \argmax_{i'\in Q} \weight_t(i')$ be an item of maximum weight in dimension~$t$. 
		We focus on the iteration of  \Cref{lowprofit:loop} that considers $i$ and  $i_2,\ldots, i_{d-1}$. 
		By assumption, $p(Q) > \frac{\eps}{n}\cdot p(i) = \eta \, \tp$. 
		Also, $\weight_t(Q)\leq \abs{Q}\cdot \weight_t(i_t) \leq n \cdot \capacity_t \leq \eta^{-1} \capacity_t$ for every $t\in \{2,\ldots, d-1\}$.
		Therefore, by  \Cref{lem:ddim_dpv} there is $v\in \tilde{\cL}$ such that 
		\begin{equation}
			\label{eq:v_basics}
			\begin{aligned}
			p(v)& \geq \min\left\{(1-\eps')\, p(Q), \eta^{-1} \tp\right\},\\
			\weight_1(v)& \leq (1+\eps')\, \weight_1(Q),\\
			\weight_{d}(v) & \leq   \weight_d(Q),\text{ and }\\
			\weight_t(v)& \leq \max\left\{(1+\eps')\,  \weight_t(Q), \eta \, \capacity_t\right\}\text{ for each } 1 < t < d.
			\end{aligned}
		\end{equation}
		For the first dimension this implies that 
		$$\weight_1(v)\overset{\eqref{eq:v_basics}}{\leq} (1+\eps') \cdot \weight_1(Q)\overset{\eqref{eq:Q_weight}}{\leq} (1+\eps')(1-\eps') \cdot \weight_1(S) \leq \weight_1(S). $$ 
		
		For all $t\in \{2,\ldots, d-1\}$, as $i_t\in Q$ we have $\weight_t(Q) \geq \weight_t(i_t)=\capacity_t$ and therefore
		$$
		\weight_t(v)\leq \max\left\{  (1+\eps')\cdot \weight_t(Q), \eta \capacity_t\right\}  =(1+\eps') \cdot \weight_t(Q)\overset{\eqref{eq:Q_weight}}{\leq} (1+\eps')(1-\eps') \cdot \weight_t(S) \leq \weight_t(S).
		$$
		Since it also holds that $\weight_d(v)\leq \weight_d(Q) \leq \weight_d(S)$, we obtain $\weight(v)\leq \weight(S)$. 
		By \eqref{eq:v_basics} we have
		$$
			p(v) \geq \min\left\{(1-\eps')\, p(Q), \eta^{-1} \tp\right\} = (1-\eps')\,p(Q),
		$$
		where the equality holds as $p(Q)\leq |Q|\cdot p(i) \leq n\cdot p(i) = n \cdot \tp \le \eta^{-1} \tp$. 
		Combining the above with  \eqref{eq:Q_profit} and~\eqref{eq:Q_profit_caprara}, we get that 
		$$
		p(v) \geq (1-\eps') \cdot \left( (1-\eps') \cdot p(S)-(d-1)\cdot p(i) \right)  \geq (1-\eps)\cdot p(S)- \left(d-1+\frac{\eps}{n}\right)\cdot p(i)$$
		and 
		$$
			p(v) \geq (1-\eps') \cdot \frac{ 1-\eps'}{d} p(S)   \geq \frac{1-\eps}{d} p(S) -\frac{\eps}{n}\,p(i).
		$$
		That is, we showed that the last property of the lemma holds in this case.

		\paragraph*{Case 2: $p(Q)= \max\{ p(Q), \, \max_{i'\in S} p(i')\}$ and $p(Q)\leq \frac{\eps}{n}\cdot p(i)$.} By construction  $v_{\emptyset}\in \cL_i$ and it holds that 
		$$
		p(v_{\emptyset}) =0 \geq p(Q)-\frac{\eps}{n}\cdot p(i) \overset{\eqref{eq:Q_profit}}{\geq} (1-\eps )\cdot p(S) -\left( d-1 +\frac{\eps}{n}\right) \cdot p(i),
		$$
		and similarly,
		$$
		p(v_{\emptyset}) =0 \geq p(Q)-\frac{\eps}{n}\cdot p(i) =\max\{ p(Q), \, \max_{i'\in S} p(i')\} - \frac{\eps}{n}\cdot p(i) \overset{\eqref{eq:Q_profit_caprara}}{\geq} \frac{1-\eps}{d}\cdot p(S) -\frac{\eps}{n}\cdot p(i). 
		$$
		Finally,  it trivially holds that $\weight(v_{\emptyset} ) \leq \weight(S)$.
		That is, the lemma holds in this case. 
\paragraph*{Case 3: $p(Q) \neq \max\{ p(Q), \, \max_{i'\in S} p(i')\}$.} Then there exists an item $i^*\in S$ such that $p(i^*)= 	\max\big\{ p(Q), \max_{i' \in S} p(i') \big\}$. By construction we have $v_{\{i^*\}}\in \cL_i$ (see \Cref{lowprofit:init}). Furthermore,
$$
p(v_{\{i^*\}}) = p(i^*) = \max\big\{ p(Q), \max_{i' \in S} p(i') \big\} \overset{\eqref{eq:Q_profit_caprara}}{\geq } \frac{1-\eps'}{d} \cdot p(S) \geq \frac{1-\eps}{d} \cdot p(S) -\frac{\eps}{n}\cdot p(i)
$$
and 
$$
\begin{aligned}
	p(v_{\{i^*\}}) = p(i^*) &= \max\big\{ p(Q), \max_{i' \in S} p(i') \big\} \ge p(Q) \\
	&\overset{\eqref{eq:Q_profit}}{\geq }  (1-\eps') \cdot p(S) - (d-1)\cdot p(i)
	\\
	&\geq  (1-\eps) \cdot p(S) - \left(d-1+\frac{\eps}{n}\right)\cdot p(i).
\end{aligned}
$$
Furthermore, as $i^*\in S$ we also have $\weight(v_{\{i^*\}})=\weight(i^*)\leq \weight(S)$. That is, the last property of the  lemma holds in this case as well.
Overall, we showed that the last property holds in all cases, which completes the proof.
\end{proof}

\subsection{Meet in the Middle via Orthogonal Range Searching}
\label{sec:dalgo_ORS}

We say that two sets of items $S_1,S_2 \subseteq I$ can be \emph{matched} if $\weight(S_1) + \weight(S_2) \le \capacity$ and $\max(S_1) < \min(S_2)$. 
That is, we require that the total weights of $S_1$ and $S_2$ do not exceed the capacities, and the maximum index of an item in $S_1$ is less than the minimum index of an item in $S_2$.
The former ensures that $S_1 \cup S_2$ is feasible, while the latter ensures that $S_1$ and $S_2$ are disjoint, and thus $p(S_1 \cup S_2) = p(S_1) + p(S_2)$. 

Next we write this notion of \emph{matching} in a geometric language. To this end, we replace $S_1$ by a point $\embed(S_1)$, and we replace $S_2$ by a box $\range_{\capacity}(S_2)$ such that $\embed(S_1) \in \range_{\capacity}(S_2)$ if and only if $S_1$ and $S_2$ can be matched.

\begin{definition}[Embedding and Range]
	\label{def:embedding}
	The {\em embedding} $\embed(S)$ of $S\subseteq I$ is the point $\bar{x}\in (\mathbb{R} \cup \{-\infty\})^{d+1}$ defined by
	\begin{itemize}
		\item  $\bar{x}_t= \weight_t(S)$ for every $1\leq t\leq d$ and
		\item  $\bar{x}_{d+1} =\max(S)$.
	\end{itemize}
	
	\noindent
	The {\em range} $\range_{\capacity}(S)$ of $S\subseteq I$ is the $(d+1)$-dimensional box $R=[0,b_1]\times \ldots \times [0,b_d] \times [-\infty,b_{d+1}]$ defined by
	\begin{itemize}
		\item   $b_t=\capacity_t-\weight_t(S)$ for every $1\leq t\leq d$ and
		\item  $b_{d+1} =\min(S) -1$.
	\end{itemize}
\end{definition}

Observe that in case $S$ is not feasible (i.e., $\weight_t(S)>\capacity_t$ for some $t\in \{1,\ldots, d\}$) then its range is empty. One can easily observe the following.
\begin{lemma}
	\label{lem:embed_to_range}
  For any sets of items $S_1,S_2 \subseteq I$, we have $\embed(S_1) \in \range_{\capacity}(S_2)$ if and only if $S_1$ and $S_2$ can be matched (i.e., $\weight(S_1) + \weight(S_2) \le \capacity$ and $\max(S_1) < \min(S_2)$).
\end{lemma}

This geometric setup allows us to use orthogonal range searching from computational geometry. 

\begin{lemma}[Orthogonal Range Searching (ORS) \cite{BergCKO08}]
	\label{lem:kvors}
	For every $d\in \mathbb{N}$, 
	given a set $\cC$ of $n$ objects,
	where each object $S\in\cC$ is associated with a point
	$\bar{x}_S\in(\mathbb{R} \cup \{-\infty,\infty\})^d$ and a profit  $p_S\in\mathbb{R}$, in
	time $O(n\log^{d-1} n)$ we can build a data structure that supports the following query in time $O(\log^d n)$:
	Given an axis-aligned range
	$R = [a_1,b_1]\times\cdots\times[a_d,b_d]$,
	the query returns an object $S\in\cC$ such that
	$\bar{x}_S\in R$ and $p_S$ is maximized among all such objects.
	If no object $S \in \cC$ satisfies $\bar{x}_S\in R$, the query returns~$\bot$.
\end{lemma}

This data structure can be used to implement meet in the middle for $d$-Knapsack along the following lines.
Let ${\cal S}_1, {\cal S}_2$ be some families of sets of items.
Build an ORS data structure containing for each $S_1 \in {\cal S}_1$ the object $S_1$ with point $\embed(S_1)$ and profit $p(S_1)$. 
Then for any $S_2 \in {\cal S}_2$ issue the query $\range_{\capacity}(S_2)$. This returns a set $S_1 \in {\cal S}_1$ maximizing $p(S_1)$ subject to the constraint that $S_1$ can be matched with $S_2$. That is, in time $O(\log^{d+1} |{\cal S}_1|)$ we obtain the optimal match for $S_2$ among all sets in ${\cal S}_1$. 
In particular, we can determine the pair $(S_1,S_2) \in {\cal S}_1 \times {\cal S}_2$ that can be matched and maximizes the total profit in time $O((|{\cal S}_1| + |{\cal S}_2|) \log^{d+1} |{\cal S}_1|)$, which significantly improves over naively checking all pairs of sets in ${\cal S}_1$ and ${\cal S}_2$.

\subsection{Fast Algorithm for \boldmath$d$-Knapsack}
\label{sec:dalgo_combine}
\label{sec:ddim_mitm}

The main bottleneck of the simple algorithm $\simpleddim$ (\Cref{alg:ddim_simple}) is the
enumeration of the $\big\lceil \frac{d-1}{\eps} \big\rceil$ most profitable items
of an optimal solution. We reduce this cost by adopting a meet-in-the-middle strategy, which integrates the geometric language presented in \Cref{sec:dalgo_ORS} together with the representative solutions from \Cref{sec:dalgo_representativesols}. Conceptually, our meet-in-the-middle approach matches one of the representative  solutions and the set $\OPT[q-\ell_2+1:q]$ on the one side with the set $\OPT[\semi q-\ell_2]$ on the other side, where the values of $q$ and $\ell_2$ are carefully selected to ensure the running time and approximation ratio. See \Cref{alg:mitmddim} for detailed pseudocode of our algorithm.

\begin{algorithm}[t]
	\caption{$\mitmddim(I=\{1,\ldots,n\},\capacity,\eps,\rho)$}
	\label{alg:mitmddim}
	
	Sort items by non-increasing profits, that is, $p(1) \ge \ldots \ge p(n)$.
	
	Define $q \gets \ceil{\frac{d-1}{\eps} -d +\rho}$, $\ell_1 \gets \floor{  \frac{q + d-1}{2}}$, and $\ell_2 \leftarrow \max\{q-\ell_1, 1\}$ \label{mitm:ell} 
	
	Initialize an ORS data structure (\Cref{lem:kvors}) storing an object $S\subseteq I$ with point $\bar x_S \coloneqq \embed(S)$ and profit $p_S \coloneqq p(S)$ for every $S\subseteq I$ with $\abs{S}\leq \ell_1$. \label{mitm:ors}

	$\cL_1,\ldots,\cL_{n+1} \leftarrow \lowprofit(I,\delta)$ where $\delta =  \frac{\eps^2\cdot \rho }{2\cdot (d-1)}$ (\Cref{lem:low_profit}).\label{mitm:low}
	
	$\best \leftarrow\emptyset$.

	\For{every $S_2 \subseteq I$ such that $1\leq \abs{S_2}\leq \ell_2$ and every $v\in \cL_i$ where $i=\max(S_2) +1$\label{mitm:loop}}{
		Issue a query to the ORS with $\range_{\capacity}( \set(S_2\cup\{v\}))$; let $S_1$ be the returned object.
		\label{mitm:query}
		
		If $S_1\neq \bot$ and $p(S_1\cup S_2\cup \{v\})>p(\textsf{\best})$ then update $\best \leftarrow S_1\cup S_2 \cup \{v\}$.  
		\label{mitm:update}
	}
	\Return $\set(\best)$ 
\end{algorithm}

\begin{lemma}
	\label{lem:mitm_approx}
	For every integer $d\geq 2$, $\rho \in (0,1)$, and  $\eps \in \left(0,\frac{d-1}{d}\right]$, 
	\Cref{alg:mitmddim} is a $(1-\eps)$-approximation algorithm for $d$-Knapsack.
\end{lemma}
\begin{proof}
	\Cref{alg:mitmddim} works as follows. We assume that the items in $I$ are indexed by $1,\ldots,n$, that is, $I = \{1,\ldots,n\}$, and we sort them by non-increasing profits. We build an ORS data structure storing the embedding of all sets of size at most $\ell_1$. We compute representative solutions $\cL_1,\ldots,\cL_{n+1}$ by \Cref{lem:low_profit}. Then we iterate over all sets $S_2$ of size in $[1,\ell_2]$ and all representative solutions $v \in \cL_{\max(S_2)+1}$. For each pair $(S_2,v)$, we query the ORS data structure to find the best match $S_1$, and we maintain the best found solution $S_1 \cup S_2 \cup \{v\}$. Throughout, we work on virtual items in their compressed form, only in the last line we expand the virtual item into a set, so that the algorithm returns a set. In the remainder we analyze this algorithm.

	We first observe that the algorithm always returns a feasible solution.
	Specifically, the algorithm maintains the invariant that  $\set(\best)$ is always feasible. Initially, $\best=\emptyset$, and therefore $\set(\best) = \emptyset$ which is feasible. 
Whenever  $\best$ is updated in  \Cref{mitm:update}, it  is set to $S_1\cup S_2\cup \{v\}$, where $\embed(S_1)\in \range_{\capacity}(\set(S_2\cup \{v\}))$. Therefore, $\set(S_1\cup S_2\cup\{v\})= S_1\cup \set(S_2\cup \{v\})$ is feasible  by \Cref{lem:embed_to_range}.
Thus, $\set(\best)$ remains feasible throughout the execution of the algorithm, and the algorithm returns a feasible solution.

	Let $\OPT$ be an optimal solution for the given $d$-Knapsack instance $(I,\capacity)$.
	 Since $\eps\leq \frac{d-1}{d}$ it holds that 
	 $q\geq \frac{d-1}{\eps}-d+\rho \geq \rho >0$, and therefore $q\geq 1$.  
	Consider the following three cases.
	 \begin{itemize}
		\item Case 1: $\abs{\OPT}>q$. 
	Let $H=\OPT[\semi q]$ be the  set of $q$ most profitable items in $\OPT$ and let $L=\OPT[q+1 \semi]=\OPT\setminus H$ be the set of low-profit items in $\OPT$. 
	Define 
	$\tS_1= \OPT[1:q-\ell_2]$ and $\tS_2=\OPT[q-\ell_2+1:q]$.   Clearly, $\tS_1\cup \tS_2 =H$.   Since $q,\ell_2\geq 1$ we  have that $\tS_2\neq \emptyset$. 
	Also, let $i=\max(\tS_2)+1$ and note that $L\subseteq \{i,i+1,\ldots,n\}$. By \Cref{lem:low_profit} there exists $v\in \cL_{i}$ such that $\weight(v)\leq \weight(L)$, 
	\begin{equation}
		\label{eq:v_profit} 
	\begin{aligned}
		p(v)\,&\ge\ (1-\delta)\cdot p(L)\ -\ \left(d-1+\frac{\delta}{n}\right)\cdot p(i) \qquad \textnormal{ and } \\
		p(v)\,&\ge \frac{1-\delta}{d} \cdot p(L) - \frac{\delta}{n} \cdot p(i).
	\end{aligned} 
	\end{equation}
	
	Consider the iteration of  \Cref{mitm:loop}  in which $S_2=\tS_2$ and $v$ is the virtual item defined above. It holds that $\weight(\tS_1\cup \tS_2 \cup\{v\})\leq \weight(H) +\weight(L) \leq \capacity$ and $\max(\tS_1) <\min(\set(\tS_2\cup\{v\}))$, therefore the sets 
	 $\tS_1$ and $\set(\tS_2\cup \{v\})$ can be matched. 
	 By \Cref{lem:embed_to_range} this implies that $\embed(\tS_1)\in \range_{\capacity}(\set(\tS_2\cup\{v\}))$. Since $|\tS_1| = q-\ell_2 \leq \ell_1$, the set $\tS_1$ is added to the ORS data structure in \Cref{mitm:ors}. Therefore, the query in \Cref{mitm:query} returns a set $S_1 \subseteq I$ such that $\embed(S_1)\in \range_{\capacity}(\set(S_2\cup\{v\}))$ and $p(S_1)\geq p(\tilde{S_1})$, and thus by \Cref{lem:embed_to_range} it holds that 
	 \begin{equation}
	 	\label{eq:HV_lowerbound}
	 p(S_1\cup S_2\cup\{v\}) = p(S_1)+p(S_2)+p(v)  \geq p(\tS_1)+p(\tS_2) + p(v) = p(H)+p(v).
		\end{equation}
	Since the items are sorted by profit it holds that $p(i) \leq p(\max(S_2)) =p(\OPT[q]) \leq \frac{p(\OPT[\semi q])}{q} = \frac{p(H)}{q}$.  Combining this with \eqref{eq:v_profit} yields
	 	$$
	 	\begin{aligned}
	 			p(v)\,&\geq (1-\delta)\cdot p(L) -(d-1+\delta ) \cdot \frac{p(H)}{q} \qquad \textnormal{ and } \\
	 	p(v)\,&\ge \frac{1-\delta}{d} \cdot p(L) - \delta \cdot \frac{p(H)}{q}.
	 	\end{aligned}
	 	$$
	 	Let $\alpha = \frac{q}{q+d}\in (0,1)$. From the above inequalities we can deduce that 
	 	\begin{align*}
	 	p(v) &= \alpha \cdot p(v) + (1-\alpha) \cdot p(v)  \\
	 	&\geq p(L)\cdot \left(\alpha (1-\delta) + (1-\delta)\cdot \frac{1-\alpha}{d} \right) - p(H)\cdot \left(\frac{\alpha}{q} \cdot (d-1) + \frac{\delta}{q}  \right) \\
	 	&\geq p(L)\cdot \left(\alpha  + \frac{1-\alpha}{d} - \delta \right) - p(H)\cdot \left(\frac{\alpha}{q} \cdot (d-1) + \delta  \right).
	 	\end{align*}
	 	
	 	By the above inequality and \eqref{eq:HV_lowerbound} we have,
	 \[
	 \begin{aligned}
	 	 p(S_1\cup S_2\cup\{v\})&\geq p(H)+p(v)\\
	 	 &\geq p(H)\cdot \left( 1-\frac{\alpha}{q}\cdot (d-1) -\delta\right) +p(L)\cdot \left( \alpha +  \frac{1-\alpha}{d}  -\delta \right) \\
	 	 &= p(H)\cdot \left(1-\frac{d-1}{q+d} -\delta\right) + p(L)\cdot \left(1- \frac{d-1}{q+d}-\delta \right)\\
	 	 &= p(\OPT) \cdot \left( 1-\frac{d-1}{q+d}-\delta\right),
	 \end{aligned}
	 \]
	 where the first equality follows from $\alpha =\frac{q}{q+d}$. By the definitions of $q$  and $\delta$ we also have that 
	 $$
	 \frac{d-1}{q+d} +\delta \leq \frac{d-1}{ \frac{d-1}{\eps} -d+\rho +d} +\frac{\rho\cdot \eps^2}{2(d-1)} \leq \eps\cdot \left(\frac{d-1}{d-1+
	 \rho \cdot \eps} + \frac{\rho \cdot \eps}{d-1+\rho\cdot \eps} \right)  \leq \eps,
	 $$
	 where the second inequality holds as $2(d-1)\geq d-1+\rho\cdot \eps$, by $d \ge 2$ and $\eps,\rho \in (0,1)$. 
	 Therefore,
	 $$
	 p(S_1\cup S_2 \cup\{v\}) \geq p(\OPT) \cdot \left( 1-\frac{d-1}{q+d}-\delta\right)  \geq p(\OPT)\cdot (1-\eps). 
	 $$ 
	  
	  That is, following this iteration it  holds that $p(\best) \geq (1-\eps) \cdot p(\OPT)$, and we can conclude that the algorithm returns a $(1-\eps)$-approximation for $d$-Knapsack. 
	  
	  We remark that the weighted averaging by $\alpha$ is essentially the argument used by Caprara et al.~\cite{CapraraKPP00} in their factor-$n$ improvement in the previous state-of-the-art algorithm.

	\item Case 2: $1 \le \abs{\OPT} \leq q$. Then 
	let $r = \min\{\ell_1, \abs{\OPT}-1\}$ and define $\tS_1 = \OPT[\semi r]$ 
	 and  $\tS_2 = \OPT[r+1 \semi]$. 
	 Since $r\leq \abs{\OPT}-1$ and $\abs{\OPT} \ge 1$ it follows that $\tS_2 = \OPT[r+1 \semi]\neq \emptyset$. Furthermore, if $r=\abs{\OPT}-1$ we have $|\tS_2| =1\leq \ell_2$, and otherwise $r=\ell_1$ which implies $|\tS_2| =\abs{\OPT }-\ell_1 \leq q-\ell_1 \leq \ell_2$. Therefore, $|\tS_2|\leq \ell_2$ in both cases. 

	 Consider the iteration of \Cref{mitm:loop} in which $S_2=\tS_2$ and $v=v_{\emptyset}$. Note that $v_{\emptyset} \in \cL_i$ by \Cref{lem:low_profit}.   As $|\tS_1|=r\leq \ell_1$ the set $\tS_1$ is  added to the ORS structure in \Cref{mitm:ors}.  Since $\tS_1$ and $\tS_2$  can be matched we have that  $\embed(\tS_1) \in \range_{\capacity}(\tS_2) = \range_{\capacity}(\set(\tS_2\cup \{v_{\emptyset}\}))$, and therefore the ORS query in \Cref{mitm:query} returns a set $S_1$ such that $S_1$ and $\set(\tS_2\cup\{v_{\emptyset}\})=\tS_2$ can be matched and $p(S_1) \geq p(\tS_1)$. Therefore, $p(S_1\cup S_2\cup\{v_{\emptyset}\})\geq p(\tS_1 )+p(\tS_2) = p(\OPT)$.
		  That is, following this iteration it holds that $p(\textsf{\best}) \geq p(\OPT)$, and in particular the algorithm returns a $(1-\eps)$-approximate solution. 
		  
	\item Case 3: $|\OPT| = 0$. Then the initial $\best \leftarrow\emptyset$ already is an optimal solution.
	\end{itemize}
\end{proof}

Having established the approximation ratio, we now turn to the running time of the algorithm.
\begin{lemma}
	\label{lem:runtime} For every $\eps, \rho \in (0,1)$, \Cref{alg:mitmddim} runs in time $\left( n^{\ceil{\frac{d-1}{2\eps} -\frac{1}{2} +\rho}}+n^{d} \right)\cdot \big(\frac{d \log n}{\rho\cdot \eps}\big)^{O(d)}$.
\end{lemma}
\begin{proof}

	The initialization of the ORS data structure in \Cref{mitm:ors} inserts $O(n^{\ell_1})$ objects, where each set is stored as a vector of dimension $d+1$. Therefore,   by \Cref{lem:kvors}, the construction of the data structure takes time 
	$$
O\big( n^{\ell_1}	 \cdot \log^{d} \big(n^{\ell_1}\big) \big) = O\big(\ell_1^d\cdot n^{\ell_1} \cdot\log^{d} n   \big) \\
		=n^{\ell_1}\cdot \left(\frac{d \log n}{\eps}\right)^{O(d)}.
 $$

	Next, the algorithm generates the  sets of representative solutions  $\cL_1,\ldots, \cL_{n+1} $ in \Cref{mitm:low}  using \Cref{lem:low_profit} with error parameter $\delta = \frac{\eps^2\cdot \rho }{2\cdot(d-1)}$. This takes time
	$n^{d} \cdot \big(\frac{ \log n }{\delta}\big)^{O(d)}=
	n^{d} \cdot \big( \frac{d \log n }{\rho\cdot \eps}\big)^{O(d)}$. 
	
	Finally, consider the loop in \Cref{mitm:loop}. \Cref{mitm:query} requires computing $\range_{\capacity}( \set(S_2\cup\{v\}))$, which naively takes time $O(dn)$ as $\set(v)$ can have size up to $\Omega(n)$. More cleverly, by observing that \Cref{mitm:loop} ensures $\min(\set(v)) > \max(S_2)$ it holds that $\min(\set(S_2\cup\{v\})) = \min(S_2)$ and $\weight_t(S_2\cup\{v\}) = \weight_t(S_2) + \weight_t(v)$ and thus 
	\begin{align*} 
	&\range_{\capacity}( \set(S_2\cup\{v\})) \\
	&= [0, \capacity_1 - \weight_1(S_2) - \weight_1(v)] \times \ldots \times [0, \capacity_d - \weight_d(S_2) - \weight_d(v)] \times [-\infty, \min(S_2) - 1].
	\end{align*}
	Since $\weight_t(v)$ can be read from $v$ in time $O(1)$, and $|S_2| = O(d/\eps)$, this range can be computed in time $O(d^2/\eps)$. Similarly, the evaluation of $p(S_1\cup S_2\cup\{v\})$ in \Cref{mitm:update} can be implemented in time $O(d^2/\eps)$, as it is equal to $p(S_1) + p(S_2) + p(v)$. 
	Therefore each iteration takes time $O(d^2/\eps) $ for the computation of the embedding and update operations, and $O(\log^{d+1}(n^{\ell_1}))= O\big( \ell_1^{d+1} \cdot \log^{d+1} n \big) =\big( \frac{ d \log n }{\eps}\big)^{O(d)}$ for the ORS query.
	As $\abs{\cL_i }\leq n^{d-1}\big(\frac{\log n}{\delta}\big)^{O(d)}=n^{d-1}\big(\frac{\log n}{\rho\cdot \eps }\big)^{O(d)}$ for all $i\in \{1,\ldots,n+1\}$, 
	the number of iterations of  \Cref{mitm:loop} is 
	$$
	n^{\ell_2} \cdot n^{d-1} \cdot  \left(\frac{d\cdot \log n}{\rho\cdot \eps }\right)^{O(d)}=  n^{\ell_2+d-1}\cdot \left(\frac{d\cdot \log n}{\rho\cdot \eps }\right)^{O(d)}
	$$
	Overall, the total running time of  the loop in \Cref{mitm:loop} is 
	$$
n^{\ell_2+d-1}\cdot \left(\frac{d\cdot \log n}{\rho\cdot \eps }\right)^{O(d)} \cdot \left( O\left(\frac{d^2}{\eps}\right)+ \left( \frac{ d \log n }{\eps}\right)^{O(d)} \right)
=n^{\ell_2+d-1}\cdot \left(\frac{d\cdot \log n}{\rho\cdot \eps }\right)^{O(d)}.$$

Hence, the total running time of the algorithm is $\left( n^{\ell_1}+n^{\ell_2+d-1} +n^d\right) \cdot \left(\frac{d\cdot \log n}{\rho\cdot \eps }\right)^{O(d)}$.

In the following derivations we will use the fact that for any $x \in \mathbb{R}_{\ge 0}$ and $m \in \mathbb{N}$ we have \begin{align} \label{eq:rounding_fact}
\bigg\lfloor \frac{\lceil x \rceil} m \bigg\rfloor \le \bigg\lceil \frac {\lceil x \rceil} m \bigg\rceil = \bigg\lceil \frac x m \bigg\rceil. 
\end{align}
We bound $\ell_1$ as
\begin{equation}
	\label{eq:l1_runtime}
\ell_1 =\floor{\frac{q+d-1}{2}} = \floor{\frac{\ceil{\frac{d-1}{\eps} -d+\rho }+d-1}{2} } = \floor {\frac{\ceil{\frac{d-1}{\eps}+\rho -1} }{2}} \stackrel{\eqref{eq:rounding_fact}}{\leq}  \ceil{\frac{d-1}{2\cdot \eps} -\frac{1}{2} +\frac{\rho}{2}},
\end{equation}
We also bound
$$
q-\ell_1 +d-1 =q-\floor{\frac{q+d-1}{2} } +d-1=  \ceil{\frac{q+d-1}{2}}= \ceil{ \frac{\ceil{\frac{d-1}{\eps} - d+\rho +d-1}}{2}} \stackrel{\eqref{eq:rounding_fact}}{=} \ceil{\frac{d-1}{2\cdot \eps} -\frac{1}{2} +\frac{\rho}{2}}.
$$
Therefore,
\begin{equation}
	\label{eq:l2_runtime}
\ell_2 +d-1= \max\{q-\ell_1, 1\} +d-1 \leq \max\left\{ \ceil{\frac{d-1}{2\cdot \eps} -\frac{1}{2} +\frac{\rho}{2}}, \,d\right\}.
\end{equation}
By \eqref{eq:l1_runtime} and \eqref{eq:l2_runtime} the overall running time of the algorithm is 
\begin{align*}
\left( n^{\ceil{\frac{d-1}{2\cdot \eps} -\frac{1}{2} +\frac{\rho}{2}}} +n^d\right) \cdot \left(\frac{d\cdot \log n}{\rho\cdot \eps }\right)^{O(d)} \le \left( n^{\ceil{\frac{d-1}{2\cdot \eps} -\frac{1}{2} +\rho}} +n^d\right) \cdot \left(\frac{d\cdot \log n}{\rho\cdot \eps }\right)^{O(d)}.&\hfill\qedhere
\end{align*}
\end{proof}

\thmMainAlgo*
\begin{proof}
	By \Cref{lem:mitm_approx,lem:runtime}, for every
	$\eps\leq \frac{d-1}{d}$ there is an algorithm that returns a
	$(1-\eps)$-approximate solution for a $d$-Knapsack instance in time
	\(
	\big(
	n^{\ceil{\frac{d-1}{2\eps}-\frac{1}{2}+\rho}}
	+n^d
	\big)
	\cdot
	\big(
	\frac{d\log n}{\rho\eps}
	\big)^{O(d)}.\)
	
	It remains to consider the case
	$\frac{d-1}{d}<\eps<1$.
	Set $\eps'=\frac{d-1}{d}$ and run \Cref{alg:mitmddim} with error parameter $\eps'$.
	The algorithm returns a $(1-\eps')$-approximate solution, which is also a
	$(1-\eps)$-approximate solution since $\eps'\leq \eps$.
	We can bound its running time by using $\frac{d-1}{2\eps'}-\frac{1}{2}+\rho = \frac{d-1}2+\rho \le \frac {d+1}2 \le d$ and $\eps' = \Theta(\eps)$ as
	\[
	\left(
	n^{\ceil{\frac{d-1}{2\eps'}-\frac{1}{2}+\rho}}
	+n^d
	\right)
	\cdot
	\left(
	\frac{d\log n}{\rho\eps'}
	\right)^{O(d)}
	=
	n^d\cdot
	\left(
	\frac{d\log n}{\rho\eps}
	\right)^{O(d)}.
	\]
	Thus, the running time matches the bound stated in the theorem.
\end{proof}

\subsection{Proof of \Cref{lem:ddim_dpv}}
\label{sec:dalgo_proofapxlem}

\newcommand{\roundup}{\mathrm{round}^{\uparrow}}
\newcommand{\rounddown}{\mathrm{round}^{\downarrow}}
\newcommand{\round}{\mathrm{round}}
\newcommand{\candidates}{\mathcal{C}}
\newcommand{\mincandidate}{c}
\newcommand{\Rnd}{\mathcal{R}}
\newcommand{\minwd}{\mathrm{min}^d}

It remains to prove~\cref{thm:simple-apx-data-structure}. We will focus on the following variant, because it is easier to prove and easily implies \cref{thm:simple-apx-data-structure}.

\begin{lemma}[Internal Approximation Lemma]\label{lem:internal-APX-Lemma}
	Given a sequence $I=\{1,\ldots,n\}$ of $n$ items, weight estimates $\capacity_2,\ldots,\capacity_{d-1} \in \mathbb{R}_{\ge 0}$, and error parameters $\eps,\eta \in (0,1]$, in time $n\cdot \left(\log(n/\eta)/\eps\right)^{O(d)}$ we can compute
	sets $\mathcal{L}_1,\ldots,\mathcal{L}_n$ of virtual items each with $|\mathcal{L}_i| = (\log(n/\eta)/\eps)^{O(d)}$ and the following properties. For every $i \in [n]$ and $v \in \cL_i$ it holds that $\textsf{set}(v) \subseteq \{i,\ldots,n\}$. For every $i \in [n]$, let
	\[
	\mathcal{Q}_i \coloneqq \left\{ Q \subseteq \{i,\ldots,n\} \, \middle| \,p(Q) \ge \eta \cdot \max_{i \le j \le n} p(j),\, \weight_t(Q) \le \eta^{-1} \capacity_t \textnormal{ for each } t \in \{2,\ldots,d-1\} \right\}.
	\]
	Then for every $i \in [n]$ and $Q \in \mathcal{Q}_i$, there exists a virtual item $v \in \mathcal{L}_i$ that satisfies
	\begin{align*}
		p(v)& \geq (1-\eps)\, p(Q),\\
		\weight_1(v)& \leq \max\left\{(1+\eps)\, \weight_1(Q), \eta \cdot \max_{i \le j \le n}\weight_1(j)\right\},\\
		\weight_{d}(v) & \leq   \weight_d(Q),\text{ and }\\
		\weight_t(v)& \leq \max\left\{(1+\eps)\, \weight_t(Q), \eta \, \capacity_t\right\}\text{ for each } 1 < t < d.
	\end{align*}
\end{lemma}

\begin{proof}[Proof of~\cref{thm:simple-apx-data-structure}
assuming~\cref{lem:internal-APX-Lemma}]
We can assume that $\tp>0$, as for $\tp=0$ it suffices to return $\{v_\emptyset\}$, i.e., a singleton set containing the virtual item that represents the empty set. 
While in all other parts of this paper we sorted items by non-increasing profit, here we sort differently: 
We split the items by profit into $H \coloneqq \{i \in I \mid p(i) \ge \eta^{-1} \tp\}$ and $L \coloneqq I \setminus H$. We first sort the items such that $H$ precedes $L$, that is, $H = \{1,\ldots,|H|\}$ and $L = \{|H|+1,\ldots,n\}$. Finally, sort the items in $L$ by non-increasing first weight, that is, $\weight_1(|H|+1)\geq\cdots\geq\weight_1(n)$.

Now we apply \cref{lem:internal-APX-Lemma} on $I$ with error
parameters $\eps$ and $\eta' \coloneqq \eta^2$, and weight estimates
$\capacity_2,\ldots,\capacity_{d-1}$. Denote its output by
$\mathcal{L}_1,\ldots,\mathcal{L}_n$, and return
\[
	\mathcal{L}\coloneqq \big\{ v_{\{i\}} \mid i \in H \big\} \cup \bigcup_{i=|H|+1}^n\mathcal{L}_i.
\]
For correctness, consider an arbitrary $Q\in\mathcal{Q}$.
In case $Q \cap H \ne \emptyset$, consider any $i \in Q \cap H$. Then $v_{\{i\}} \in \cL$ satisfies 
\begin{align*}
	p(v_{\{i\}}) = p(i) \ge \eta^{-1} \tp \ge \min\big\{ (1-\eps)\, p(Q), \eta^{-1} \tp \big\},
\end{align*}
and $\weight_t(v_{\{i\}}) = \weight_t(i) \le \weight_t(Q)$ for all $t \in [d]$, so $v_{\{i\}}$ satisfies the claims of \cref{thm:simple-apx-data-structure}.

It remains to consider the case $Q \cap H = \emptyset$. 
Let $i$ be the smallest index of an item in $Q$, and observe that $|H| < i \le n$, 
	$Q\subseteq\{i,\ldots,n\}$, and $\weight_1(i)\leq\weight_1(Q)$. 
	Since $Q\in\mathcal{Q}$ it holds that $\weight_t(Q) \le \eta^{-1} \capacity_t \le (\eta')^{-1} \capacity_t$ for each $t \in \{2,\ldots,d-1\}$ and 
	\[ p(Q) \ge \eta \, \tp \ge \eta^2 \max_{i \le j \le n} p(j) = \eta' \cdot \max_{i \le j \le n} p(j), \]
	and thus $Q$ belongs to the family $\mathcal{Q}_i$ of the application of \cref{lem:internal-APX-Lemma}. 
	Hence, by~\cref{lem:internal-APX-Lemma} there exists a virtual item $v\in\mathcal{L}_i\subseteq\mathcal{L}$
	that satisfies
	\begin{align*}
		p(v)&\geq (1-\eps)\, p(Q)
		\ge \min\left\{(1-\eps)\, p(Q),\eta^{-1}\tp\right\},\\
		\weight_1(v)&\leq
		\max\left\{(1+\eps)\weight_1(Q),
		\eta' \cdot \max_{i\leq j\leq n}\weight_1(j)\right\}
		=\max\left\{(1+\eps)\weight_1(Q),
		\eta' \cdot \weight_1(i)\right\}
		=(1+\eps)\weight_1(Q),\\
		\weight_d(v)&\leq\weight_d(Q),\\
		\weight_t(v)&\leq
		\max\left\{(1+\eps)\weight_t(Q),
		\eta' \capacity_t\right\}
		\leq \max\left\{(1+\eps)\weight_t(Q),
		\eta \capacity_t\right\}
		\qquad\text{for every }t \in \{2,\ldots,d-1\}.
	\end{align*}
	Hence, all claimed bounds of \cref{thm:simple-apx-data-structure} hold.
	Finally, note that the running time is dominated by the call to \cref{lem:internal-APX-Lemma}, which takes time $n\cdot(\log(n/\eta')/\eps)^{O(d)} = n\cdot(\log(n/\eta)/\eps)^{O(d)}$.
\end{proof}

The rest of this section is dedicated to the proof
of~\cref{lem:internal-APX-Lemma}. 

\begin{proof}[Proof of \cref{lem:internal-APX-Lemma}]
We may assume that $n\geq2$ is a power of
two (otherwise add dummy items with profit and weights 0). Define 
\[ \delta\coloneqq \frac \eps{10\log n} \quad \text{ and } \quad \Delta\coloneqq \Big(\frac{\eps \, \eta}{10n} \Big)^5. \]
For every nonempty $J\subseteq I$, we define the error scales
\begin{align*}
	E_0^J \coloneqq \Delta\max_{i\in J}p(i), \qquad E_1^J \coloneqq \Delta\max_{i\in J}\weight_1(i), \qquad E_t^J \coloneqq \Delta\capacity_t \quad \text{for }t\in\{2,\ldots,d-1\},
\end{align*}
and the domains:
\begin{align*}
	\mathcal{P}^J&\coloneqq
	\big([E_0^J,
	\Delta^{-2} \cdot E_0^J] \cap \{(1+\delta)^k\mid k\in\mathbb Z\}\big) \cup\{0\},\\
	\mathcal{W}_t^J&\coloneqq
	\big([E_t^J,
	\Delta^{-2} \cdot E_t^J] \cap \{(1+\delta)^k\mid k\in\mathbb Z\} \big) \cup\{0\}, \qquad
	\text{for every }t\in [d-1],\\
	\Rnd^J&\coloneqq \mathcal{P}^J\times \mathcal{W}_1^J\times\cdots\times\mathcal{W}_{d-1}^J.
\end{align*}
Since for any number $x > 0$ the set $[x, \Delta^{-2} x] \cap \{(1+\delta)^k\mid k\in\mathbb Z\}$ has size $O(\log_{1+\delta} (\Delta^{-2})) = O(\log(\Delta^{-1})/\delta) = (\log(n/\eta)/\eps)^{O(1)}$, it holds that
$|\Rnd^J|=(\log(n/\eta)/\eps)^{O(d)}$.

\medskip
We represent a set of items $S \subseteq I$ by a vector $\vecx \in \mathbb{R}^d_{\ge 0}$ with $\vecx_0 = p(S)$ and $\vecx_t = \weight_t(S)$ for each $t \in [d-1]$.
Correspondingly, the natural partial order on vectors $\vecx, \vecy \in \mathbb{R}^d_{\ge 0}$ is
$\vecx \succeq \vecy$ if and only if $\vecx_0 \ge \vecy_0$ and $\vecx_t \le \vecy_t$ for each $t \in [d-1]$. We also define $\vecx+\vecy$ by
coordinate-wise addition. 

We will use the
dummy symbol $\bot$ with $p(\bot) \coloneqq \infty$ and $\weight_t(\bot) \coloneqq - \infty$ for
all $t \in [d]$. 

\begin{definition}[Convolution]\label{def:convolution}
For nonempty $A \subseteq I$ we call $f$ an \emph{$A$-function} if $f$ maps each $\vecx \in \Rnd^A$ either to $f(\vecx) = \bot$ or to a virtual item $f(\vecx)$ with $\textsf{set}(f(\vecx)) \subseteq A$.
For nonempty and disjoint $A,B \subseteq I$, an $A$-function $f$, and a $B$-function $g$, we define the \emph{convolution} $f \oplus g$ as the $(A \cup B)$-function that maps each $\vecx \in \Rnd^{A \cup B}$ to the virtual item $v = f(\vecy) \cup g(\vecz)$ that minimizes the $d$-th weight $\weight_d(v)$, breaking ties arbitrarily, over all $\vecy \in \Rnd^A, \vecz \in \Rnd^B$ with $\vecy + \vecz \succeq \vecx, f(\vecy) \ne \bot$, and $g(\vecz) \ne \bot$, that is,
\[ (f \oplus g)(\vecx) \coloneqq \argmin_v\{ \weight_d(v) \mid v = f(\vecy) \cup g(\vecz), \vecy \in \Rnd^A, \vecz \in \Rnd^B, \vecy+\vecz \succeq \vecx, f(\vecy) \ne \bot, g(\vecz) \ne \bot \}. \]
If the set on the right hand side is empty, then we define $(f \oplus g)(\vecx) \coloneqq \bot$.
\end{definition}

Next, we analyze the properties of the convolution operation~$\oplus$.

\begin{definition}[Approximate function]\label{def:correct_function}
	Let $C \subseteq I$ and $\gamma \ge 0$. For any $\vecx \in \mathbb{R}^d_{\ge 0}$ define
	\begin{align*} 
		\mathcal{Q}_{C}^{(\gamma)}(\vecx) \coloneqq \{X \subseteq C \mid \; &\vecx_0 \le \max\{ (1+\delta)^{-\gamma} \cdot p(X) - (8^\gamma - 1) E_0^C, 0\}, \\
		&\vecx_t \ge (1+\delta)^\gamma \cdot \weight_t(X) + (8^\gamma - 1) E_t^C \;\text{ for each } t \in [d-1]\}. 
	\end{align*}
	A $C$-function $h$ is \emph{$\gamma$-approximate on $C$} if
	for all $\vecx \in \Rnd^C$ we have:
	\begin{enumerate}
		\item\label{clm:correctness_of_dpv_profit}
		$p(h(\vecx))\geq\vecx_0$,
		\item\label{clm:correctness_of_dpv_weights}
		$\weight_t(h(\vecx))\leq\vecx_t$
		for every $t\in [d-1]$,
		\item\label{clm:correctness_of_dpv_item3} $\weight_d(h(\vecx)) \le \min\{\weight_d(X) \mid X \in \mathcal{Q}_{C}^{(\gamma)}(\vecx) \}$, and
		\item\label{clm:correctness_of_dpv_item4} if $\mathcal{Q}_{C}^{(\gamma)}(\vecx) \ne \emptyset$ then $h(\vecx) \ne \bot$.
	\end{enumerate}
	Observe that for any $\gamma \le \gamma'$ it holds that $\mathcal{Q}_{C}^{(\gamma')}(\vecx) \subseteq \mathcal{Q}_{C}^{(\gamma)}(\vecx)$, and thus any $\gamma$-approximate function is also $\gamma'$-approximate.
\end{definition}

Note that Properties \ref{clm:correctness_of_dpv_profit}-\ref{clm:correctness_of_dpv_item3} are trivially satisfied if $h(x) = \bot$, since $p(\bot) = \infty$ and $\weight_t(\bot) = -\infty$ for all $t \in [d]$. Property \ref{clm:correctness_of_dpv_item4} limits when this trivial solution is allowed.

\begin{claim}\label{clm:correctness_of_dpv}
	Let $A,B \subseteq I$ be nonempty and disjoint, let $f$ be $\alpha$-approximate on $A$, and let $g$ be $\beta$-approximate on $B$. If $\gamma\coloneqq\max\{\alpha,\beta\}+1\leq\log n$, then $h \coloneqq f \oplus g$
	is $\gamma$-approximate on $C \coloneqq A \cup B$. Moreover, given $f,g$, the function $h$ can be computed in time $(\log(n/\eta)/\eps)^{O(d)}$.
\end{claim}
\begin{claimproof}
	For the running time, we first recall that $|\Rnd^J| = (\log(n/\eta)/\eps)^{O(d)}$ for any $J \subseteq I$. To compute $h = f \oplus g$, we iterate over $(\log(n/\eta)/\eps)^{O(d)}$ triples $\vecx,\vecy,\vecz$. If $\vecy + \vecz \succeq \vecx$, $f(\vecy) \ne \bot$ and $g(\vecz) \ne \bot$ then we compute the new virtual item $f(\vecy) \cup g(\vecz)$ in time $O(d)$, and update $(f \oplus g)(\vecx)$ in case we find a smaller $d$-th weight. Hence, the total running time is $O(d) \cdot (\log(n/\eta)/\eps)^{O(d)} = (\log(n/\eta)/\eps)^{O(d)}$.
	
	It remains to show that $h$ is $\gamma$-approximate. 
	If
	$h(\vecx)=\bot$, Properties~\ref{clm:correctness_of_dpv_profit} and
	\ref{clm:correctness_of_dpv_weights} follow from the definition of $\bot$, since $p(\bot) = \infty$ and $\weight_t(\bot) = -\infty$ for all $t \in [d]$.
	Otherwise, let $\vecy\in\Rnd^A$ and $\vecz\in\Rnd^B$ be the pair
	selected by $\oplus$, so that
	$h(\vecx)=f(\vecy)\cup g(\vecz)$. 
	Since $A,B$ are disjoint, it holds that $p(h(\vecx)) = p(f(\vecy)\cup g(\vecz)) = p(f(\vecy)) + p(g(\vecz)) \ge \vecy_0 + \vecz_0 \ge \vecx_0$, where the last two steps used that $f$ is $\alpha$-approximate on $A$, $g$ is $\beta$-approximate on $B$, and $\vecy + \vecz \succeq \vecx$. Analogously, $\weight_t(h(\vecx)) = \weight_t(f(\vecy)) + \weight_t(g(\vecz)) \le \vecy_t + \vecz_t \le \vecx_t$. 
	
	Property \ref{clm:correctness_of_dpv_item3} is trivial if $\mathcal{Q}_{C}^{(\gamma)}(\vecx) = \emptyset$, since then the right hand side is $\infty$. Hence, for proving Properties \ref{clm:correctness_of_dpv_item3} and \ref{clm:correctness_of_dpv_item4} we can assume $\mathcal{Q}_{C}^{(\gamma)}(\vecx) \ne \emptyset$. So let $X \in \mathcal{Q}_{C}^{(\gamma)}(\vecx)$ minimize $\weight_d(X)$.
	We partition $X$ into $Y = X \cap A$ and $Z = X \cap B$. Define $\vecy',\vecz' \in \mathbb{R}^d_{\ge 0}$ by
	\begin{align*}
		\vecy'_0 &\coloneqq \max\{(1+\delta)^{-(\gamma-1)}p(Y) - (8^{\gamma-1} - 1) E_0^A, 0\},
		\\
		\vecy'_t &\coloneqq (1+\delta)^{\gamma-1}\weight_t(Y) + (8^{\gamma-1} - 1) E_t^A,
		\qquad\text{for each }t\in [d-1], \\
		\vecz'_0 &\coloneqq \max\{(1+\delta)^{-(\gamma-1)}p(Z) - (8^{\gamma-1} - 1) E_0^B, 0\}, \\
		\vecz'_t &\coloneqq (1+\delta)^{\gamma-1}\weight_t(Z) + (8^{\gamma-1} - 1) E_t^B,
		\qquad\text{for each }t\in [d-1].
	\end{align*}
	Thus, $Y\in\mathcal{Q}_A^{(\gamma-1)}(\vecy')$ and
	$Z\in\mathcal{Q}_B^{(\gamma-1)}(\vecz')$.
	We first show that the entries of $\vecy'$ lie in a good range.
	\begin{claim} \label{cla:vecy_range}
		For every $1 \le \gamma \le \log n$ and $t \in \{0,1,\ldots,d-1\}$ we have 
		$$0 \le \vecy'_t \le \max(\{0\} \cup \{ (1+\delta)^k \mid k \in \mathbb{Z}, (1+\delta)^k \le \Delta^{-2} E_t^A \}).$$
	\end{claim}
	\begin{proof}
	The lower bound $\vecy'_t \ge 0$ is immediate. For the upper bound, for $t = 0$ we have
	\[ \vecy'_0 = \max\{(1+\delta)^{-(\gamma-1)}p(Y) - (8^{\gamma-1} - 1) E_0^A, 0\} \le p(Y) \le n \cdot \max_{i \in A} p(i) = n \Delta^{-1} E_0^A \le \tfrac 12 \Delta^{-2} E_0^A, \]
	and since $\delta \le 1$ there is at least one power of $1+\delta$ in $[\tfrac 12 \Delta^{-2} E_0^A, \Delta^{-2} E_0^A]$ (unless $E_0^A = 0$, in which case $\vecy'_0 = 0$, so the statement is satisfied).
	
	For $t = 1$, since $E_1^A = \Delta \max_{i \in A} \weight_1(i) \ge \Delta \weight_1(Y) / n$ and $(1+\delta)^\gamma \le 2$ we have
	\begin{align*} \vecy'_1 &= (1+\delta)^{\gamma-1}\weight_1(Y) + (8^{\gamma-1} - 1) E_1^A  \\
	&\le \big((1+\delta)^{\gamma-1} n \Delta^{-1} + (8^{\gamma-1} - 1)\big) E_1^A \le (2 n \Delta^{-1} + n^3) E_1^A \le \tfrac 12 \Delta^{-2} E_1^A, 
	\end{align*}
	and the same argument as before applies.
	
	Finally, for $t > 1$, since $X \in \mathcal{Q}_{C}^{(\gamma)}(\vecx)$, $\vecx \in \Rnd^C$, and $E_t^C = E_t^A = \Delta \capacity_t$ we have
	\[ \vecy'_t = (1+\delta)^{\gamma-1}\weight_t(Y) + (8^{\gamma-1} - 1) E_t^A \le (1+\delta)^\gamma \weight_t(X) + (8^\gamma - 1) E_t^C \le \vecx_t \le \Delta^{-2} E_t^C = \Delta^{-2} E_t^A. \]
	Since $\vecx \in \Rnd^C$, $\vecx_t$ is either 0 or a power of $1+\delta$, in particular we have $\vecy'_t = 0$ or
	\[
	\vecy'_t \le \vecx_t \le \max\{ (1+\delta)^k \mid k \in \mathbb{Z}, (1+\delta)^k \le \Delta^{-2} E_t^A \}. \qedhere
	\]
	\end{proof}
	
	Since $\vecy'$ might not be in $\Rnd^A$, we round it to $\vecy \in \Rnd^A$ as follows. 
	Round down $\vecy'_0$ to the largest value $\vecy_0 \in \mathcal{P}^A$ not
	exceeding $\vecy'_0$. For each $t \in [d-1]$, round up $\vecy'_t$ to the smallest value $\vecy_t \in \mathcal{W}_t^A$ not below $\vecy'_t$. By the definitions of $\mathcal{P}^A, \mathcal{W}_t^A, \Rnd^A$ and \Cref{cla:vecy_range}, this rounding is well-defined and yields a vector $\vecy \in \Rnd^A$ that satisfies
	\begin{align} \label{eq:truncated-rounding}
		\vecy'_0 \ge \vecy_0 \ge \tfrac {\vecy'_0}{1+\delta} - E_0^A \quad \text{ and } \quad \vecy'_t \le \vecy_t \le (1+\delta)(\vecy'_t + E_t^A) \quad \text{ for any } t \in [d-1]. 
	\end{align}
	By symmetry, an analogous statement holds for $\vecz'$ which we round to $\vecz \in \Rnd^B$.
	
	Using~\eqref{eq:truncated-rounding} and the definition of $\vecy'$, and the analogous statements for $\vecz'$, we have
	\begin{align*}
		\vecy_0+\vecz_0
		&\geq(1+\delta)^{-\gamma}\cdot p(Y)
		-8^{\gamma-1} E_0^A + (1+\delta)^{-\gamma}\cdot p(Z)
		-8^{\gamma-1} E_0^B\\
		&\geq (1+\delta)^{-\gamma} p(X) - (8^{\gamma} - 1) E_0^C,
	\end{align*}
	and since also $\vecy_0+\vecz_0 \ge 0$, and $X\in\mathcal{Q}_C^{(\gamma)}(\vecx)$, we obtain
	\[ \vecy_0+\vecz_0 \ge \max\{ (1+\delta)^{-\gamma} p(X) - (8^{\gamma} - 1) E_0^C, 0 \} \ge \vecx_0.
	\]
	By similar arguments, for every $t\in\{1,\ldots,d-1\}$,
	\begin{align*}
		\vecy_t+\vecz_t
		&\leq (1+\delta)^\gamma \weight_t(Y) + 2 \cdot 8^{\gamma-1} E_t^A + (1+\delta)^\gamma \weight_t(Z) + 2 \cdot 8^{\gamma-1} E_t^B \\
		&\leq (1+\delta)^\gamma \weight_t(X)
		+ (8^\gamma - 1) E_t^C \\
		&\leq \vecx_t.
	\end{align*}
	Thus $\vecy+\vecz\succeq\vecx$, so this pair is considered by the
	convolution. 
	Since profit was rounded down and weights
	were rounded up, it holds that $\mathcal{Q}_A^{(\gamma-1)}(\vecy') \subseteq \mathcal{Q}_A^{(\gamma-1)}(\vecy)$ and
	$\mathcal{Q}_B^{(\gamma-1)}(\vecz') \subseteq \mathcal{Q}_B^{(\gamma-1)}(\vecz)$, and thus $Y$ shows that $\mathcal{Q}_A^{(\gamma-1)}(\vecy) \ne \emptyset$ and $Z$ shows that $\mathcal{Q}_B^{(\gamma-1)}(\vecz) \ne \emptyset$. Since $f$ is $\alpha$-approximate on $A$ and $g$ is $\beta$-approximate on $B$, it follows that $f(\vecy) \ne \bot$ and $g(\vecz)\neq\bot$. By definition of $\oplus$ it now follows that $h(\vecx)\neq\bot$, proving Property~\ref{clm:correctness_of_dpv_item4}. 
	Moreover, since the convolution chooses a minimizing feasible
	pair and $f,g$ are $(\gamma-1)$-approximate, we obtain
	\[
		\weight_d(h(\vecx))
		\leq\weight_d(f(\vecy))+\weight_d(g(\vecz))
		\leq\weight_d(Y)+\weight_d(Z)=\weight_d(X),
	\]
	which proves Property~\ref{clm:correctness_of_dpv_item3}.
\end{claimproof}

Recall that we assume $n$ to be a power of two. We partition $I = \{1,\ldots,n\}$ in a binary-tree-like way. Specifically, for every $\alpha \in
\{0,\ldots,\log{n}\}$ and $\beta \in [n/2^\alpha]$, let $I_{\alpha,\beta}$ be the items with index in $[(\beta - 1) \cdot 2^\alpha + 1, \beta \cdot 2^\alpha]$, that is,
\[
	I_{\alpha,\beta}
	\coloneqq[(\beta - 1) \cdot 2^\alpha + 1, \beta \cdot 2^\alpha] \cap I.
\]
Observe that $I_{\log{n},1} = I$ and $I_{\alpha,\beta} = I_{\alpha-1,2\beta-1}
\cup I_{\alpha-1,2\beta}$ for every $\alpha \ge 1$ and every $\beta$. 
Using \Cref{clm:correctness_of_dpv}, we can efficiently compute an $\alpha$-approximate function for each interval $I_{\alpha,\beta}$.

\begin{claim}\label{lem:functions_fJ}
	In time $n\cdot \left(\log(n/\eta)/\eps\right)^{O(d)}$, for all $\alpha \in \{0,\ldots,\log{n}\}$ and $\beta \in [n/2^{\alpha}]$ we can compute functions
	$f_{\alpha,\beta}$ such that each $f_{\alpha,\beta}$ is $\alpha$-approximate on $I_{\alpha,\beta}$.
\end{claim}
\begin{claimproof}
	In the base case $\alpha = 0$ and $\beta \in [n]$, for any $\vecx \in \Rnd^{I_{0,\beta}}$ we set
	\begin{equation*}
		f_{0,\beta}(\vecx) \coloneqq \begin{cases}
			v_{\{ \beta \}} & \text{if } p(\beta) \ge \vecx_0 > 0 \text{ and } \weight_t(\beta) \leq \vecx_t \text{ for each } t \in [d-1],\\
			v_{\emptyset} & \text{if } \vecx_0 = 0,\\
			\bot & \text{otherwise.}
		\end{cases}
	\end{equation*}
	Note that for $\gamma = 0$ all errors in the definition of $\gamma$-approximate in \cref{def:correct_function} vanish.
	It is immediate that $f_{0,\beta}$ is 0-approximate. Indeed, if
	$\vecx_0>0$, the
	feasible family consists of the singleton set~$\{\beta\}$ precisely in the first case and is
	empty otherwise. If $\vecx_0=0$, the empty set is feasible and has
	$d$-th weight zero. These observations prove all properties in the definition
	of a 0-approximate function.
	
	With this base case in hand, we compute functions for every $\alpha = 1,2,\ldots,\log{n}$ and $\beta = 1,2,\ldots,n/2^\alpha$ by
	\begin{align*}
		f_{\alpha,\beta} &\coloneqq f_{\alpha-1,2\beta-1} \oplus f_{\alpha-1,2\beta}.
	\end{align*}
	By~\cref{clm:correctness_of_dpv}, each computed function $f_{\alpha,\beta}$ is $\alpha$-approximate on $I_{\alpha,\beta}$. It remains to analyze the running time. For each $\beta \in [n]$, the explicit table of $f_{0,\beta}$ can be constructed in time $O(d) \cdot |\Rnd^{I_{0,\beta}}|=(\log(n/\eta)/\eps)^{O(d)}$. By~\cref{clm:correctness_of_dpv}, the total running time is therefore bounded by
	\begin{align*}
		\sum_{\alpha = 0}^{\log n} \sum_{\beta = 1}^{n/2^\alpha} \left(\frac{\log(n/\eta)}{\eps}\right)^{O(d)} =
		\sum_{\alpha = 0}^{\log n} \frac{n}{2^\alpha} \cdot \left(\frac{\log(n/\eta)}{\eps}\right)^{O(d)}=
		n \left(\frac{\log(n/\eta)}{\eps}\right)^{O(d)}
		.&\hfill\qedhere
	\end{align*}
\end{claimproof}

\begin{claim}\label{clm:suffix-functions}
	In time $n\cdot \left(\log(n/\eta)/\eps\right)^{O(d)}$, for all $i \in [n]$ we can compute functions $f_i$ such that each $f_i$ is $\log{n}$-approximate on $I_i \coloneqq \{i,\ldots,n\}$.
\end{claim}
\begin{claimproof}
	We use~\cref{lem:functions_fJ} to compute functions $f_{\alpha,\beta}$ for every $\alpha \in \{0,\ldots,\log{n}\}$ and $\beta \in [n/2^\alpha]$, where $f_{\alpha,\beta}$ is $\alpha$-approximate on $I_{\alpha,\beta}$.
	Next, for every $i \in [n]$ we construct the partition $P_i$ of $I_i$ into the maximal intervals $I_{\alpha,\beta}$ that are contained in $I_i$.\footnote{This partition is easily computed by a recursive procedure: Start with the call $\textsc{Rec}(\log n,1,i,n)$. In the call $\textsc{Rec}(\alpha,\beta,i,j)$, if $I_{\alpha,\beta} \subseteq \{i,\ldots,j\}$ then add $I_{\alpha,\beta}$ to $P_i$ and recurse on $\textsc{Rec}(\alpha-1,2\beta-2,i,j - 2^\alpha)$, otherwise recurse on $\textsc{Rec}(\alpha-1,2\beta,i,j)$. Abort once $\beta \le 0$ or $\alpha < 0$.} 
	Recall that we assumed $n$ to be a power of two. From this one can observe that $P_i = \{I_{\alpha_1,\beta_1}, \ldots, I_{\alpha_k,\beta_k}\}$ with $\alpha_1 < \ldots < \alpha_k$. In particular, we have $k \le \log n$. 
	We now let $f^{(1)} \coloneqq f_{\alpha_1,\beta_1}$ and compute for each $q=2,3,\ldots,k$:
	\[ f^{(q)} \coloneqq f^{(q-1)} \oplus f_{\alpha_q,\beta_q}. \]
	Finally, we let $f_i \coloneqq f^{(k)}$.
	Since $P_i$ is a partition of $I_i$ and $\alpha_1 < \ldots < \alpha_k$, \cref{clm:correctness_of_dpv} implies that $f^{(q)}$ is $(\alpha_q+1)$-approximate on $I_{\alpha_1,\beta_1} \cup \ldots \cup I_{\alpha_q,\beta_q}$. In particular, $f_i$ is $(\alpha_k+1)$-approximate on~$I_i$. We finish by observing that if $k \ge 2$ then $\alpha_k < \log n$, so $f_i$ is $\log n$-approximate on $I_i$, and if $k = 1$ then $f_i = f^{(1)} = f_{\alpha_1,\beta_1}$ is $\alpha_1 \le \log n$-approximate on $I_i$.
	
	The construction performs at most $\log n$
	convolutions per $i \in [n]$, so the total running time is
	$n\cdot(\log(n/\eta)/\eps)^{O(d)}$ by~\cref{clm:correctness_of_dpv}.
\end{claimproof}

We now construct the sets promised by~\cref{lem:internal-APX-Lemma}. Compute the functions $f_i$ from \cref{clm:suffix-functions}, where each $f_i$ is $\log n$-approximate on $I_i = \{i,\ldots,n\}$.
Return, for all $i \in [n]$, the sets
\begin{equation*}
	\mathcal{L}_i\coloneqq
	\left\{f_i(\vecx)\,\middle|\,
	\vecx\in\Rnd^{I_i},\ f_i(\vecx)\neq\bot\right\}.
\end{equation*}

Since $|\Rnd^{I_i}|=(\log(n/\eta)/\eps)^{O(d)}$, every list $\mathcal{L}_i$ has the required
size. As the functions $f_i$ are computed in total time $n\cdot(\log(n/\eta)/\eps)^{O(d)}$ and stored explicitly, the lists $\mathcal{L}_1,\ldots,\mathcal{L}_n$ can be constructed in
time $n\cdot(\log(n/\eta)/\eps)^{O(d)}$. Since $f_i$ is $\log n$-approximate on $I_i$, each virtual item $v \in \mathcal{L}_i$ satisfies $\textsf{set}(v) \subseteq I_i$.

It remains to prove the approximation guarantee. Fix $i\in[n]$ and
$Q\in\mathcal{Q}_i$. 
Define $\vecx' \in \mathbb{R}_{\geq0}^d$ by
\begin{align*}
		\vecx'_0 &\coloneqq \max\{(1+\delta)^{- \log n}p(Q) - n^3 E_0^{I_i}, 0\},\\
		\vecx'_t &\coloneqq (1+\delta)^{\log n}\weight_t(Q) + n^3 E_t^{I_i},
		\qquad\qquad\qquad \text{for }t\in [d-1],
\end{align*}
so that $Q \in \mathcal{Q}^{(\log n)}_{I_i}(\vecx')$.
\begin{claim} \label{cla:range_vecx}
We have $\vecx'_t \in \big[2 E_t^{I_i}, \tfrac 12 \Delta^{-2} E_t^{I_i}\big]$ for all $t \in \{0,1,\ldots,d-1\}$.
\end{claim}
\begin{claimproof}
Using $(1+\delta)^{\log n} \le 2$, the definitions of $\Delta$ and $E_t^{I_i}$, and the properties 
$p(Q)\geq\eta \cdot \max_{j \in I_i}p(j)$ and
$\weight_t(Q)\leq \eta^{-1} \capacity_t$ for all $t \in \{2,\ldots,d-1\}$ by $Q \in \mathcal{Q}_i$, we can bound
\begin{align*}
	&\vecx'_0 \le p(Q) \le n \max_{j \in I_i} p(j) \le \tfrac 12 \Delta^{-2} E_0^{I_i}, \\
	&\vecx'_0 \ge \tfrac 12 p(Q) - n^3 E_0^{I_i} \ge \tfrac 12 \cdot \eta \cdot \max_{j \in I_i} p(j) - n^3 E_0^{I_i} = \big( \tfrac 12 \cdot \eta \cdot \Delta^{-1} - n^3 \big) E_0^{I_i} \ge 2 E_0^{I_i}, \\
	&2 E_1^{I_i} \le n^3 E_1^{I_i} \le \vecx'_1 \le 2 \weight_1(Q) + n^3 E_1^{I_i} \le 2n \max_{j \in I_i} \weight_1(j) + n^3 E_1^{I_i} \le \tfrac 12 \Delta^{-2} E_1^{I_i}, \\
	&2 E_t^{I_i} \le n^3 E_t^{I_i}  \le \vecx'_t \le 2 \cdot \eta^{-1} \cdot \capacity_t + n^3 E_t^{I_i} \le \tfrac 12 \Delta^{-2} E_t^{I_i} \qquad \text{ for } t \in \{2,\ldots,d-1\}. \qedhere
\end{align*}
\end{claimproof}
We round $\vecx'$ to $\vecx \in \Rnd^{I_i}$ as follows. 
	Round down $\vecx'_0$ to the largest value $\vecx_0 \in \mathcal{P}^{I_i}$ not
	exceeding~$\vecx'_0$. For each $t \in [d-1]$, round up $\vecx'_t$ to the smallest value $\vecx_t \in \mathcal{W}_t^{I_i}$ not below $\vecx'_t$. By the definitions of $\mathcal{P}^{I_i}, \mathcal{W}_t^{I_i}, \Rnd^{I_i}$ and \Cref{cla:range_vecx}, this rounding is well-defined and yields a vector $\vecx \in \Rnd^{I_i}$ that satisfies
	\begin{align} \label{eq:vecx_vs_vecxprime}
		\vecx'_0 \ge \vecx_0 \ge \tfrac 1{1+\delta} \cdot \vecx'_0 \quad \text{ and } \quad \vecx'_t \le \vecx_t \le (1+\delta) \vecx'_t \quad \text{ for any }t \in [d-1]. 
	\end{align}
Since we rounded down profits and rounded up weights, we have $\mathcal{Q}^{(\log n)}_{I_i}(\vecx') \subseteq \mathcal{Q}^{(\log n)}_{I_i}(\vecx)$, and thus $\mathcal{Q}^{(\log n)}_{I_i}(\vecx)$ contains $Q$ and is thus nonempty. It follows that $f_i(\vecx) \ne \bot$, since $f_i$ is $\log n$-approximate on $I_i$. Hence, $v = f_i(\vecx)$ is a virtual item in $\mathcal{L}_i$, and it satisfies
\begin{align*}
	p(v) &\ge \vecx_0, \\
	\weight_t(v) &\le \vecx_t \text{ for all }t \in [d-1], \text{ and} \\
	\weight_d(v) &\le \min\{ \weight_d(X) \mid X \in \mathcal{Q}^{(\log n)}_{I_i}(\vecx) \} \le \weight_d(Q).
\end{align*}
For the profit, by \eqref{eq:vecx_vs_vecxprime}, the definition of $\vecx'$, $E_0^{I_i} = \Delta \max_{j \in I_i} p(j)$, and $p(Q) \ge \eta \cdot \max_{j \in I_i} p(j)$, we can further bound
\begin{align*}
	p(v) \ge \vecx_0 \ge \tfrac 1{1+\delta} \cdot \vecx'_0 \ge (1+\delta)^{-\log(n) - 1} p(Q) - n^3 E_0^{I_i} &\ge (1+\delta)^{-\log(n) - 1} p(Q) - n^3 \Delta \eta^{-1} p(Q) \\
	&\ge (1 - \tfrac \eps 2) p(Q) - \tfrac \eps 2 p(Q)
	\ge (1-\eps) p(Q).
\end{align*}
Similarly, for the weights we obtain for any $t \in [d-1]$
\begin{align*}
	\weight_t(v) \le (1+\delta)^{\log(n)+1} \weight_t(Q) + (1+\delta) n^3 E_t^{I_i} &\le (1 + \tfrac \eps 3) \weight_t(Q) + 2 n^3 E_t^{I_i}  \\
	&\le \max\{(1 + \eps) \weight_t(Q), \tfrac{1+\eps}{2\eps/3} \cdot 2 n^3 E_t^{I_i} \}  \\
	&\le \max\{(1 + \eps) \weight_t(Q),  \eta\, \Delta^{-1} E_t^{I_i} \},
\end{align*}
where in the second-to-last step we used that for all $a,b \ge 0$ and $c \in [0,1]$ it holds that $c \cdot a + (1-c) \cdot b \le \max\{a,b\}$, applied with $a = (1 + \eps) \weight_t(Q), b = \tfrac{1+\eps}{2\eps/3} \cdot 2 n^3 E_t^{I_i}$, and $c = (1 + \tfrac \eps 3)/(1+\eps)$.

Since $E_1^{I_i} = \Delta \max_{j \in I_i} \weight_1(j)$ and $E_t^{I_i} = \Delta \capacity_t$ for $t \in \{2,\ldots,d-1\}$, this yields
\begin{align*}
	\weight_1(v) &\le \max\big\{(1 + \eps) \weight_1(Q), \eta \cdot \max_{j \in I_i} \weight_1(j) \big\},  \\
	\weight_t(v) &\le \max\big\{(1 + \eps) \weight_t(Q), \eta \cdot  \capacity_t \big\} \qquad\qquad \text{ for }t \in \{2,\ldots,d-1\}.
\end{align*}
Hence, all claimed bounds of \cref{lem:internal-APX-Lemma} hold. 
\end{proof}

\section{A Near-Linear-Time \boldmath$\left(\frac{2}{3}-\eps\right)$-Approximation for \boldmath$2$-Knapsack}
\label{sec:linear2d}

In this section, we present a $\left(\frac{2}{3}-\eps\right)$-approximation algorithm for $2$-Knapsack with running time $n \cdot \big(\frac{\log n}{\eps}\big)^{O(1)}$, proving \Cref{thm:linear}. Our algorithm builds on the toolkit developed in the previous section: the meet-in-the-middle approach developed in \Cref{sec:dalgo_ORS}, the Slack-Generating Lemma (\Cref{lem:structural}), and the Approximation Lemma (\Cref{thm:simple-apx-data-structure}). The key difference is that we tailor these ingredients to the target approximation ratio in two dimensions. This allows us to avoid the quadratic dependence on~$n$, which arises in \Cref{thm:dalgo} independently of the desired approximation ratio.

\medskip
Let us first explain the intuition behind our algorithm. We handle four cases.

Case 0: The case $|\OPT| \le 1$ is trivial, as by checking the empty set and all singleton sets we find the optimal solution. Therefore, in what follows we assume $|\OPT| \ge 2$.

Case 1: Recall that the Slack-Generating Lemma loses the profit of $d-1 = 1$ item (as well as a negligible factor $1-\eps$). 
Thus, in case $p(\OPT[1]) \le p(\OPT)/3$ the Slack-Generating Lemma shows existence of a subset $Q \subseteq \OPT$ of profit at least $\approx \frac 23 p(\OPT)$ that has slack in the first weight dimension. Because this subset exists, the Approximation Lemma computes a feasible virtual item~$v$ of profit $\approx p(Q)\approx \frac 23 p(\OPT)$. Hence, the output of the Approximation Lemma contains a sufficient approximation.

Case 2: By using meet in the middle with one item in the left half and one item in the right half we can find the optimal solution of cardinality 2. 
In case $p(\OPT[\semi 2]) \ge (\frac 23 - \eps) p(\OPT)$, the optimal solution of cardinality 2 has profit at least $p(\OPT[\semi 2])$, which gives a sufficient approximation.

Case 3: If the previous cases do not apply then we have $p(\OPT[1]) > p(\OPT)/3$ and $p(\OPT[2]) < (\frac 13 - \eps)p(\OPT)$. This allows to guess a number $\hp$ with $p(\OPT[2]) \le \hp < p(\OPT[1])$. (More precisely, by considering all powers of $1 + \frac \eps{10}$ in a small range around the maximum profit of any feasible item, we obtain a list of $O(1/\eps)$ candidate numbers, one of which separates $p(\OPT[2])$ and $p(\OPT[1])$.) On the items of profit at most $\hp$ use the same argument as in Case 1: the Slack-Generating Lemma shows that the Approximation Lemma computes a virtual item $v$ that can act as a surrogate for $\OPT[2 \semi]$, while losing only profit $\approx p(\OPT[2])$. By using meet in the middle with one item of profit more than~$\hp$ in the left half and a virtual item in the right half we find a solution that is at least as good as $\{\OPT[1],v\}$, which has profit $\approx p(\OPT) - p(\OPT[2]) \approx \frac 23 p(\OPT)$. In all cases we obtain the desired approximation.

\medskip
For a detailed description of our algorithm see the pseudocode in \Cref{alg:linear}. 
\begin{algorithm}[t]
	\caption{$\linear(I,\capacity,\eps)$}
	
	\label{alg:linear}
	Let $\hp \gets \max\{ p(i) \mid i \in I, \weight(i) \le \capacity\}$ and $\best\gets \emptyset$. 
	
	Initialize an ORS data structure (\Cref{lem:kvors}) storing an object $S\subseteq I$ with point $\bar x_S \coloneqq \embed(S)$ and profit $ p(S)$ for every $S\subseteq I$ with $\abs{S}\leq 1$. \label{linear:ors}

	\For{each $j\in \big\{0,1,\ldots, \lceil \log_{1+\eps/10} 3 \rceil\big\}$\label{linear:loop} }{
	
	Define $\tp_j = \left(1+\frac{\eps}{10}\right)^j \cdot \frac{1}{3}\cdot \hp$ and $I_j \gets \{i\in I ~|~p(i)\leq \tp_j\}$. \label{linear:threshold}
	
	Apply the procedure from \Cref{lem:ddim_dpv} to $I_j$, profit estimate $\hp$, and error parameters $\frac \eps{10}$ and $\eta \coloneqq \frac 1 {3n}$; let $\cL_j$ be the result.

	\For{each  $v\in \cL_j$\label{linear:inner_loop}}{
		Define  $R_v\gets[0,\capacity_1-\weight_1(v)]\times[0,\capacity_2-\weight_2(v)]\times[-\infty,\min I_j-1]$. \label{linear:range}
		
		Issue a query to the ORS with range $R_v$ and let $S$ be the returned object.
		\label{linear:query}
		
		If $S\neq \bot$ and $p(S\cup \{v\})>p(\best)$ then update $\best \leftarrow S\cup \{v\}$.  
		\label{linear:update}
		}
	}
	\For{each $i\in I$\label{linear:singleton_loop}}{
			Issue a query to the ORS with the range $\range_{\capacity}(\{i\})$; let $S$ be the returned object.
			\label{linear:singleton_query}
			
			If $S\neq \bot$ and $p(S\cup \{i\})>p(\best)$ then update $\best \leftarrow S\cup \{i\}$.
}
	
	\Return $\set(\best)$ \label{linear:return}
\end{algorithm}
We first show that the algorithm returns a feasible solution.
\begin{lemma}
	\label{lem:linear-feasibility}
	\Cref{alg:linear} returns a feasible solution.
\end{lemma}

\begin{proof}
	We show that $\set(\best)$ remains feasible throughout the execution of the algorithm. Initially, $\best=\emptyset$, so it is feasible. Consider an update of $\best$ in \Cref{linear:update}, performed during the iteration indexed by $j$. Since the update is performed only when $S\neq\bot$, the guarantee of the ORS data structure gives $\embed(S)\in R_v$. The third coordinate of $R_v$, together with $\set(v)\subseteq I_j$, ensures that $S\cap\set(v)=\emptyset$. The first two coordinates of $R_v$ imply that $\weight_t(S)+\weight_t(v)\leq\capacity_t$ for each $t\in\{1,2\}$. Therefore, $\weight_t(S\cup\{v\})\leq\capacity_t$ for each $t\in\{1,2\}$, so $\set(S\cup\{v\})$ is feasible. Hence, setting $\best\gets S\cup\{v\}$ preserves feasibility of $\set(\best)$.

	Now consider an update in \Cref{linear:singleton_loop}, corresponding to an item $i\in I$. The query in \Cref{linear:singleton_query} returns a set $S$ satisfying $\embed(S)\in\range_{\capacity}(\{i\})$. Hence, by \Cref{lem:embed_to_range}, $S$ and $\{i\}$ can be matched, and therefore $S\cup\{i\}$ is feasible. Thus, this update also preserves feasibility of $\set(\best)$. It follows that the returned solution $\set(\best)$ is feasible.
\end{proof}

The next lemma shows the approximation guarantee of the algorithm.
\begin{lemma}
	\label{lem:linear-approximation}
	For every $\eps\in\left(0,\frac{2}{3}\right)$, \Cref{alg:linear} returns a feasible solution of profit at least $\left(\frac{2}{3}-\eps\right)\cdot p(\OPT)$.
\end{lemma}

\begin{proof}
	Feasibility follows from \Cref{lem:linear-feasibility}. 		Since $\hp$ is the maximum profit of any feasible item, 
	\begin{equation}
		\label{eq:linear-hp-bounds}
		p(\OPT[1])\leq \hp\leq p(\OPT).
	\end{equation}
	Throughout the proof, $\tp_j$ and $I_j$ refer to the values defined in \Cref{linear:threshold}. Let $J=\lceil \log_{1+\eps/10}3 \rceil$. By the definition of $J$ we have
$
		\hp\leq \tp_J< \left(1+\frac{\eps}{10}\right)\hp,
	$ 
	and consequently, $\{ i \in I \mid \weight(i) \le \capacity \} \subseteq I_J$.

	We prove the approximation guarantee by considering four cases.
	
	\paragraph{Case 0: $|\OPT| \le 1$.}
	If $|\OPT| = 0$ then the initial $\best \gets \emptyset$ is already an optimal solution. 
	If $|\OPT| = 1$, then consider the iteration of \Cref{linear:singleton_loop} corresponding to $i=\OPT[1]$, that is, $\{i\} = \OPT$. Since $\emptyset$ and $\OPT$ can be matched, \Cref{lem:embed_to_range} gives $\embed(\emptyset)\in\range_{\capacity}(\OPT)$. Moreover, $\emptyset$ is stored in the ORS data structure. Therefore, the query in \Cref{linear:singleton_query} returns a set $S \ne \bot$. Hence, the algorithm finds a solution whose profit is at least
	\begin{equation*}
		p(S)+p(\{i\})\geq p(\emptyset)+p(\OPT)=p(\OPT).
	\end{equation*}
	Therefore, following this iteration it holds that $p(\best)\geq p(\OPT)$.
	
	\paragraph{Case 1: $|\OPT| \ge 2$ and $p(\OPT[1])\leq p(\OPT)/3$.}
	Consider the iteration of \Cref{linear:loop} in which $j=J$. Applying the Slack-Generating Lemma (\Cref{lem:structural}) to $\OPT$ with error parameter $\eps/{10}$ yields a set $Q\subseteq\OPT$ such that
	\begin{equation}
		\label{eq:linear-case1-Q}
		\begin{aligned}
			\weight_1(Q)&\leq\left(1-\frac{\eps}{10}\right)\weight_1(\OPT),\\
			\weight_2(Q)&\leq\weight_2(\OPT),\\
			p(Q)&\geq\left(1-\frac{\eps}{10}\right)p(\OPT)-p(\OPT[1])\geq\left(\frac{2}{3}-\frac{\eps}{10}\right)p(\OPT).
		\end{aligned}
	\end{equation}
	Since $\{ i \in I \mid \weight(i) \le \capacity \} \subseteq I_J$, we have $Q\subseteq I_J$. 
	Note that $p(Q) \ge \frac 12 \, p(\OPT) \stackrel{\eqref{eq:linear-hp-bounds}}{\ge} \frac \hp {3n} = \eta \, \hp$ and $\weight_t(Q) \le \weight_t(\OPT) \le \capacity_t \le \eta^{-1} \capacity_t$ for each $t \in \{2,\ldots,d-1\}$. Hence, by \Cref{lem:ddim_dpv}, applied with error parameters $\frac \eps{10}$ and $\eta = \frac 1{3n}$, there exists a virtual item $v\in\cL_J$ such that
	\begin{equation}
		\label{eq:linear-case1-v}
		p(v)\geq\min\left\{\left(1-\frac \eps{10} \right)p(Q), n \hp\right\},\qquad
		\weight_1(v)\leq\left(1+\frac{\eps}{10}\right)\weight_1(Q),\qquad
		\weight_2(v)\leq\weight_2(Q).
	\end{equation}
	By \eqref{eq:linear-case1-Q} and \eqref{eq:linear-case1-v}, we have $\weight_2(v)\leq\weight_2(Q)\leq\weight_2(\OPT)$ and
	\[
	\weight_1(v)
	\leq \left(1+\frac{\eps}{10}\right)\weight_1(Q)
	\leq \left(1+\frac{\eps}{10}\right)\left(1-\frac{\eps}{10}\right)\weight_1(\OPT)
	\leq \weight_1(\OPT).
	\]
	Consequently, $\weight(v)\leq\weight(\OPT)\leq\capacity$, and therefore $\set(v)$ is feasible.
		
	Consider the iteration of \Cref{linear:inner_loop} corresponding to this virtual item $v$. Since $\set(v)$ is feasible, $\embed(\emptyset)\in R_v$. As the empty set is stored in the ORS data structure, the query in \Cref{linear:query} returns $S\neq\bot$. It remains to simplify the profit guarantee in \eqref{eq:linear-case1-v}. Since $\hp$ is the maximum profit of a feasible item and $|Q|\leq n$, we have $p(Q)\leq n\hp$. Thus,
	\[
	p(v) \stackrel{\eqref{eq:linear-case1-v}}{\geq} \min\left\{\left(1-\frac \eps{10}\right) p(Q), n \hp\right\} = \left(1-\frac \eps{10}\right) p(Q).
	\]
	It follows that
	\begin{align*}
		p(S\cup\{v\})
		&\geq p(v)\\
		&\geq \left(1-\frac \eps{10}\right)p(Q)\\
		&\overset{\eqref{eq:linear-case1-Q}}{\geq} \left(1-\frac \eps{10}\right)\left(\frac{2}{3}-\frac{\eps}{10}\right)p(\OPT)\\
		&\geq \left(\frac{2}{3}-\eps\right)p(\OPT).
	\end{align*}
	Therefore, after this iteration it holds that $p(\best)\geq\left(\frac{2}{3}-\eps\right)p(\OPT)$.
	
	\paragraph{Case 2: $|\OPT| \ge 2$ and $p(\OPT[\semi 2])\geq\left(\frac{2}{3}-\eps\right)p(\OPT)$.}
	Consider the iteration of \Cref{linear:singleton_loop} corresponding to $i=\OPT[2]$. Since $\{\OPT[1]\}$ and $\{\OPT[2]\}$ can be matched, \Cref{lem:embed_to_range} gives $\embed(\{\OPT[1]\})\in\range_{\capacity}(\{\OPT[2]\})$. Moreover, $\{\OPT[1]\}$ is stored in the ORS data structure. Therefore, the query in \Cref{linear:singleton_query} returns a set $S$ satisfying $p(S)\geq p(\OPT[1])$ and $\OPT[2]\notin S$. Hence, the algorithm finds a solution whose profit is at least
	\begin{equation*}
		p(S)+p(\OPT[2])\geq p(\OPT[1])+p(\OPT[2])=p(\OPT[\semi 2])\geq\left(\frac{2}{3}-\eps\right)p(\OPT).
	\end{equation*}
	Therefore, following this iteration it holds that $p(\best)\geq\left(\frac{2}{3}-\eps\right)p(\OPT)$.

	\paragraph{Case 3: neither of the previous cases applies.}
	That is, we have
	\begin{align}
		\label{eq:linear-case3-profit-gap}
			p(\OPT[1])&>\frac{p(\OPT)}{3} \qquad \text{ and } \qquad
			p(\OPT[1])+p(\OPT[2])<\left(\frac{2}{3}-\eps\right)p(\OPT),
	\end{align}
	and hence
	\begin{align} \label{eq:pOtstwostuff}
		p(\OPT[2])&<\left(\frac{1}{3}-\eps\right)p(\OPT).
	\end{align}
	
	\begin{claim} \label{cla:linear-separating-threshold}
		The main loop contains an index $j$ for which
	\begin{equation}
		\label{eq:linear-separating-threshold}
		p(\OPT[2])\leq\tp_j<p(\OPT[1]).
	\end{equation}
	\end{claim}
	\begin{claimproof}
	Let $j$ be the largest integer satisfying
	\begin{equation}
		\left(1+\frac{\eps}{10}\right)^j\hp \le p(\OPT).
	\end{equation}
By \eqref{eq:linear-hp-bounds} we have 
	$\hp \le p(\OPT)$, and hence $j\geq 0$. On the other hand, \eqref{eq:linear-hp-bounds} and \eqref{eq:linear-case3-profit-gap} give $\hp\geq p(\OPT[1])>p(\OPT)/3$. Since $\left(1+\frac{\eps}{10}\right)^J\geq3$, it follows that $\left(1+\frac{\eps}{10}\right)^J\hp>p(\OPT)$, and therefore $j\leq J-1$. Thus, $j$ is an index of the main loop (\Cref{linear:loop}).

	By the definition of $\tp_j$ and the choice of $j$,
	\begin{equation*}
		\tp_j=\frac{1}{3}\left(1+\frac{\eps}{10}\right)^j\hp
		\le \frac{p(\OPT)}{3}<p(\OPT[1]).
	\end{equation*}
	Moreover, the maximality of $j$ implies $\left(1+\frac{\eps}{10}\right)^{j+1}\hp > p(\OPT)$, and hence
	\begin{equation*}
		\tp_j=\frac{1}{3}\frac{\left(1+\frac{\eps}{10}\right)^{j+1}\hp}{1+\frac{\eps}{10}} > \frac{p(\OPT)}{3\left(1+\frac{\eps}{10}\right)}
		>\left(\frac{1}{3}-\eps\right)p(\OPT) \stackrel{\eqref{eq:pOtstwostuff}}{>} p(\OPT[2]),
	\end{equation*}
	where the second inequality holds for every $\eps>0$. This proves the claim.
	\end{claimproof} 
	
	In the iteration $j$ guaranteed by \Cref{cla:linear-separating-threshold}, it holds that $\OPT[1]\notin I_j$ and $L\coloneqq \OPT[2\semi]\subseteq I_j$. This also implies that $\OPT[1]< \min I_j$.
	By the Slack-Generating Lemma (\Cref{lem:structural}) with error parameter $\eps/10$, there is a set $Q\subseteq L$ such that
	\begin{equation}
		\label{eq:linear-case3-Q}
		\weight_1(Q)\leq\left(1-\frac{\eps}{10}\right)\weight_1(L),\qquad
		\weight_2(Q)\leq\weight_2(L),\qquad
		p(Q)\geq\left(1-\frac{\eps}{10}\right)p(L)-p(\OPT[2]).
	\end{equation}
	Note that $\weight(Q) \le \weight(L) \le \weight(\OPT) \le \capacity \le \eta^{-1} \capacity$ and
	\[
		p(Q)\geq\left(1-\frac{\eps}{10}\right)p(L)-p(\OPT[2])
		\ge p(\OPT) - p(\OPT[1]) - p(\OPT[2]) - \frac \eps {10} \, p(\OPT) \stackrel{\eqref{eq:linear-case3-profit-gap}}{\ge} \frac 13 \, p(\OPT) \stackrel{\eqref{eq:linear-hp-bounds}}{\ge} \eta \, \hp.
	\]
	Hence, by \Cref{lem:ddim_dpv}, applied with error parameters $\frac \eps{10}$ and $\eta = \frac 1{3n}$, there exists $v\in\cL_j$ such that
	\begin{equation}
		\label{eq:linear-case3-v}
		p(v)\geq\min\left\{ \left(1-\frac \eps{10}\right) p(Q),n \hp\right\},\qquad
		\weight_1(v)\leq\left(1+\frac{\eps}{10}\right)\weight_1(Q),\qquad
		\weight_2(v)\leq\weight_2(Q).
	\end{equation}
	By \eqref{eq:linear-case3-Q} and \eqref{eq:linear-case3-v}, we have $\weight_2(v)\leq\weight_2(Q)\leq\weight_2(L)$ and
	\[
	\weight_1(v)
	\leq\left(1+\frac{\eps}{10}\right)\weight_1(Q)
	\leq\left(1+\frac{\eps}{10}\right)\left(1-\frac{\eps}{10}\right)\weight_1(L)
	\leq\weight_1(L).
	\]
	Consequently, $\weight(v)\leq\weight(L)$. Thus, for each $t\in\{1,2\}$,
	\begin{equation*}
		\weight_t(\OPT[1])+\weight_t(v)
		\leq\weight_t(\OPT[1])+\weight_t(L)
		=\weight_t(\OPT)
		\leq\capacity_t.
	\end{equation*}
	Moreover, $\OPT[1]<\min(I_j)$, so $\max(\{\OPT[1]\})\leq\min(I_j)-1$. By the definition of $R_v$ in \Cref{linear:range}, these inequalities imply that $\embed(\{\OPT[1]\})\in R_v$. Since $\{\OPT[1]\}$ is stored in the ORS structure, the query returns a set $S \ne \bot$ with $p(S)\geq p(\OPT[1])$. Furthermore, $\embed(S)\in R_v$ and $\set(v)\subseteq I_j$, so $\max(S)\leq\min(I_j)-1<\min(\set(v))$. Thus, $S$ and $\set(v)$ are disjoint, and therefore $p(S\cup\{v\})=p(S)+p(v)$.
	Since $\hp$ is the maximum profit of a feasible item and $Q\subseteq I_j$ is feasible, we have $p(Q) \le n \hp$ and thus
	\begin{equation}
		p(v) \stackrel{\eqref{eq:linear-case3-v}}{\geq} \min\left\{ \left(1-\frac \eps{10}\right) p(Q),n \hp\right\} = \left(1-\frac \eps{10}\right)p(Q).
	\end{equation}
	Therefore, the solution found in this iteration has profit at least
	\begin{equation*}
		\begin{aligned}
		p(S)+p(v)&\geq p(\OPT[1]) + \left(1-\frac \eps{10}\right) p(Q)\\
		&\overset{\eqref{eq:linear-case3-Q}}{\geq} p(\OPT[1])+\left(1-\frac \eps{10}\right)\left(\left(1-\frac{\eps}{10}\right)p(\OPT[2\semi])-p(\OPT[2])\right) \nonumber\\
		&\geq p(\OPT[1])+\left(1-\eps\right)p(\OPT[2\semi])-p(\OPT[2])\nonumber\\
		&\geq (1-\eps) p(\OPT)-p(\OPT[2])\nonumber\\
		&\stackrel{\eqref{eq:pOtstwostuff}}{\geq} \frac{2}{3}\, p(\OPT).
		\end{aligned}
	\end{equation*}
	Therefore, from this iteration onward it holds that $p(\best)\geq \frac{2}{3} \, p(\OPT)$.

	In all four cases the algorithm finds a feasible solution of profit at least $\left(\frac{2}{3}-\eps\right)p(\OPT)$, and the lemma follows.
\end{proof}

Having established the approximation guarantee, we now turn to the running time analysis.
\begin{lemma}
	\label{lem:linear-runtime}
	For every $\eps\in\left(0,\frac{2}{3}\right)$, \Cref{alg:linear} runs in time
	$n\cdot\big(\frac{\log n}{\eps}\big)^{O(1)}$.
\end{lemma}

\begin{proof}
	The initialization of the ORS data structure in \Cref{linear:ors} inserts $n+1$ objects, corresponding to the empty set and the $n$ singleton sets. Each set is stored as a point of dimension three. Therefore, by \Cref{lem:kvors}, constructing the data structure takes time $O(n\log^2 n)$, and each query takes time $O(\log^3 n)$.

	The number of iterations of the main loop in \Cref{linear:loop} is $\lceil \log_{1+\eps/10}3 \rceil +1=O(1/\eps)$. In each iteration, the set $I_j$ can be constructed in time $O(n)$. By \Cref{lem:ddim_dpv}, generating $\cL_j$ with error parameters $\frac \eps{10}$ and $\eta = \frac 1{3n}$ takes time $\abs{I_j}\cdot(\log(\abs{I_j}/\eta)/\eps)^{O(1)}\leq n\cdot(\log (n)/\eps)^{O(1)}$. In particular, the size of $\cL_j$ is bounded by the same expression. Thus, over all iterations of the main loop, generating the lists and performing one ORS query for every $v\in\cL_j$ takes time $n\cdot(\log(n)/\eps)^{O(1)}$.
	The range computation in each iteration of \Cref{linear:inner_loop} takes constant time, since the weights and profit of every virtual item are stored explicitly. Moreover, every set $S$ returned by the ORS has size at most one. The last coordinate of $R_v$, together with $\set(v)\subseteq I_j$, ensures that $S\cap\set(v)=\emptyset$. Hence, $p(S\cup\{v\})=p(S)+p(v)$ can be computed in time $O(1)$, and the update also takes constant time.

	Finally, the loop in \Cref{linear:singleton_loop} performs $n$  queries to the ORS. Computing each range takes constant time. Every returned set $S$ has size at most one and, by \Cref{lem:embed_to_range}, is disjoint from $\{i\}$. Therefore, $p(S\cup\{i\})=p(S)+p(i)$ can be computed in time $O(1)$, and processing the update also takes constant time. Thus, this loop takes time $O(n\log^3 n)$. 
	
	The time needed to reconstruct the returned set in \Cref{linear:return} is $O(n)$.
	
	 Combining all parts, the total running time of \Cref{alg:linear} is $n\cdot(\log(n)/\eps)^{O(1)}$.
\end{proof}

 \Cref{thm:linear} follows immediately from  \Cref{lem:linear-approximation,lem:linear-runtime}.

\section{Reduction to \texorpdfstring{$k$-SUM}{k-SUM}}
\label{sec:reduction_to_ksum}

In this section, we prove \Cref{thm:2reductiontokSUM} by designing
an approximation scheme for $d$-Knapsack that uses a $k$-SUM algorithm
as a black box. The algorithm follows the meet-in-the-middle approach
of \Cref{sec:ddim_improved}, replacing the step based on
orthogonal range queries by calls to a $k$-SUM
algorithm.

Throughout this section, by $O_k$ we hide a function depending only on $k$ (similarly for $O_{k,d}$ etc.), and by $O$ we hide an absolute constant that is independent of $k$ (as well as of other parameters such as $n$, $d$, etc.).
To simplify notation, we assume access to an algorithm that can solve $k$-SUM for any given $k$, and runs in time $\Oh_k(n^{\alpha_k})$ for some constant $\alpha_k \ge 1$ depending only on $k$. Later, in the proof of \Cref{thm:2reductiontokSUM}, we remove this assumption and argue that an algorithm solving $k$-SUM for a specific $k$ suffices.

The reduction works as follows.
As in \Cref{sec:ddim_improved}, we construct a small family of
representative solutions to serve as surrogates for the low-profit
items of  $\OPT[k+3\semi]$. After fixing a representative and guessing the
items in $\OPT[k+1\semi k+2]$ by enumeration,
the remaining task is to find at most $k$ items satisfying the
residual capacity and profit requirements.

To solve this variant of $d$-Knapsack, we use a vector version of
$k$-SUM in which the selected vectors must satisfy coordinatewise
inequalities. 
Standard techniques from fine-grained complexity reduce inequality constraints to equality constraints, see, e.g.,~\cite{VassilevskaWW13,NederlofLZ12,GokajK25}, with overhead polynomial in the number
of items and the encoding length of the input integers.
This almost yields a reduction to $k$-SUM, except that, because we work on the RealRAM model, our input numbers a priori do not have bounded encoding length. 
Therefore, we first apply the Frank--Tardos Lemma~\cite{FrankT87}, which replaces real numbers by integers of bounded encoding length while
preserving the comparisons needed by the reduction.

In the next three subsections we gather these tools:
\Cref{sec:ksum_representatives} discusses representative solutions,
\Cref{sec:ksum_compression} uses Frank-Tardos to reduce to integers, and \Cref{sec:ksum_inequalities} reduces
the special case of $d$\nobreakdash-Knapsack where we want to select at most $k$ items to Inequality Vector $k$-SUM and further to $k$-SUM. Finally,
\Cref{sec:ksum_algorithm} combines these ingredients and proves
\Cref{thm:2reductiontokSUM}.

\subsection{A  Family of Representative Solutions}
\label{sec:ksum_representatives}
Our reduction from $d$-Knapsack to $k$-SUM uses representative solutions as in \Cref{sec:ddim_improved}. The following lemma is an immediate consequence of  \Cref{lem:low_profit}.
\begin{lemma}[Simple Representative Solutions]
\label{lem:ksum_representatives}
There is an algorithm that, given  a set of items
$I=\{1,\ldots,n\}$ such that $p(1)\geq p(2)\geq \ldots \geq p(n)$, 
and $\delta\in(0,1)$, computes in time $(n/\delta)^{O(d)}$
families $\cL_1,\ldots,\cL_{n+1}$, each of size
$(n/\delta)^{O(d)}$, with the following properties:
\begin{enumerate}[nosep]
 \item\label{prop:ksum_representatives_empty} $\emptyset\in\cL_i$ for every $i\in[n+1]$.
 \item\label{prop:ksum_representatives_suffix} $S\subseteq\{i,\ldots,n\}$ for every $S\in\cL_i$ and $i\in \{1,\ldots ,n+1\}$.
 \item\label{prop:ksum_representatives_guarantee} For every $Q\subseteq\{i,\ldots,n\}$ there is $S\in\cL_i$
 with $\weight(S)\leq\weight(Q)$ and
 $$p(S)\geq(1-\delta)p(Q)-(d-1+\delta/n)p(i),$$
 where $p(n+1)=0$.
\end{enumerate}
\end{lemma}

While \Cref{lem:ksum_representatives} follows from
\Cref{lem:low_profit}, it also has a direct proof that is much simpler than the proof of \Cref{lem:low_profit}. In particular, the complex dynamic programming algorithm underlying
\Cref{thm:simple-apx-data-structure} that is used in the proof of \Cref{lem:low_profit} can be replaced by a much simpler dynamic program. We sketch this simpler proof below.

\begin{proof}[Sketch of simpler proof]
First, show a variant of
\Cref{lem:ddim_dp} that minimizes the weight of the returned set  in dimension $d$, while attaining a profit above a given threshold and  preserving given weight bounds in the other dimensions. More precisely, this variant takes as input a set $I$ of items, capacities $\capacity_1,\ldots,\capacity_{d-1}$, a profit
threshold $\tilde p\geq0$, and an error parameter $\eps\in(0,1)$.
It returns a subset $T\subseteq I$ with
$\weight_t(T)\leq\capacity_t$ for each $t\in[d-1]$ and
$p(T)\geq(1-\eps)\tilde p$, whose  weight in the $d$-th dimension is no greater than
that of any subset $R\subseteq I$ satisfying
$\weight_t(R)\leq(1-\eps)\capacity_t$ for each $t\in[d-1]$ and
$p(R)\geq\tilde p$. The variant is based on rounding weights and profits, as in
\Cref{lem:ddim_dp}, and can be implemented in time
$(n/\eps)^{O(d)}$. Its analysis uses standard arguments.
% The variant can be used to generate the subset in $\cL_i$, while only $(n/\delta)^{O(d)}$ solutions produced by this variant suffice to satisfy the Property~\eqref{prop:ksum_representatives_guarantee} for every potential set $Q$.

To construct the representative families, one can run this variant 
with a small collection of capacity and profit guesses. The key
observation is that any  $Q\subseteq \{i,\ldots,n\}$ contains a subset
$S\subseteq Q$ of at most $d$ items such that  $p(S)\leq p(Q)\leq n\cdot p(S)$ and $ \weight_t(S)\leq \weight_t(Q)\leq n\cdot \weight_t(S)$ for each $t \in [d-1]$. 
 Once $S$ is guessed, we can guess values $\capacity_t$ and $\hat p$ which approximate $\weight_t(Q)$ and $p(Q)$ up to a multiplicative factor of $(1+\eps)$ for every $t \in [d-1]$ with only $(\log(n)/\delta)^{O(d)}$ additional guesses per $S$. The Slack-Generating lemma (\Cref{lem:structural})  provides a subset of $Q$
with enough room to accommodate the errors in these guesses and in
the dynamic program. For the corresponding call (whose input values are derived from the guesses and $\delta$), the dynamic program
returns a set whose weights do not exceed those of $Q$, while losing
only the allowed amount of profit. Collecting the outputs of all
calls  gives the desired representative family.
\end{proof}

\subsection{Frank-Tardos Reduction}
\label{sec:ksum_compression}

In order to apply standard fine-grained techniques to replace inequality constraints by equality constraints, we first need to replace real input numbers by integers with polynomially many bits.
This is achieved by the Frank-Tardos Lemma~\cite{FrankT87}, which replaces real coefficients with
small integers while preserving the signs of bounded linear combinations. 
This lemma is a standard tool in optimization, and has also recently been used in fine-grained complexity~\cite{FischerPS26}.
We use the variant of the Frank-Tardos lemma as stated in~\cite[Theorem~1.4 and Section~3.1]{FischerPS26}.

\begin{lemma}[Frank-Tardos Lemma]
\label{thm:ksum_frank_tardos}
Given $u\in\mathbb R^m$, in $m^{O(1)}$ time on the RealRAM we can
compute $v\in\mathbb Z^m$ with $\|v\|_\infty\leq2^{O(m^3)}$
such that $\sign(u\cdot x)=\sign(v\cdot x)$ for every
$x\in\{-1,0,1\}^m$.
\end{lemma}
By taking $x$ to be the indicator vector with a $1$ in the $i$-th position and $0$ in all others, we get the following property of the vectors produced by Frank-Tardos.
\begin{observation}
	\label{obs:FrankTardos}
	Given a vector $u\in \mathbb{R}^m$ of non-negative numbers,  \Cref{thm:ksum_frank_tardos} returns a vector $v\in \mathbb{Z}^m$ of non-negative integers.  
	\end{observation}
The following lemma uses Frank--Tardos to convert a $d$-Knapsack
instance into one where all numbers are integers with polynomial bit complexity, while preserving
feasibility and comparisons of the profit with both $T$ and $\alpha T$,
for a given rational $\alpha\in(0,1)$.

\begin{lemma}[Threshold-Preserving Compression]
\label{lem:ksum_compression}
There is an algorithm that, given a rational $\alpha=a/b$ with
positive integers $a<b$, a $d$-Knapsack instance $(I,\capacity)$
with $n$ items, and $T \geq 0$, computes in $(nd\log b)^{O(1)}$ time
on the RealRAM a $d$-Knapsack instance $(I',\capacity')$, a threshold $T'\geq 0$,
and a bijection $\sigma:I\to I'$ such that, for every $A\subseteq I$:
\begin{itemize}[nosep]
 \item $A$ is feasible if and only if $\sigma(A)$ is feasible.
 \item $p(A)\geq T$ if and only if $p(\sigma(A))\geq T'$.
 \item $p(\sigma(A))\geq\alpha T'$ if and only if $p(A)\geq\alpha T$.
 \item Each of the numbers $T'$ and $\alpha T'$ and every weight, profit, and capacity in the instance $(I',\capacity')$ is a nonnegative $\Oh((n+\log b)^3)$-bit integer.
\end{itemize}
\end{lemma}
\begin{proof}
Index the items by $I=[n]$, and set $\ell=\lfloor\log_2 b\rfloor$.
Construct the compressed instance as follows.
\begin{itemize}[nosep]
 \item For every $t\in[d]$, apply \Cref{thm:ksum_frank_tardos} to
 $
 u^{(t)}=(\weight_t(1),\ldots,\weight_t(n),\capacity_t),
 $
 and denote the resulting integer vector by
 $
 v^{(t)}=(w'_{t,1},\ldots,w'_{t,n},c'_t).
 $
 \item Apply \Cref{thm:ksum_frank_tardos} to
 \[
 u^{(p)}=
 \bigl(p(1),\ldots,p(n)\bigr)
 \mathbin{\Vert}
 \left(\frac{T}{b},\frac{T}{b},\frac{2T}{b},\ldots,
       \frac{2^\ell T}{b}\right),
 \]
 where $\Vert$ denotes vector concatenation. The second part of the vector
 consists of an initial copy of $T/b$ followed by the entries $2^i T/b$ for $i=0,\ldots,\ell$.
 Denote the resulting integer vector by
 \[
 v^{(p)}=(p'_1,\ldots,p'_n)\mathbin{\Vert}(r',q'_0,\ldots,q'_{\ell}).
 \]
 Here $r'$ corresponds to the first copy of $T/b$, and $q'_i$
 corresponds to $2^iT/b$.
 \item Set $I'=\{j':j\in[n]\}$ and define
 \[
 \sigma(j)=j',\qquad p(j')=p'_j,\qquad
 \weight_t(j')=w'_{t,j},\qquad \capacity'_t=c'_t.
 \]
 Finally, set $T'=b\cdot q'_0$.
\end{itemize}

By construction, $\sigma:I\to I'$ is a bijection.
Fix $A\subseteq I$. Let $\mathbf 1_A\in\{0,1\}^n$ denote the indicator vector of
$A$: its $j$th coordinate is $1$ if $j\in A$ and $0$ otherwise. We prove the three required properties in the lemma. 

\medskip
\noindent\emph{Feasibility.}
Set $x=\mathbf 1_A \mathbin{\Vert} (-1)$ (i.e., the concatenation of the vector $\mathbf 1_A$ with the vector of dimension~$1$ whose only entry is $-1$).
Clearly, $x\in\{-1,0,1\}^{n+1}$. For every $t\in [d]$ it holds that 
\[
u^{(t)}\cdot x=\sum_{j\in A}\weight_t(j)-\capacity_t,
\qquad
v^{(t)}\cdot x=\sum_{j\in A}\weight_t(\sigma(j))-\capacity'_t.
\]
By \Cref{thm:ksum_frank_tardos} we have that  $\sign(u^{(t)} \cdot x) =\sign(v^{(t)}\cdot x)$, and therefore
$
\sum_{j\in A}\weight_t(j)\leq\capacity_t
$ if and only if $
\sum_{j\in A}\weight_t(\sigma(j))\leq\capacity'_t.$
Therefore, $A$ is feasible if and only if $\sigma(A)$ is feasible. 

\medskip
In order to show the additional properties 
we use the following claim.
\begin{claim}
	\label{clm:ksum_compression_auxiliary}
	It holds that  $r'=q'_0$ and
	$q'_i=2^i q'_0$ for every $i\in\{0,\ldots,\ell\}$.
\end{claim}
\begin{claimproof}
The profit vector $u^{(p)}$ has two entries of value $T/b$ which correspond to the entries $r'$ and $q'_0$ of $v^{(p)}$. Therefore, by  \Cref{thm:ksum_frank_tardos} it must hold that  $r'=q'_0$.
For every $i\in\{1,\ldots,\ell\}$, we also have
\[
 \frac{2^iT}{b}-\frac{T}{b}
       -\sum_{h=0}^{i-1}\frac{2^hT}{b}=0.
\]
Use the coordinate corresponding to $r'$ for the standalone term
$T/b$, and those corresponding to $q'_0,\ldots,q'_{i-1}$ for the
sum. Every coefficient is in $\{-1,0,1\}$, so \Cref{thm:ksum_frank_tardos} gives
$
 q'_i=r'+\sum_{h=0}^{i-1}q'_h.$  We use a simple induction to show  that $q'_i=2^i q'_0$ for every $i$:
the case $i=0$ is immediate, and for $i\geq1$ the induction
hypothesis gives
$
 q'_i=q'_0+\sum_{h=0}^{i-1}2^h q'_0=2^i q'_0.
$
This proves the claim.
\end{claimproof}

\medskip
\noindent\emph{Comparisons.}
For each integer $q\in\{0,\ldots,b\}$, let
$\gamma_0(q),\ldots,\gamma_{\ell}(q)\in\{0,1\}$ be its binary
digits, so that $q=\sum_{i=0}^\ell\gamma_i(q)2^i$.
For $z\in \{a,b\}$, define the coefficient vector
\[
x_z=\mathbf 1_A\mathbin{\Vert}(0,-\gamma_0(z),\ldots,-\gamma_{\ell}(z))
\in\{-1,0,1\}^{n+\ell+2}.
\]
Its scalar products with the input and output vectors are
\begin{align*}
	u^{(p)}\cdot x_z
	&=p(A)-\sum_{i=0}^\ell\gamma_i(z)\frac{2^iT}{b}
	=p(A)-\frac{z\cdot T}{b}, \quad \text{and}\\
	v^{(p)}\cdot x_z
	&=p(\sigma(A))-\sum_{i=0}^{\ell}\gamma_i(z)q'_i\\
	&=p(\sigma(A))-zq'_0  &\text{(by \Cref{clm:ksum_compression_auxiliary})} \\
	&=p(\sigma(A))-\frac{zT'}{b}.  &\text{(by the definition of $T'$)}
\end{align*}
%The second equality in the second equation follows from \Cref{clm:ksum_compression_auxiliary}.
By \Cref{thm:ksum_frank_tardos}, 
$\sign(u^{(p)}\cdot x_z ) =\sign(v^{(p)}\cdot x_z )$,
and therefore  $p(A)\geq \frac{z\cdot T}{b}$ if and only if
$p(\sigma(A))\geq \frac{z\cdot T'}{b}$, for $z\in \{a,b\}$. Taking $z=b$ this implies that $p(A)\geq T$ if and only if $p(\sigma(A))\geq T'$. Similarly, taking $z=a$ this implies that $p(A) \geq \frac{a\cdot T}{b} = \alpha T$ if and only if $p(\sigma(A))\geq \frac{a\cdot T'}{b} = \alpha T'$.

\medskip
\noindent\emph{Nonnegativity, encoding bounds, and running time.}
By \cref{obs:FrankTardos} all the weights, profits, and capacities of the produced instance are non-negative. 
It remains to analyze the running time and the bit complexity of the returned numbers. 
The weight vectors $u^{(t)}$ have $n+1$ coordinates each, and the profit vector $u^{(p)}$ has $n+\ell+2$ coordinates. Thus, by \Cref{thm:ksum_frank_tardos}, every
output coordinate of Frank-Tardos is an $O((n+\ell+2)^3)$-bit integer. The thresholds
$T'=bq'_0$ and $\alpha T' = a q'_0$ are $O((n+\ell+2)^3+\log b)$-bit integers.
Since $\ell = \lfloor \log_2 b \rfloor$, every constructed number is an $O((n+\log b)^3)$-bit integer.
The construction uses $d+1$ calls to Frank-Tardos on vectors of
dimension $O(n+\log b)$.
Therefore, the total running time on the RealRAM is
$(nd\log b)^{O(1)}$.
\end{proof}

\subsection{Knapsack with Bounded Solution Size and Inequality Vector \boldmath$k$-SUM}
\label{sec:ksum_inequalities}

In this section we introduce a variant of $d$-Knapsack in which the solutions size is restricted to be at most $k$, for a given parameter $k$. 
We also define a variant of $k$-SUM on vectors with inequality constraints that we call Inequality Vector $k$-SUM.
Then we reduce the former to the latter, and the latter to $k$-SUM.

\begin{definition}[$d$-Knapsack with Bounded Solution Size]
\label{def:ksum_cardinality_knapsack}
The input consists of a $d$-Knapsack instance $(I,\capacity)$, in which each weight, profit, and capacity is an integer in $\{0,\ldots,\Delta\}$ (for some $\Delta \ge 2$), a profit threshold $T \in \{0,\ldots,\Delta\}$, and an integer
$k\geq1$. The task is to find a subset $S\subseteq I$ with
$|S|\leq k$, $\weight(S)\leq\capacity$, and $p(S)\geq T$,
or report that no such subset exists.
\end{definition}

%The above problem can be easily cast as the following vector variant of $k$-SUM.

\begin{definition}[$D$-dimensional Inequality Vector $k$-SUM]
\label{def:ksum_inequality_vector}
The input is a collection of $n$ vectors $\veca_1,\ldots,\veca_n\in \{-\Delta,\ldots,\Delta\}^D$ and a bound vector $\vecs\in \{-\Delta,\ldots,\Delta\}^D$ (for some $\Delta \ge 2$).
The task is to find indices $1\leq i_1<\cdots < i_k\leq n$ satisfying $\sum_{j=1}^{k} \veca_{i_j} \leq \vecs$ or report that none exist.
\end{definition}

%The above problem can be easily cast 
%as the following vector variant of $k$-SUM.
%In the {\em Dimension $D$ Inequality Vector $k$-SUM} problem the input is a collection of $n$ vectors $\veca_1,\ldots,\veca_n\in \mathbb Z^D$ and a bound vector $\vecs\in\mathbb Z^D$.
%The objective is to find indices $1\leq i_1<i_2<\cdots < i_k\leq n$ satisfying $\sum_{j=1}^{k} \veca_{i_j} \leq \vecs$ or report that none exists.  

We stress the difference that all previously defined problems in this paper have real input numbers, while the above two problems have integer input numbers bounded by $\Delta$. We nevertheless design algorithms for these problems on the RealRAM model, because they are intermediate problems in our reduction from the $d$-Knapsack problem (on real numbers) to the $k$-SUM problem (on real numbers).

$D$-Dimensional Inequality Vector $k$-SUM can be solved using an algorithm for $k$-SUM:

\begin{lemma}
\label{lem:ksum_inequalities}
Suppose $k$-SUM can be solved in time
$O_k(n^{\alpha_k})$. Then $D$-dimensional Inequality Vector $k$-SUM can be solved in time 
$O_{k,D}\bigl(n^{\alpha_k}(n\log\Delta)^{D+O(1)}\bigr)$.
\end{lemma}

The proof of \Cref{lem:ksum_inequalities} uses a standard technique from fine-grained complexity that converts inequality constraints to equality constraints, see, e.g.,~\cite{VassilevskaWW13,NederlofLZ12,GokajK25}. This technique crucially uses that the input numbers are integers in a small range, due to the $(\log \Delta)^{D+O(1)}$ factor in the running time.
Because we could not find this exact variant in the literature, we include
the proof of \cref{lem:ksum_inequalities} in \Cref{sec:ksum_ineq_proof} for completeness.

We next use the above algorithm to solve $d$-Knapsack with Bounded Solution Size.

\begin{lemma}
\label{lem:ksum_cardinality_knapsack}
Suppose $k$-SUM can be solved in time $O_k(n^{\alpha_k})$. Then $d$-Knapsack with Bounded Solution Size
can be solved on the RealRAM in time
$O_{k,d}\bigl(n^{\alpha_k}(n\log\Delta)^{d+O(1)}\bigr)$.
\end{lemma}
\begin{proof}
Index the original items by $I=[n]$. The algorithm is as follows.
\begin{enumerate}[nosep]
 \item Add $k$ distinct dummy items with indices $n+1,\ldots,n+k$.
 Give each dummy item profit zero and weight zero in every dimension.
 Write $m=n+k$ for the number of items after this step.
 \item For every $j\in[m]$, define
 \[
  \veca_j=(\weight_1(j),\ldots,\weight_d(j),-p(j))
       \in\mathbb Z^{d+1},\qquad
  \vecs=(\capacity_1,\ldots,\capacity_d,-T).
 \]
 \item Apply \Cref{lem:ksum_inequalities} to
 $\veca_1,\ldots,\veca_m,\vecs$ with dimension $D=d+1$ and
 cardinality $k$. If the call reports failure, report failure.
 \item Otherwise, let $X\subseteq[m]$ be the returned index set,
 and return $S=X\cap[n]$.
\end{enumerate}
For any $X\subseteq[m]$, the first $d$ coordinates of
$\sum_{j\in X}\veca_j$ are its total weights, and the last coordinate
is $-p(X)$. Thus the first $d$ inequalities encode the capacity
constraints, and the last inequality, $-p(X)\leq-T$, encodes
the profit requirement. The following properties are easy to verify.
\begin{enumerate}[nosep]
 \item If the input has a solution $R$, then adding $k-|R|$
 distinct dummy items to $R$ gives a set $X$ of exactly $k$ items
 satisfying $\sum_{j\in X}\veca_j\leq\vecs$, since the added
 items have zero weight and profit.
 \item If $X=\{i_1,\ldots,i_k\}$ solves the constructed Inequality
 Vector $k$-SUM instance, then $S=X\cap[n]$ has at most $k$
 items and satisfies $\weight(S)=\weight(X)\leq\capacity$ and
 $p(S)=p(X)\geq T$. Hence $S$ solves the input instance.
\end{enumerate}
Together with the guarantee of \Cref{lem:ksum_inequalities}, these
properties establish correctness.

All coordinates of the constructed vectors
are integers of absolute value at most $\Delta$. By
\Cref{lem:ksum_inequalities}, the call takes time
\[
 O_{k,d}\bigl((n+k)^{\alpha_k}((n+k)\log\Delta)^{d+1+O(1)}\bigr)
 =O_{k,d}\bigl(n^{\alpha_k}(n\log\Delta)^{d+O(1)}\bigr).
\]
Constructing the instance and removing dummy items take additional time
$O(d(n+k))$, which is absorbed by this bound.
\end{proof}

\subsection{The Approximation Algorithm}
\label{sec:ksum_algorithm}

Before proving \Cref{thm:2reductiontokSUM}, as a first step we design an approximate decision algorithm for $d$-Knapsack. This algorithm is given a profit threshold $T$, and it either returns $\bot$ or a solution of profit at least $\beta T$. If there exists a solution of profit at least $T$, then the algorithm returns a solution of profit at least $\beta T$. 
Here, the approximation factor $\beta$ is carefully chosen for the application of this algorithm below, specifically $\beta$ is slightly larger than $1 - \frac{d-1}{k+1}$ for a given parameter $k$.

The algorithm uses the ingredients developed in previous sections: the representative solutions of \Cref{lem:ksum_representatives}, the Frank-Tardos reduction of \Cref{lem:ksum_compression}, and the algorithm for $d$-Knapsack with Bounded Solution Size of \Cref{lem:ksum_cardinality_knapsack}. 
The pseudocode is given in \Cref{alg:ksum_approximation}.

The proof of \Cref{thm:2reductiontokSUM} at the end of this section then applies the approximate decision algorithm for several potential values of $T$.   

\begin{algorithm}[h]
 \caption{Approximate decision algorithm $(I,\capacity,T,k)$}
 \label{alg:ksum_approximation}
 Set $\delta\gets\frac{d-1}{(k+1)(k+3)}$ and $\beta\gets (1+\delta)(1-\frac{d-1}{k+1})$

 Apply \Cref{lem:ksum_compression} with threshold $T$ and ratio
 $\beta$, obtaining $(I',\capacity')$, $T'$, and $\sigma$
 
 Set $\tau\gets\lceil\beta T'\rceil$

 Sort $I'$ by non-increasing  profits and index it by $[n]$

 Apply \Cref{lem:ksum_representatives} to $I'$ with error parameter $\delta$,
 obtaining $\cL_1,\ldots,\cL_{n+1}$

 \For{each $i\in[n+1]$, $S\in\cL_i$, and $B\subseteq\{1,\ldots,i-1\}$ with $|B|\leq2$ and $\weight(B)+\weight(S)\leq\capacity'$ \label{thresh:loop}}{
  Apply \Cref{lem:ksum_cardinality_knapsack} to items
  $\{1,\ldots,i-1\}\setminus B$, with capacities
  $\capacity'-\weight(B)-\weight(S)$, profit threshold
  $\max\{0,\tau-p(B)-p(S)\}$, and cardinality bound $k$\label{thresh:atmost} 

  If a set $H$ is returned, return $\sigma^{-1}(B\cup H\cup S)$
 }
 Return $\bot$
\end{algorithm}

For $k\geq d-1$, the parameters $\delta$ and $\beta$  set by the algorithm are rational and satisfy
$0<\delta< \frac{d-1}{k+1}$ and thus $0<\beta<1$. 
We first show the properties of solutions returned by the algorithm.
\begin{lemma}
\label{lem:ksum_threshold_solution}
If \Cref{alg:ksum_approximation} returns a set $A$ then
$\weight(A)\leq\capacity$ and $p(A)\geq\beta T$.
\end{lemma}
\begin{proof}
If the algorithm returns a set $A$ then this occurred in an iteration of \Cref{thresh:loop}, say in the iteration for $i, S, B$. In this iteration, \Cref{thresh:atmost} returns a set $H$ and it holds that $A=\sigma^{-1}(B\cup H\cup S)$. 
Then we have $S\subseteq \{i,\ldots,n\}$, $B \subseteq \{1,\ldots,i-1\}$ and $H \subseteq \{1,\ldots,i-1\} \setminus B$, so $B,H,S$ are pairwise
disjoint. Furthermore, \Cref{lem:ksum_cardinality_knapsack} guarantees  that $\weight(H)\leq \capacity' - \weight(B)-\weight(S)$ and $p(H) \geq \tau-p(B)-p(S)$. Therefore, 
\[
 \weight(B\cup H\cup S)\leq\capacity',\qquad
 p(B\cup H\cup S)=p(B)+p(H)+p(S)
       \geq\tau\geq\beta T'.
\]
By \Cref{lem:ksum_compression}, the returned set
$A=\sigma^{-1}(B\cup H\cup S)$ is feasible in the original instance
and has profit at least $\beta T$.
\end{proof}

%Next, we show that if $T$ is sufficiently small, then the algorithm returns a set. 
\begin{lemma}
\label{lem:ksum_threshold_success}
If
$T\leq p(\OPT)$, then \Cref{alg:ksum_approximation} returns a set.
\end{lemma}
\begin{proof}
Let $\OPT$ be an optimal solution and put $\OPT'=\sigma(\OPT)$.
Since $p(\OPT)\geq T$, \Cref{lem:ksum_compression} ensures that 
 $\OPT'$ is feasible and  $p(\OPT')\geq T'$.
Let $i=\min(\OPT'[k+3\semi])$, or $i=n+1$ if the set is empty. We use the convention $p(n+1)=0$.
All items of $\OPT'[\semi k+2]$ precede $i$, and
$p(i)\leq p(\OPT')/(k+3)$. 
By \Cref{lem:ksum_representatives}, there is a set $S\in\cL_i$ with
$\weight(S)\leq\weight(\OPT'[k+3\semi])$ and
\[
 p(S)\geq(1-\delta)p(\OPT'[k+3\semi])
       -\left(d-1+\frac{\delta}{n}\right)p(i).
\]
Therefore,
\begin{equation*}
	\begin{aligned}
 p(\OPT'[\semi k+2])+p(S)
 &\geq p(\OPT')-\delta\bigl(p(\OPT'[k+3\semi])+p(i)/n\bigr)-(d-1)p(i)\\
 &\geq\left(1-\delta-\frac{d-1}{k+3}\right)p(\OPT'),
 \end{aligned}
\end{equation*}
where we used $p(i)\leq p(\OPT')/(k+3)$ and $p(\OPT'[\semi k+2])\geq p(i)$, which implies $p(\OPT'[k+3\semi])+p(i)/n\leq p(\OPT')$. 
By the definition of $\delta$ we can rewrite $\delta + \frac{d-1}{k+3} = \frac{d-1}{k+1} - \delta$, so we can further bound
\begin{equation*}
 p(\OPT'[\semi k+2])+p(S)
 %\geq\left(1-\delta-\frac{d-1}{k+3}\right)p(\OPT')
 \geq \left(1-\frac{d-1}{k+1} + \delta\right)p(\OPT')
 \geq (1+\delta) \left(1-\frac{d-1}{k+1}\right)p(\OPT')
 = \beta p(\OPT') \ge \beta T'.
\end{equation*}
Since the left-hand side is an integer, we even get 
\begin{equation} \label{eq:dfsooandsf}
	p(\OPT'[\semi k+2])+p(S) \ge \lceil \beta T' \rceil = \tau.
\end{equation}

Choose $B = \OPT'[k+1 \semi k+2]$, so $|B| \le 2$, and let $R = \OPT'[\semi k]$. This set $R$ satisfies $|R| \le k$,  $R \subseteq \{1,\ldots,i-1\} \setminus B$, and
\[
 \weight(R)=\weight(\OPT'[\semi k])
 \leq\capacity'-\weight(S)-\weight(B).
\]
By \eqref{eq:dfsooandsf} and $p(R) \ge 0$ we also have $p(R) \ge \max\{0, \tau - p(B) - p(S)\}$. 

Consider the iteration of the loop in \Cref{thresh:loop} for $i,S,B$. 
In this iteration, $R$ is a solution of the constructed $d$-Knapsack with Bounded Solution Size instance. 
Therefore, the call to \Cref{lem:ksum_cardinality_knapsack} succeeds and the algorithm returns a set.
\end{proof}

Finally, we evaluate the running time of \Cref{alg:ksum_approximation}. 
\begin{lemma}
\label{lem:ksum_threshold_time}
There is an absolute constant $c'$ such that, if $k$-SUM can be solved
in time $O_k(n^{\alpha_k})$, then
\Cref{alg:ksum_approximation} runs in time
$O_{k,d}(n^{\alpha_k+c'd})$ for $k\geq d-1$.
\end{lemma}
\begin{proof}
The applications of \Cref{lem:ksum_compression,lem:ksum_representatives}
in \Cref{alg:ksum_approximation} take $O_{k,d}(n^{O(d)})$ time.
The  weights, profits and capacities of $(I',\capacity')$ are integers which require $O_{k,d}(n^3)$ bits, as do
all residual capacities and profit thresholds considered by the algorithm.
By \Cref{lem:ksum_cardinality_knapsack}, each constructed instance of $d$\nobreakdash-Knapsack with Bounded Solution Size can be solved in time
$O_{k,d}(n^{\alpha_k+O(d)})$.
There are $n^{O(d)}$ choices of $i,S$ and $O(n^2)$ choices of $B$.
Combining these bounds gives the claimed running time.
\end{proof}

With all the properties of \Cref{alg:ksum_approximation} in place, we can proceed to the proof of \Cref{thm:2reductiontokSUM}. 
\begin{proof}[Proof of \Cref{thm:2reductiontokSUM}]
Let $k=\lceil(d-1)/\eps\rceil-1$, and assume that $k$-SUM can be solved in time $O(n^\alpha)$ for some $\alpha \ge 1$. We may assume without loss of
generality that $n\geq1$ and every item is individually feasible
and has positive profit. 
Let $\delta$ and $\beta$ be defined as in \Cref{alg:ksum_approximation}. 
Also, let 
$P=\max_{j\in I}p(j)$. Run \Cref{alg:ksum_approximation} on
$(I,\capacity,T,k)$ for
each $T=P(1+\delta)^h$, $h=0,1,\ldots,   \floor{\log_{1+\delta}(n) }$, 
return a set of highest  profit among the successful runs.

Since $P\leq p(\OPT)\leq nP$, there is $h\in \{0,1,\ldots, \floor{\log_{1+\delta}(n) } \}$ such that 
\begin{equation}\label{eq:ksum_threshold_guess}
 \frac{p(\OPT)}{1+\delta}\leq P\cdot (1+\delta)^h\leq p(\OPT).
\end{equation}
Let $T=P\cdot (1+\delta)^h$ for this particular $h$. 
By \Cref{lem:ksum_threshold_success}, the run for this $T$ returns
a set. By \Cref{lem:ksum_threshold_solution}, every returned set is
feasible, and the profit of the returned solution is at least
\[
 \beta\cdot  T\geq\frac{\beta}{1+\delta}\cdot p(\OPT)
 =\left(1-\frac{d-1}{k+1}\right)\cdot p(\OPT)
 \geq(1-\eps)p(\OPT),
\]
since $(d-1)/(k+1)\leq\eps$.

By  \Cref{lem:ksum_threshold_time}, each call to \Cref{alg:ksum_approximation} runs in time $O_{k,d}(n^{\alpha+c'd})=O_{d,\eps}(n^{\alpha+c'd})$, where $c'$ is an absolute constant (that is independent of $n,k,d,\eps$ and $\alpha$). 
There are $O(\log_{1+\delta}(n))=O_{d,\eps}(\log n)$ calls to \Cref{alg:ksum_approximation}, so the total running time  is $O_{d,\eps}(n^{\alpha+cd})$, where $c$ again is an absolute constant, as required.
\end{proof}

\paragraph{AI Declaration}
OpenAI Codex assisted the development of the scientific content of \Cref{sec:reduction_to_ksum} and its writeup, as well as general text editing, proofreading, and creating \Cref{fig:exponentplots}.
AI tools were not used to develop the scientific content of \Cref{sec:ddim,sec:ddim_improved,sec:linear2d}.

\bibliographystyle{alpha}
\bibliography{refs}

\appendix 

\section{Proof of \Cref{cor:2Dtime}} 
\label{sec:proof_cor}

%\begin{corollary} \label{cor:2Dtime}
%	Fix $\eps \in (0,1)$. If $1/\eps$ is not an odd integer, then 2-Knapsack can be solved in time $\tOh(n^{\lceil \frac 1{2\eps} - \frac 12 \rceil}) = \tOh(n^{\lceil (\lceil \frac 1 \eps \rceil - 1 ) / 2 \rceil})$. Otherwise, it can be solved in time $\tOh(n^{\lceil \frac 1{2\eps} - \frac 12 \rceil + 1})$.
%\end{corollary}

\begin{proof}
Fix $\eps \in (0,1)$.
If $\eps \in (1/3,1)$, then $1/\eps$ is not an odd integer. By \cref{thm:linear} with $\delta \coloneqq \eps-1/3$, a $(1-\eps)$-approximation for 2-Knapsack can be computed in time $\tOh(n) = \tOh(n^{\lceil \frac 1{2\eps} - \frac 12 \rceil})$.

If $\eps \in (0,1/3)$ and $1/\eps$ is not an odd integer, then $\frac 1{2\eps} - \frac 12$ is not an integer, and thus we can choose a sufficiently small $\rho = \rho(\eps) \in (0,1)$ such that $\lceil \frac 1{2\eps} - \frac 12 + \rho \rceil = \lceil \frac 1{2\eps} - \frac 12 \rceil$. Since $\eps < 1/3$ we have $\frac 1{2\eps} - \frac 12 > 1$ and thus $\lceil \frac 1{2\eps} - \frac 12 \rceil \ge 2$. Hence, the running time of \cref{thm:dalgo} for 2-Knapsack simplifies to $\big(n^{\lceil \frac{1}{2\eps} - \frac 12 + \rho \rceil}+n^{2}\big) \cdot \big( \frac{\log n}{\rho\cdot \eps }\big)^{O(1)} = \tOh(n^{\lceil \frac 1{2\eps} - \frac 12 \rceil})$.

Using the identities $\lceil x \rceil - m = \lceil x - m \rceil$ and $\lceil \frac{\lceil x \rceil}m \rceil = \lceil \frac xm \rceil$ for any $x \in \mathbb{R}$ and $m \in \mathbb{N}$, we can alternatively rewrite the exponent as $\lceil \frac 1{2\eps} - \frac 12 \rceil = \lceil (\lceil \frac 1 \eps \rceil - 1 ) / 2 \rceil$.

If $1/\eps$ is an odd integer, then for $\rho \coloneqq 1/2$ we have $\lceil \frac 1{2\eps} - \frac 12 + \rho \rceil = \lceil \frac 1{2\eps} - \frac 12 \rceil + 1 \ge 2$, and then the running time of \cref{thm:dalgo} for 2-Knapsack is $\tOh(n^{\lceil \frac 1{2\eps} - \frac 12 \rceil + 1})$.

In all cases, the time is bounded by $\tOh(n^{\lceil \frac 1{2\eps} - \frac 12 \rceil + 1})$.
\end{proof}

\section{Inequality Vector \boldmath$k$-SUM: Proof of \Cref{lem:ksum_inequalities}}
\label{sec:ksum_ineq_proof}
% The proof of Lemma~\ref{lem:ksum_inequalities} is now in Section~\ref{sec:ksum_algorithm}.

Replacing inequality by equality constraints is a common technique in fine-grained complexity, see, e.g.,~\cite{VassilevskaWW13,NederlofLZ12,GokajK25}. Here we use the following variant of this technique.

\begin{theorem}[{\cite[Theorem~1]{NederlofLZ12}}]
	\label{thm:ksum_shrinking_intervals}
	There is a polynomial-time algorithm that, given a set $U$ of
	$n\geq1$ elements, an integer $N\geq2$, a function
	$\omega:U\to\{0,\ldots,N\}$, and integers $0\leq l\leq u\leq N$,
	computes a family of $O(n\log N)$ pairs $(\omega_q,s_q)$,
	where $\omega_q:U\to\{0,\ldots,N\}$ and
	$s_q\in\{0,\ldots,N\}$, such that, for every $X\subseteq U$,
	\[
	\sum_{i\in X}\omega(i)\in[l,u]
	\quad\Longleftrightarrow\quad
	\exists q \text{ such that }
	\sum_{i\in X}\omega_q(i)=s_q.
	\]
\end{theorem}

We are ready to prove \cref{lem:ksum_inequalities}.

\begin{proof}[Proof of~\cref{lem:ksum_inequalities}]
	%We may take $\Delta\geq2$ to be an integer. 
	We solve a given Inequality Vector $k$-SUM instance $\veca_1,\ldots,\veca_n,\vecs$ as follows.
	\begin{enumerate}[nosep]
		\item Write $a_{i,\ell}$ for
	the $\ell$-th coordinate of $\veca_i$ and $s_\ell$ for
	the $\ell$-th coordinate of
	$\vecs$. 
		\item Set $N=(k+1)\Delta$ and $M=10k^2\Delta$.
		For every $\ell\in[D]$, define
		$\omega_\ell(i)=a_{i,\ell}+\Delta$ for $i\in[n]$ and
		$v_\ell=s_\ell+k\Delta$.
		\item For every $\ell\in[D]$, apply
		\Cref{thm:ksum_shrinking_intervals} with the set $U=[n]$, the  bound $N$,
		the  function $\omega_\ell$, and the interval $[0,v_\ell]$.
		Denote the returned family by
		$\{(\omega_{\ell,q},v_{\ell,q}):q\in[r_\ell]\}$.
		\item For every tuple $(q_1,\ldots,q_D)\in
		[r_1]\times\cdots\times[r_D]$, form the $k$-SUM instance $a'_1,\ldots ,a'_n$ defined by 
		\[
		a'_i\coloneqq \sum_{\ell=1}^D\omega_{\ell,q_\ell}(i)M^{\ell-1} - \frac{s'}k
		\qquad \text{ for every $i\in[n]$, where}\qquad
		s'\coloneqq \sum_{\ell=1}^D v_{\ell,q_\ell}M^{\ell-1}.
		\]
		Run the assumed $k$-SUM algorithm on $a'_1,\ldots,a'_n$. If it returns a solution $i_1 < \ldots < i_k$, then return $X = \{i_1,\ldots,i_k\}$.
		\item If no call succeeds, report that no solution exists.
	\end{enumerate}
	
	Note that the conditions of \Cref{thm:ksum_shrinking_intervals} hold,
	since $|a_{i,\ell}|,|s_\ell|\leq\Delta$ and $k\geq1$, and thus 
	$0\leq\omega_\ell(i)\leq2\Delta\leq N$ and
	$0\leq v_\ell\leq N$. 
	We prove correctness in the following two claims.
	\begin{claim}
		\label{clm:ksum_inequalities_complete}
		If the given Inequality Vector $k$-SUM instance has a solution, the algorithm
		returns a set of $k$ indices.
	\end{claim}
	\begin{claimproof}
		Let $X\subseteq[n]$ satisfy $|X|=k$ and
		$\sum_{i\in X}\veca_i\leq\vecs$. 
		For every $\ell\in[D]$, nonnegativity of $\omega_\ell$ and the
		inequality in coordinate $\ell$ give
		\[
		0\leq\sum_{i\in X}\omega_\ell(i)
		=\sum_{i\in X}a_{i,\ell}+k\Delta
		\leq s_\ell+k\Delta=v_\ell.
		\]
		Thus, by \Cref{thm:ksum_shrinking_intervals} there is
		$q_\ell\in[r_\ell]$ such that
		$\sum_{i\in X}\omega_{\ell,q_\ell}(i)=v_{\ell,q_\ell}$.
		Choose one such index $q_\ell$ for each coordinate $\ell$. When the algorithm enumerates
		the tuple $(q_1,\ldots,q_D)$, the constructed scalar instance satisfies
		\[
		\sum_{i\in X}a'_i
		=\sum_{\ell=1}^D\left(\sum_{i\in X}
		\omega_{\ell,q_\ell}(i)\right)M^{\ell-1} - |X| \frac{s'}k
		=\sum_{\ell=1}^D v_{\ell,q_\ell}M^{\ell-1} - s'=0.
		\]
		Thus $X$ is a solution to the constructed $k$-SUM instance, so the call to the $k$-SUM algorithm
		returns a solution $i_1 < \ldots < i_k$. The algorithm therefore returns a set of $k$ indices,
		either in this iteration or in an earlier successful iteration.
	\end{claimproof}
	The following claim shows that solutions returned by the algorithm are correct. 
	\begin{claim}
		\label{clm:ksum_inequalities_sound}
		Every set of indices returned by the algorithm is a solution to the given
		Inequality Vector $k$-SUM instance.
	\end{claim}
	\begin{claimproof}
		Let $X$ be the index set returned in the iteration for
		$(q_1,\ldots,q_D)$. The scalar $k$-SUM algorithm guarantees that
		$X\subseteq[n]$, $|X|=k$, and $\sum_{i\in X}a'_i=0$, which rearranges to
		\begin{equation} \label{eq:app:vectorksum}
		\sum_{\ell=1}^D\left(\sum_{i\in X}
		\omega_{\ell,q_\ell}(i)\right)M^{\ell-1}
		=\sum_{\ell=1}^D v_{\ell,q_\ell}M^{\ell-1}.
		\end{equation}
		For every $\ell\in[D]$, the bounds on the returned weights and
		targets imply
		\[
		0\leq\sum_{i\in X}\omega_{\ell,q_\ell}(i)
		\leq kN=k(k+1)\Delta<M,\qquad
		0\leq v_{\ell,q_\ell}\leq N<M.
		\]
		Hence the identity \eqref{eq:app:vectorksum}
		has all coefficients in $\{0,\ldots,M-1\}$ on both sides.
		Uniqueness of base\nobreakdash-$M$ representation therefore gives
		$\sum_{i\in X}\omega_{\ell,q_\ell}(i)=v_{\ell,q_\ell}$
		for every $\ell\in[D]$.
		By \Cref{thm:ksum_shrinking_intervals}, these equalities
		imply $\sum_{i\in X}\omega_\ell(i)\in[0,v_\ell]$, for every $\ell \in [D]$.
		Consequently, for every $\ell\in[D]$,
		\[
		\sum_{i\in X}a_{i,\ell}
		=\Big(\sum_{i\in X}\omega_\ell(i) \Big)-k\Delta
		\leq v_\ell-k\Delta=s_\ell.
		\]
		Thus $\sum_{i\in X}\veca_i\leq\vecs$, as required.
	\end{claimproof}
	
	Together, \Cref{clm:ksum_inequalities_complete,clm:ksum_inequalities_sound}
	establish correctness: any returned set is a solution, and failure
	is reported precisely when no solution exists.
	
	By \Cref{thm:ksum_shrinking_intervals} we have $r_\ell=O(n\log N)=O_k(n\log\Delta)$.
	Thus, the algorithm solves  $\prod_{\ell=1}^D r_\ell
	=O_{k,D}((n\log\Delta)^D)$  instances of $k$-SUM, each with $n$ numbers.
	Constructing each instance uses $O(nD)$ arithmetic operations.
	Each application of \Cref{thm:ksum_shrinking_intervals} takes
	polynomial time in $n$ and $\log N$, and each  $k$-SUM instance can be solved in time
	$O_k(n^{\alpha_k})$. Therefore, the total running
	time on the RealRAM is
	$O_{k,D}(n^{\alpha_k}(n\log\Delta)^{D+O(1)})$, as claimed.
\end{proof}

\end{document}